\documentclass[%
  jmp,
 amsmath,amssymb,
 reprint
]{revtex4-1}

\usepackage{graphicx}
\usepackage{dcolumn}
\usepackage{bm}

\usepackage{xurl}

\usepackage{mathptmx}
\usepackage{etoolbox}

\usepackage{bm}
\usepackage[mathscr]{euscript}
\usepackage{amssymb,amsmath}
\usepackage{amsthm}
\usepackage{comment}
\usepackage{epstopdf}
\usepackage{color}
\usepackage[english]{babel}
\usepackage{enumerate}
\usepackage{bbold}
\usepackage{leftidx}
\usepackage{mathrsfs}
\usepackage{mathtools}
\usepackage{array}
\usepackage{hyperref}

\usepackage{tikz}
\usepackage{tikz-cd}
\usetikzlibrary{shapes.geometric,arrows.meta,positioning}  
\usetikzlibrary{matrix,calc}
\usetikzlibrary {shapes.misc,decorations.pathreplacing,topaths}
\usetikzlibrary{graphs,scopes,chains}

\usepackage{xr}
\makeatletter

\hypersetup{
           breaklinks=true,   
           colorlinks=true   
           }

\definecolor{MyOrange}{rgb}{0.8, 0.33, 0}
\definecolor{MyBrown}{rgb}{0.28, 0.20, 0.20}
\definecolor{MyBlue}{rgb}{0.12, 0.55, 1.73}

\newtheorem{theorem}{Theorem}[section]
\newtheorem{corollary}{Corollary}[theorem]
\newtheorem{lemma}[theorem]{Lemma}
\newtheorem{definition}{Definition}
\newtheorem{proposition}{Proposition}

\makeatletter
\def\@email#1#2{%
 \endgroup
 \patchcmd{\titleblock@produce}
  {\frontmatter@RRAPformat}
  {\frontmatter@RRAPformat{\produce@RRAP{*#1\href{mailto:#2}{#2}}}\frontmatter@RRAPformat}
  {}{}
}%
\makeatother
\begin{document}

\tikzset{   
            every picture/.style={ultra thick,text height=1.5ex,text depth=.25ex},
            roundnode/.style={circle, draw=MyBlue!100, fill=MyBlue!5, ultra thick, minimum size=5mm},
            specialnode/.style={circle, draw=red!100, fill=red!5, ultra thick, minimum size=5mm},
            sigmanode/.style={circle, draw=MyOrange!100, fill=MyOrange!5,  minimum size=5mm, ultra thick},
            spin/.style={circle, draw=black, inner sep=0mm, fill=black, ultra thick, minimum size=2mm},
            squarenode/.style={rectangle,minimum size=15mm,rounded corners=5mm,ultra thick,draw=black!50,top color=white,bottom color=black!20},
            ampersand replacement =\&}


\title{Decay of quantum conditional mutual information for purely generated finitely correlated states}
\author{Pavel Svetlichnyy}
 
\author{T.A.B. Kennedy}%
\affiliation{School of Physics, Georgia Institute of Technology, Atlanta, GA, 30332-0430, USA.}

\date{\today}

\begin{abstract}
The connection between quantum state recovery and quantum conditional mutual information (QCMI) is studied for the class of purely generated finitely correlated states (pgFCS) of one-dimensional quantum spin chains. For a tripartition of the chain into two subsystems separated by a buffer region, it is shown that a pgFCS is an approximate quantum Markov chain, and stronger, may be approximated by a quantum Markov chain in trace distance, with an error exponentially small in the buffer size. This implies that, (1) a locally corrupted state can be approximately recovered by action of a quantum channel on the buffer system, and (2) QCMI is exponentially small in the size of the buffer region. Bounds on the exponential decay rate of QCMI and examples of quantum recovery channels are presented.
\end{abstract}

\maketitle

\section{Introduction} 

In this paper we investigate properties of quantum conditional mutual information (QCMI) and quantum state recovery for the class of many-particle states known alternatively as purely generated finitely correlated states (pgFCS) or uniform matrix product states.\cite{fcs,pgfcs,klumper} These states were introduced in the early 1990's and used to describe the ground states of translationally invariant one-dimensional spin chains, including the well-known model of Affleck, Kennedy Lieb and Tasaki (AKLT) as a special case.\cite{aklt} Independently, the matrix product states were introduced by Klümper et al \cite{klumper} and separately shown by Vidal \cite{vidal-mps} to be an efficient representation of slightly entangled states, such as the ground states of local gapped Hamiltonians.\cite{mps-repr} Proposals for the controlled experimental preparation of such states have been reported in the literature.\cite{sequent-gen,lukin} The larger class of finitely correlated states (FCS) and the closely related matrix product density operators (MPDO),\cite{mpdo-introd} can be used to describe mixed states of systems with local Hamiltonians and finite range interactions, such as finite temperature Gibbs states.\cite{kuwahara,kato}

Positivity of QCMI is one of the cornerstones of quantum information theory, being equivalent to the strong subadditivity of quantum entropy.\cite{wilde,lieb-subadd} In the theory of quantum state recovery,\cite{petz-1,petz-2,petz-recovery,universal-rec-1,universal-rec-2,universal-rec-3,storage} viewed as a generalized quantum error correction problem, the states for which recovery is exact are known as quantum Markov chains. These states are intimately connected with the QCMI as is clear in the following setting. Consider a quantum system (a collection of spatially separated spins) subdivided into three parts, denoted $A$, $B$ and $C,$ see Figure 1a. The QCMI, denoted $I(A:C|B)$ indicates in broad terms the quantum information mutually shared by subsystems $A$ and $C$ given specific information about the state of subsystem $B$. Two essential and non-trivial properties of $I(A:C|B)$ are that it is non-negative $I(A:C|B) \geq 0,$ as a direct consequence of the strong subadditivity of quantum entropy, \cite{lieb-subadd} and that the equality condition $I(A:C|B) = 0$
defines the class of quantum states with density operator $\rho_{ABC}$, which are precisely the quantum Markov chains. In this setting the problem of quantum state recovery, in either exact or approximate form, may be approached by investigating when $I(A:C|B)$ is zero or very small, conditions which require investigation of quantum states described as either exact or approximate quantum Markov chains. This specific quantum information-theoretic task gives QCMI a compelling physical interpretation.

We study quantum state recovery for a one-dimensional quantum system $ABC$ for which the density operator $\rho_{ABC}$ is the reduced density operator of a pgFCS, respectively, uniform matrix product state. From here on we will principally use the former appellation and refer the reader to the dictionary between these two languages compiled in Appendix \ref{app-dictionary}. As we discuss further below, $\rho_{ABC}$ is an approximate quantum Markov chain and hence can be approximately recovered. 

The pgFCS is a subset of the class of FCS, some of whose properties we review in Section \ref{subsec-fcs-itrod}. The FCS can efficiently describe systems with limited entanglement and their preparation by using quantum circuits of limited depth.\cite{kato,eff-prep} It was recently shown that FCS may be used to approximate Gibbs states of systems described by one-dimensional local Hamiltonians.\cite{kato, kuwahara} (In Refs. \onlinecite{efficient-mpdo,improved-thermal} it was shown that the same class of states can be approximated by matrix product operators (MPO). It is not immediately clear that the resulting MPO are FCS or MPDO. The latter are called \emph{locally purifiable} in the sense of Ref. \onlinecite{mpdo-introd}.) Conversely, it is conjectured that a generic FCS approximates the Gibbs state of some local gapped Hamiltonian.\cite{fcs-mpdo} While Gibbs states are of broad physical interest, their sampling has been employed in a quantum algorithm proposed to speed up the computation of semi-definite programs.\cite{q-sdp-1,q-sdp-2,q-sdp-3}

The QCMI is defined for a system $ABC$ with the Hilbert space $\mathcal{H}=\mathcal{H}_A\otimes\mathcal{H}_B\otimes\mathcal{H}_C$ by \cite{watrouse}
\begin{equation}\label{qcmi-def}
I(A:C|B) := S(\rho_{AB})+S(\rho_{BC})-S(\rho_{ABC})-S(\rho_{B}),
\end{equation} 
where $S(\rho)=-\mathrm{tr}(\rho\log\rho)$ is the von Neumann entropy, and omission of a subscript implies taking a partial trace, e.g., $\rho_{AB}=\mbox{tr}_C\rho_{ABC}$. As noted earlier, QCMI is non-negative, $I(A:C|B)\geq 0,$ as an immediate consequence of the strong subadditivity of von-Neumann entropy. \cite{lieb-subadd} In Refs. \onlinecite{petz-1, petz-2} it was established that the condition $I(A:C|B)=0$ ($\rho_{ABC}$ is a quantum Markov chain) is equivalent to the existence of a quantum channel, a completely positive trace-preserving map, \cite{stinespring, watrouse} $\mathcal{R}_{B\rightarrow BC}$, which recovers $\rho_{ABC}$ from $\rho_{AB}$ exactly, i.e., $\mathcal{R}_{B\rightarrow BC}(\rho_{AB})=\rho_{ABC}$. From this property follows the particular structure of quantum Markov chains which we summarize in Theorem \ref{qmc-structure}. \cite{petz-recovery} The recovery map is not unique, and we can always explicitly construct the Petz recovery map, which is defined (on the support of $\rho_{B}$) by \cite{petz-1, petz-2, petz-recovery}
\begin{equation}\label{petz-map}
    \mathcal{P}_{B\rightarrow BC}(X):=\rho_{BC}^{\frac{1}{2}}(\rho_B^{-\frac{1}{2}}X\rho_B^{-\frac{1}{2}}\otimes\mathbb{1}_C)\rho_{BC}^{\frac{1}{2}},
\end{equation}
where $\rho_B^{-\frac{1}{2}}$ is a pseudo-inverse of $\rho_B^{\frac{1}{2}}$. In the case when a recovery quantum channel does not restore the state exactly, the \emph{recovery error} is quantified by the trace distance, 
\begin{equation*}
    \epsilon:=\|\rho_{ABC}-\mathcal{R}_{B\rightarrow BC}(\rho_{AB})\|_1, 
\end{equation*}
where $\|X\|_1=\mathrm{tr}(\sqrt{X^\dagger X})$. The recovery error, for a candidate recovery channel $\mathcal{R}_{B\rightarrow BC}$, bounds QCMI from above via an Alicki-Fannes type inequality \cite{alicki-fannes, tight-bounds} (see Appendix \ref{app-bound}),
\begin{align}\label{qcmi-rec}
I(A:C|B)&\leq \frac{1}{2}\epsilon \ln \dim\mathcal{H}_A+(1+\frac{1}{2}\epsilon)h\left(\frac{\epsilon}{2+\epsilon}\right), 
\end{align}
where $h(p)=-p\ln p-(1-p)\ln(1-p)$ is the binary entropy. This bound was used to show that a Gibbs state of a 1D local Hamiltonian is an approximate quantum Markov chain.\cite{kato} By constructing a specific recovery channel, for which the recovery error is subexponentially small in the size of the region $B$, it was shown that the QCMI is also subexponentially small. In this paper we similarly construct an approximate recovery channel for pgFCS, and bound QCMI from above using (\ref{qcmi-rec}).

Theoretical developments regarding the reversibility of quantum channels obtained in a series of works Refs. \onlinecite{universal-rec-1, universal-rec-2, universal-rec-3}, led to the discovery of a recovery channel $\mathcal{R}^u_{B\rightarrow BC}$, called the universal recovery map, which may be used to bound QCMI from below,
\begin{equation}\label{rec-qcmi}
    \frac{1}{4\ln2}\|\rho_{ABC}-\mathcal{R}^u_{B\rightarrow BC}(\rho_{AB})\|_1^2\leq I(A:C|B).
\end{equation}
The universal recovery map, which is a generalization of the Petz recovery map (\ref{petz-map}), is given explicitly in Ref. \onlinecite{universal-rec-3} and, as in the case of the Petz recovery map, depends only on the marginal $\rho_{BC}$. The inequality (\ref{rec-qcmi}) implies that states with small QCMI are guaranteed to have a small recovery error with respect to the universal recovery map. In Ref. \onlinecite{kuwahara}, it was shown for local Hamiltonian systems of arbitrary dimension, and high enough temperature, that QCMI for Gibbs states decays exponentially with the width of region $B$, separating regions $A$ and $C$. By (\ref{rec-qcmi}) this implies good recovery by the universal recovery map. (It is conjectured, \cite{kato} that in one dimension the exponential decay of QCMI holds for any temperature.) Similar conclusions were found for finite temperature Gibbs states of free fermions, free bosons, conformal field theories, and holographic models.\cite{swingle}

\begin{figure*}[t]
    \centering
    \includegraphics[width=0.95\textwidth]{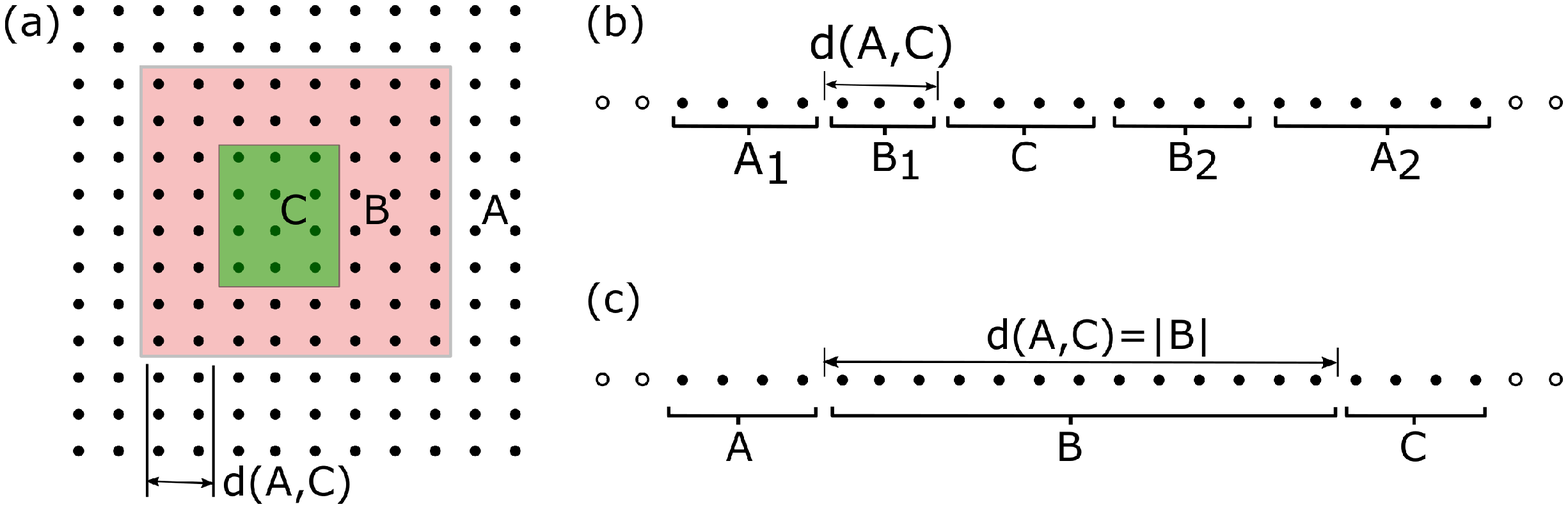}
    \caption{The figure shows lattices of spins partitioned into subsystems $A$, $B$, and $C$, which may be analyzed with the help of QCMI. In (a) we show a typical partition of a two dimensional spin lattice with the distance $d(A,C)$ shown. In (b) and (c) subsystem $B$ is disconnected and connected, respectively, and the distance $d(A,C)$ is again shown. In this paper we focus on the partition (c).}
    \label{fig1}
\end{figure*}

A more physical picture of the recovery map is as an experimental repair process, with an accuracy bounded by how rapidly QCMI decreases as a function of the distance $d(A,C)$ shown in Figure \ref{fig1}.\cite{storage}
For example, for the lattice of spins shown in Figure \ref{fig1}(a), suppose that quantum state corruption or erasure has happened in the region $C$. It is then natural to partition the lattice by means of a buffer zone $B$ surrounding $C$.\cite{swingle, eff-prep} Qualitatively, quantum states with rapidly decaying QCMI may be recovered accurately on erased regions with small buffer zones $B$.\cite{swingle, eff-prep, kato} In one dimension the analysis can be reduced to the partition shown in Figure \ref{fig1}(c),\cite{kato} which we focus on in this paper. We briefly consider the setting of Figure \ref{fig1}(b) in Appendix \ref{disk-setup}.

In the language of quantum memories, QCMI is connected to the task of storing quantum information in the presence of quantum error correction,\cite{storage} and to state preparation by (circuits of) local quantum channels.\cite{eff-prep, swingle, kato}  Let us assume that quantum channels can be applied experimentally over regions with size limited by the linear dimension $l$. Conceptually, a quantum state may be prepared on a much larger region by first constructing reduced states on a grid of disconnected regions and then patching them together by the application of a second layer of quantum channels. The deviation of the constructed state from the target state will depend on the value of QCMI for the partitions with $d(A,C) = O(l)$. An example of Gibbs state construction by the application of layers of universal recovery channels is given in Ref. \onlinecite{eff-prep}. A simpler, one-dimensional example, in which the circuit consists of two layers of quantum channels, is presented in Ref. \onlinecite{kato}. Note that in the latter reference the quantum channels proposed are not universal recovery channels. In this work we also avoid use of the universal recovery maps, relying instead on the explicit structure of pgFCS.

Our main result, a bound on the recovery error and the related decay of QCMI, is summarized in Theorem \ref{main-result} below. The latter can be compared to and contrasted with the behavior of quantum mutual information (QMI), defined by $I(A:C)=S(\rho_{A})+S(\rho_{C})-S(\rho_{AC})$. The QMI has the property that it bounds the quantum correlations between the regions $A$ and $C$ via the quantum Pinsker's inequality.\cite{quant-pinsker} Those pgFCS which correspond to injective MPS\cite{mps-repr} exhibit exponential decay of both QMI\cite{area-laws} and correlations.\cite{fcs, exp-decay-area-law} A theoretical bound on QMI can be obtained more straightforwardly than a similar bound on QCMI, since in the former case the separating region $B$ is traced out and information about its state is lost. We note that for pgFCS corresponding to non-injective MPS QMI does not necessarily converge to zero with the growth of region $B$ (see Appendix \ref{qmi-examples} for examples). By contrast, we show that for any pgFCS the QCMI converges to zero.   

The remainder of this paper is organized around the proof of 
Theorem \ref{main-result}, which we state in Section \ref{sec-main}. In Section \ref{sec-prelim} we establish our notation and conventions, and provide relevant background and theorems on FCS, pgFCS, quantum channels and quantum Markov chains, which we use throughout the paper. In Section \ref{sec-proof} we prove Theorem \ref{main-result} starting with the simplest case of pgFCS (corresponding to the injective MPS), and proceed to more general cases. While the proof does not require consideration of specific recovery channels, we give examples of the latter in Section \ref{sec-rec-maps}. In Sections \ref{sec-prelim},  \ref{sec-proof}, and \ref{sec-rec-maps} we rely on additional technical results, presented in a sequence of appendices. We summarize our conclusions in Section \ref{sec-concl}.

\section{Main result}\label{sec-main}
In this paper we consider a pgFCS density operator, denoted by $\rho_{ABC},$ for the tripartite system $ABC$ depicted in Figure \ref{fig1}(c). We construct a QMC approximating $\rho_{ABC}$ and a recovery map for which both the trace distance error and the recovery error are exponentially small in the size of the separating region $B$. Using inequality (\ref{qcmi-rec}), we conclude that QCMI decays exponentially. The methods we use are built upon those of Refs. \onlinecite{fcs, mps-repr, area-laws, exp-decay-area-law}, which we augment with the continuity theorem for Stinespring's dilation,\cite{infdisturb} cited in Section \ref{subsec-quant-chan} as Theorem \ref{inf-dist-theorem}. 

We now state our principle result,
\begin{theorem}\label{main-result}
Let $\rho_{ABC}$ be the reduced density operator for a pgFCS on a finite contiguous one-dimensional region $ABC$, where the subsystem $B$ separates the subsystems $A$ and $C$. Then, provided that the region $B$ is large enough,
\begin{enumerate}
    \item
    There exists a QMC, denoted $\tilde{\rho}_{ABC}$, and constants $q>0$, $K>0$, such that
    $$\|\rho_{ABC}-\tilde{\rho}_{ABC}\|_1\leq \frac{1}{2}Ke^{-q|B|}.$$
    \item 
    There exists a quantum channel $\mathcal{R}_{B\rightarrow BC}:\mathcal{B}(\mathcal{H}_B)\rightarrow\mathcal{B}(\mathcal{H}_B\otimes \mathcal{H}_C)$, such that $$\|\rho_{ABC}-(\mathrm{id}_A\otimes \mathcal{R}_{B\rightarrow BC}) \left( \rho_{AB} \right)\|_1\leq Ke^{-q|B|}.$$ 
    \item 
    Furthermore, the QCMI satisfies $I(A:C|B)\leq \tilde{K}e^{-q|B|}$, where $\tilde{K}$ depends linearly on the sizes of the regions $A$ and $B$.  
\end{enumerate}
\end{theorem}
\noindent \textbf{Remarks}:
\begin{enumerate}
    \item We will show that $q=-\frac{1}{2}\log\nu$, where $\nu_{\mathrm{gap}}<\nu<1$ with $\nu_{\mathrm{gap}}:=\max_{|\nu_i|<1}\{|\nu_i|\}$, where $\nu_i$ are the eigenvalues of the quantum channel $\mathcal{E}$ induced by the pgFCS under consideration (see Section \ref{subsec-fcs-itrod}). The value of $\nu$ may be chosen arbitrarily close to $\nu_{\mathrm{gap}}$, with compensating increase in $n_0$, defined in the Remark 2(b) immediately below. See Lemma \ref{b-2} proved in Appendix \ref{app-b} for details.

    \item Let $d_s$ be the dimension of the Hilbert space $\mathcal{H}_s$ of a single spin. Our proof requires that size of the region $B$, denoted by $|B|$, is large enough that $\dim\mathcal{H}_{B}=d_s^{|B|}\geq d_M^2$, where $d_M$ is the dimension of the \emph{memory space}, defined for FCS in Section \ref{subsec-fcs-itrod}. This demand is quite mild, since $\dim\mathcal{H}_B$ grows exponentially with $|B|$.

    \item We may write a bound that does not require a choice of the parameter $\nu$, however this bound will have a logarithmic correction in the exponent,
    $$
        \|\rho_{ABC}-(\mathrm{id}_A\otimes \mathcal{R}_{B\rightarrow BC}) \left( \rho_{AB} \right)\|_1\leq K'e^{-q'|B|+u'\log|B|},
    $$
    where $q'=-\frac{1}{2}\log\nu_{\mathrm{gap}}$, and $K'$ and $u'$ are constants. See Lemma \ref{b-2-app} for details. A similar modification may be done to the bound on QCMI.
    
    \item For the setup of Figure \ref{fig1}(b) we have a result similar to Theorem \ref{main-result} with an additional factor of $2$ in the pre-exponent and $d(A,C)=\min\{|B_1|,|B_2|\}$ replacing $|B|$. In this case we have to demand that both $B_1$ and $B_2$ are sufficiently large, i.e., $\dim\mathcal{H}_{B_1},\dim\mathcal{H}_{B_2}\geq d_M^2$. We prove this result in Appendix \ref{disk-setup}.
    
    \item In Theorem \ref{main-result}, the points $2$ and $3$ are consequences of the stronger statement $1$: there exist approximate QMC's, which cannot be approximated by (exact) QMC's. \cite{book-sutter,Christandl2012,Ibinson2008} 
\end{enumerate}

\section{Preliminaries}\label{sec-prelim}

In this section we take the opportunity to introduce some
relevant background on quantum channels, finitely correlated states and quantum Markov chains, and to compile and restate several important results from the literature in our notation. These are needed for the proof of Theorem \ref{main-result} in Section \ref{sec-proof}.

\subsection{Notation and conventions}\label{subsec-prelim}
All Hilbert spaces considered in this paper are finite-dimensional and are denoted, up to subscripts, by either $\mathcal{H}$ or $\mathcal{K}$. In particular, we denote the Hilbert space of a single spin by $\mathcal{H}_s$, and denote its dimension by $d_s:=\dim\mathcal{H}_s$. We define $\mathcal{B}(\mathcal{H})$ to be the space of linear operators acting on $\mathcal{H}$; the space $\mathcal{B}(\mathcal{H})$ is isomorphic to the space of $\dim\mathcal{H}\times\dim\mathcal{H}$ complex matrices. The set of density operators on the Hilbert space $\mathcal{H}$ is denoted by $\mathcal{D}(\mathcal{H})$. The $n$-fold composition of a quantum channel, for instance $\mathcal{E}$, is denoted by $\mathcal{E}^n$. 

To lighten the notation we often omit the identity operator/map in tensor products of operators/maps. For example, for a quantum channel $\Phi:\mathcal{B}(\mathcal{H})\rightarrow \mathcal{B}(\mathcal{H}_s)\otimes \mathcal{B}(\mathcal{H})$, the quantum channel $\Phi^2\equiv\Phi\circ\Phi: \mathcal{B}(\mathcal{H})\rightarrow\mathcal{B}(\mathcal{H}_s)\otimes\mathcal{B}(\mathcal{H}_s)\otimes\mathcal{B}(\mathcal{H})$ denotes the map $X\mapsto(\mathrm{id}_s\otimes\Phi)(\Phi(X))$, where $\mathrm{id}_s$ is the identity map on $\mathcal{B}(\mathcal{H}_s)$. Analogously,  if $U:\mathcal{H}\rightarrow\mathcal{K}_2\otimes\mathcal{H}$ and $W:\mathcal{H}\rightarrow\mathcal{K}_1\otimes\mathcal{H}$ are isometries, then $U W:\mathcal{H}\rightarrow\mathcal{K}_1\otimes\mathcal{K}_2\otimes\mathcal{H}$ denotes the isometry $(\mathbb{1}_{\mathcal{K}_1}\otimes U)W$, which we can also express in terms of an orthonormal basis of $\mathcal H$ as
\begin{equation}\label{prod-isom-def}
    UW:=\sum_{i,j,k=1}^{\dim\mathcal{H}}  \langle k|W|j\rangle \otimes \langle i|U|k\rangle\otimes  |i\rangle\langle j|.
\end{equation}
For system $ABC$ the notation $\rho_{AB}$ means the partial trace of $\rho_{ABC}$ over the Hilbert space of system $C$, i.e., $\rho_{AB}:=\mbox{tr}_C\rho_{ABC}$. We denote the number of spins in the region $A$ by $|A|$. The support of an operator $O:\mathcal{K}\rightarrow\mathcal{H}$, i.e., the closure of the set $\left\{|\psi\rangle\in\mathcal{K}\;|\;O|\psi\rangle\neq 0\right\}$, is denoted by $\mathrm{supp}(O)$. 

\subsection{Quantum channels}\label{subsec-quant-chan}
A quantum channel $\mathcal{M}$ is a completely positive trace-preserving map $\mathcal{M}:\mathcal{B}(\mathcal{K}_{in})\rightarrow\mathcal{B}(\mathcal{K}_{out})$. The complete positivity means that for operators in the space $\mathcal{B}(\mathcal{K}_R\otimes\mathcal{K}_{in})$, where $\mathcal{K}_R$ is a reference Hilbert space of arbitrary dimension, the map $\mathrm{id}_R\otimes\mathcal{M}$ is positive, i.e., for any $\dim\mathcal{K}_R\in\mathbb{N}$ and for any $\mathcal{B}(\mathcal{K}_R\otimes\mathcal{K}_{in})\ni X\geq 0$, $(\mathrm{id}_R\otimes\mathcal{M})(X)\geq 0$. For finite-dimensional $\mathcal{K}_{in}$, as considered here, it suffices to demand that $\mathrm{id}_R\otimes\mathcal{M}$ is positive for the case $\dim\mathcal{K}_R=\dim\mathcal{K}_{in}$. The trace-preservation implies that $\mathrm{tr}\mathcal{M}(X)=\mathrm{tr}X$ for any $X\in\mathcal{B}(\mathcal{K}_{in})$. For $p,q\in[1,\infty)$ we define the $p\rightarrow q$ norm $\|\mathcal{M}\|_{p\rightarrow q}:=\sup_{X\in\mathcal{B}(\mathcal{K}_{in})}{\|\mathcal{M}(X)\|_q}/{\|X\|_p}$, where $\|\cdot\|_r$ is the Schatten $r$-norm. \cite{watrouse} 

By Stinespring's dilation theorem, \cite{stinespring} in the finite-dimensional setting, a quantum channel $\mathcal{M}:\mathcal{B}(\mathcal{K}_{in})\rightarrow\mathcal{B}(\mathcal{K}_{out})$ may be represented in the form \cite{watrouse}
\begin{equation}\label{quant-ch}
    \mathcal{M}(X)=\mathrm{tr}_E\left(WXW^\dagger\right),
\end{equation}
where $W:\mathcal{K}_{in}\rightarrow\mathcal{K}_E\otimes\mathcal{K}_{out}$ is an isometry, $W^\dagger W=\mathbb{1}_{in}$, called the dilating isometry, and the Hilbert space $\mathcal{K}_E$ is called the dilation space (or environment). The pair $\{\mathcal{K}_E,W\}$ is referred to as a dilation of $\mathcal{M}$. For brevity, we will write that $W$ dilates the quantum channel $\mathcal{M}$. We note that $\dim\mathcal{K}_E\leq\dim\mathcal{K}_{in}\dim\mathcal{K}_{out}$.\cite{watrouse} The isometric representation is not unique,\cite{infdisturb} and, if $W:\mathcal{K}_{in}\rightarrow\mathcal{K}_E\otimes\mathcal{K}_{out}$ and $W':\mathcal{K}_{in}\rightarrow\mathcal{K}_{E'}\otimes\mathcal{K}_{out}$ are different dilations with $\dim\mathcal{K}_{E'}\geq\dim\mathcal{K}_{E}$, then $W'=UW$ for some isometry $U:\mathcal{K}_{E}\rightarrow\mathcal{K}_{E'}$, $U^\dagger U=\mathbb{1}_E$. In the case of $\dim\mathcal{K}_E=\dim\mathcal{K}_{E'}$, the isometry $U$ is unitary.

The space of all quantum channels $\mathcal{M}:\mathcal{B}(\mathcal{K}_{in})\rightarrow\mathcal{B}(\mathcal{K}_{out})$ with fixed domain and codomain can be endowed with the so-called \emph{diamond norm},\cite{kitaev} also called \emph{completely bounded norm}, $\|\mathcal{M}\|_\diamond:=\sup_{\dim{\mathcal{K}_R}}\|\mathrm{id}_R\otimes\mathcal{M}\|_{1\rightarrow 1}$. In the finite-dimensional case the supremum is achieved for $\dim\mathcal{K}_R=\dim\mathcal{K}_{in}$.\cite{kitaev} The quantum channels are continuous with respect to the diamond norm, a result that is expressed by the Continuity Theorem for Stinespring's Representation,
\begin{theorem}[Ref.\,\onlinecite{infdisturb}]\label{inf-dist-theorem}
    Let $\mathcal{K}_{in}$ and $\mathcal{K}_{out}$ be finite-dimensional Hilbert spaces, and suppose that $\mathcal{M}_1,\mathcal{M}_2:\mathcal{B}(\mathcal{K}_{in})\rightarrow\mathcal{B}(\mathcal{K}_{out})$ are quantum channels with Stinespring isometries $W_1, W_2: \mathcal{K}_{in} \rightarrow \mathcal{K}_E\otimes\mathcal{K}_{out}$ and a common dilation space $\mathcal{K}_E$. We then have
    \begin{equation}\label{inf-ineq}
        \inf_{U_E}\|W_1-U_EW_2\|^2\leq \|\mathcal{M}_1-\mathcal{M}_2\|_\diamond\leq 2\inf_{U_E}\|W_1-U_EW_2\|,
    \end{equation}
    where the minimization is with respect to all unitary $U_E\in\mathcal{B}(\mathcal{K}_E)$.
\end{theorem}
\noindent We denoted by $\|W\|$ the operator norm of $W$. In the case $\mathcal{M}_1=\mathcal{M}_2$, we obtain $W_1=U_EW_2$ for some unitary $U_E$, which implies that the dilating isometry is defined up to a unitary on the dilation space, consistent with (\ref{quant-ch}).

For a quantum channel $\mathcal{M}$, regarded as a linear map on the $\dim\mathcal{K}_{in}^2$-dimensional vector space $\mathcal{B}(\mathcal{K}_{in})$, we define the spectrum in the usual way: $X \in\mathcal{B}(\mathcal{K}_{in})$ is an eigenvector with the eigenvalue $\nu$, if $X$ is non-zero and $\mathcal{M}(X)=\nu X$. Each eigenvalue $\nu_i$ ($i=1,\cdot\cdot\cdot ,\dim^2\mathcal{K}_{in}$) of $\mathcal{M}$ satisfies $|\nu_i|\leq 1$, and there always exists at least one fixed point with the eigenvalue equal to $1$.\cite{irred-pos-maps} Hence the spectral radius of a quantum channel is equal to $1$. The set of eigenvalues with $|\nu_i|=1$ is known as the \emph{peripheral spectrum}.\cite{fcs} When the peripheral spectrum is a singleton set, consisting of the single eigenvalue $1$, it is said to be trivial.   

We now employ the spectral decomposition to construct a quantum channel $\tilde{\mathcal{E}}$ from a given quantum channel $\mathcal{E}$.\cite{wolf-szehr} This will be useful because that the $n$-fold composition $\tilde{\mathcal{E}}^n$ is an excellent approximation to $\mathcal{E}^n$ in the limit $n \rightarrow \infty$.\cite{wolf-szehr}  A dilating isometry associated to $\tilde{\mathcal{E}}^n$ possesses convenient properties, and we will take advantage of these. 

As shown in Appendix \ref{app-vec}, the space of all maps $\mathcal{B}(\mathcal{K})\rightarrow\mathcal{B}(\mathcal{K})$ is isomorphic to $\mathcal{B}(\mathcal{K}\otimes\mathcal{K})$.

\begin{definition}\label{e-tilde-e}
Let $\mathcal{K}$ be a finite-dimensional Hilbert space, and $\mathcal{E}:\mathcal{B}(\mathcal{K})\rightarrow\mathcal{B}(\mathcal{K})$ be a map. Let $E\in\mathcal{B}(\mathcal{K}\otimes\mathcal{K})$ be the operator isomorphic to $\mathcal{E}$, $E:=\mathrm{vec}(\mathcal{E})$. Let $E$ have the Jordan decomposition $E=\sum_{i}(\nu_iP_{\nu_i}+N_{\nu_i})$, with the $P_{\nu_i}$ orthogonal projectors onto the subspace of the eigenvalue $\nu_i$, and the $N_{\nu_i}$ nilpotent operators. Define $P_{\mathcal{E}}:=\sum_{|\nu_i|=1}P_{\nu_i}$, the projector onto the peripheral spectrum of $\mathcal{E}$, and $\tilde{E}:=P_{\mathcal{E}}E=\sum_{|\nu_i|=1}(\nu_iP_{\nu_i}+N_{\nu_i})$. Then we define the map $\tilde{\mathcal{E}}:=\mathrm{vec^{-1}}(\tilde{E})=\mathrm{vec^{-1}}(P_{\mathcal{E}}\cdot\mathrm{vec}(\mathcal{E}))$. 
\end{definition}
For $\mathcal{E}$ and $\tilde{\mathcal{E}}$ as given in Definition \ref{e-tilde-e} we have,
\begin{lemma}[Ref. \onlinecite{wolf-szehr}]
If $\mathcal{E}$ is a quantum channel, then $\tilde{\mathcal{E}}$ is a quantum channel.
\end{lemma}
\begin{proof}
One may argue (see the proof of Proposition 3.3 in Ref. \onlinecite{fcs}, for example) that the nilpotents $N_{\nu_i}=0$ for the peripheral eigenvalues $\nu_i$, satisfying $|\nu_i|=1$. Moreover, the peripheral eigenvalues form a cyclic group under multiplication, hence there exists $m\in\mathbb{N}$, such that $\tilde{E}^m=\tilde{E}$, and $\tilde{\mathcal{E}}^m=\tilde{\mathcal{E}}$. Then for any $n\in\mathbb{N}$, $\tilde{\mathcal{E}}^{n(m-1)+1}=\tilde{\mathcal{E}}$. As Lemma \ref{b-2} below assures, the sequence $\{\mathcal{E}^{n}-\tilde{\mathcal{E}}^{n}\}_{n\in\mathbb{N}}$ converges to zero in 2-2 norm (and in diamond norm, since all the spaces are finite-dimensional), hence the subsequence $\{\mathcal{E}^{n(m-1)+1}-\tilde{\mathcal{E}}^{n(m-1)+1}\}_{n\in\mathbb{N}}$ converges to zero as well. Then $\{\mathcal{E}^{n(m-1)+1}\}_{n\in\mathbb{N}}$ converges to $\tilde{\mathcal{E}}$. For any $X\in\mathcal{B}(\mathcal{K})$,  $\mathrm{tr}(\mathcal{E}^{n(m-1)+1}(X))=\mathrm{tr}(X)$, hence $\mathrm{tr}\tilde{\mathcal{E}}(X)=\mathrm{tr}(X)$, thus $\tilde{\mathcal{E}}$ is trace-preserving. To show that $\tilde{\mathcal{E}}$ is completely positive, it suffices to prove that for any $X \geq 0$ in $\mathcal{B}(\mathcal{K}\otimes\mathcal{K})$, the condition $(\mathrm{id}\otimes\tilde{\mathcal{E}})(X)\geq 0$ is satisfied. Select $X \geq 0$ in $\mathcal{B}(\mathcal{K}\otimes\mathcal{K})$, and note that $(\mathrm{id}\otimes\mathcal{E}^{n(m-1)+1})(X)\geq 0$. The sequence $\{(\mathrm{id}\otimes\mathcal{E}^{n(m-1)+1})(X)\}_{n\in\mathbb{N}}$ converges to $(\mathrm{id}\otimes\tilde{\mathcal{E}})(X)$. Since the set of positive semidefinite operators on a finite-dimensional Hilbert space is closed, then $(\mathrm{id}\otimes\tilde{\mathcal{E}})(X)=\lim_{n\rightarrow \infty}(\mathrm{id}\otimes\mathcal{E}^{n(m-1)+1})(X)\geq 0$, completing the proof.
\end{proof}
The following lemma gathers results on convergence from Refs. \onlinecite{fcs,wolf-szehr}. For the sake of completeness, we provide its proof in our notation in Appendix \ref{app-b}.
\begin{lemma}[Theorem III.2 of Ref.\,\onlinecite{wolf-szehr}; Ref.\,\onlinecite{fcs}]\label{b-2}
Let $\mathcal{E}:\mathcal{B}(\mathcal{K})\rightarrow\mathcal{B}(\mathcal{K})$ be a map with the spectral radius $1$, and let $\tilde{\mathcal{E}}$ be the map obtained from $\mathcal{E}$ as described in Definition \ref{e-tilde-e}. Then for any $\nu\in\mathbb{R}$, such that $\nu_{\mathrm{gap}}<\nu<1$, with $\nu_{\mathrm{gap}}$ defined in Theorem \ref{main-result}, there exists the constant $c>0$, depending on $\nu$, such that
\begin{equation*}
    \|\mathcal{E}^n-\tilde{\mathcal{E}}^n\|_{2\rightarrow 2}\leq c\nu^n.
\end{equation*}
\end{lemma}

\subsection{Finitely correlated states}\label{subsec-fcs-itrod}
FCS are a special class of translationally invariant quantum states on a chain of identical finite dimensional quantum systems (spins). The structure of FCS was characterized in Ref. \onlinecite{fcs}, and is summarized here in the form convenient for our purposes. We refer the reader to the original papers Refs. \onlinecite{fcs, pgfcs} for the unabridged treatment of FCS.

A FCS can be described by the pair of a full-rank density operator $\sigma\in\mathcal{B}(\mathcal{H}_M)$ and a quantum channel $\Phi:\mathcal{B}(\mathcal{H}_M)\rightarrow\mathcal{B}(\mathcal{H}_s)\otimes\mathcal{B}(\mathcal{H}_M)$, satisfying a compatibility condition $\mathrm{tr}_s\Phi(\sigma)=\sigma$. \cite{fcs} Here $\mathcal{H}_M$ and $\mathcal{H}_s$ are Hilbert spaces with dimensions $d_M$ and $d_s$, respectively. The space $\mathcal{H}_M$ is referred to as the \emph{memory space}, and $\mathcal{H}_s$ is the Hilbert space of a single spin. The FCS reduced density operator for a continuous region $R$ with $|R|$ spins, is generated by $\Phi$ and $\sigma$ as 
\begin{equation}\label{fcs-r}
    \rho_R=\mbox{tr}_M(\Phi^{|R|}(\sigma)).
\end{equation}
Here $\mathrm{tr}_M$ is a partial trace over the Hilbert space $\mathcal{H}_M$; it should not be confused with the trace over the spins in the chain outside the region $R$. Colloquially, each $\Phi$ generates a single spin, and a composition of $\Phi$ generates consecutive spins in the chain. Note that the pair of $\Phi$ and $\sigma$ generating $\rho_R$ is not necessarily unique, and using the results of Refs. \onlinecite{fcs, pgfcs} we may choose the most convenient representation for our purposes. In this subsection we list all the relevant properties of the representation we choose. 
 
For the FCS $(\Phi,\sigma)$, we define the induced quantum channel $\mathcal{E}:\mathcal{B}(\mathcal{H}_M)\rightarrow\mathcal{B}(\mathcal{H}_M)$ by \begin{equation}\label{E-gen-fcs}
    \mathcal{E}(X)=\mathrm{tr}_s\Phi(X),
\end{equation}
for which $\sigma$ is a fixed point, $\mathcal{E}(\sigma)=\sigma$.

\subsubsection{Ergodic FCS} 
We consider a subcollection of FCS, called \emph{ergodic} FCS. By definition, an ergodic state is extremal in the convex set of translationally invariant states. However, in the context of this paper, two other equivalent definitions (see Proposition 3.1 of Ref. \onlinecite{fcs} and Lemma 4.1 of Ref. \onlinecite{irred-pos-maps}) will be more useful: (i) the state for which the eigenvalue $1$ of $\mathcal{E}$ is non-degenerate; (ii) the state for which $\mathcal{E}$ is \emph{irreducible} in the sense that there does not exist a non-trivial projector $\tilde{\Pi}$, such that $\tilde{\Pi}\mathcal{B}(\mathcal{H}_M)\tilde{\Pi}$ is invariant under $\mathcal{E}$. The importance of ergodic states lies in the fact that any FCS can be decomposed into a convex sum of ergodic FCS, which are also referred to as \emph{ergodic components} (Corollary 3.2 of Ref. \onlinecite{fcs}). We collect together various observations that lead to this conclusion in the proposition below. The item 1 follows from the Theorem 3.1 of Ref. \onlinecite{irred-pos-maps} applied to $\mathcal{E}$ and $\sigma$. The items 2 and 3 are obtained using the same reasoning as in the proof of Proposition 3.3 of Ref. \onlinecite{fcs}. The item 4 follows from Lemma 4.1 of Ref. \onlinecite{irred-pos-maps} and item 2.
\begin{proposition}\label{fcs-erg-decomp}
    Without loss of generality, we may assume that a FCS is generated by a pair $(\Phi,\sigma)$, with the induced quantum channel $\mathcal{E}$, for which the following properties hold
    \begin{enumerate}
        \item There exist $J$ orthogonal projectors $\tilde{\Pi}_\alpha:\mathcal{H}_M\rightarrow\mathcal{H}_M$, $\sum_{\alpha=1}^{J}\tilde{\Pi}_\alpha=\mathbb{1}_M$, such that $\mathcal{E}(\tilde{\Pi}_\alpha\mathcal{B}(\mathcal{H}_M)\tilde{\Pi}_\alpha)\subseteq\tilde{\Pi}_\alpha\mathcal{B}(\mathcal{H}_M)\tilde{\Pi}_\alpha$. The restriction of $\mathcal{E}$ to $\tilde{\Pi}_\alpha\mathcal{B}(\mathcal{H}_M)\tilde{\Pi}_\alpha$ is irreducible.
        \item The density operator $\sigma$ is block-diagonal with respect to the decomposition $\mathcal{H}_M=\bigoplus_{\alpha=1}^{J}\tilde{\Pi}_\alpha\mathcal{H}_M$, i.e.,
        $\sigma=\bigoplus_{\alpha=1}^{J}\tilde{\Pi}_\alpha\sigma\tilde{\Pi}_\alpha$.
        \item $\Phi(\tilde{\Pi}_\alpha\mathcal{B}(\mathcal{H}_M)\tilde{\Pi}_\alpha)\subseteq\mathcal{B}(\mathcal{H}_s)\otimes\tilde{\Pi}_\alpha\mathcal{B}(\mathcal{H}_M)\tilde{\Pi}_\alpha$.
        \item The restriction of $\mathcal{E}$ to $\tilde{\Pi}_\alpha\mathcal{B}(\mathcal{H}_M)\tilde{\Pi}_\alpha$ has a non-degenerate eigenvalue equal to 1, corresponding to the fixed point $\tilde{\Pi}_\alpha\sigma\tilde{\Pi}_\alpha$.
    \end{enumerate}
\end{proposition}
The decomposition of FCS into ergodic components follows from items 2 and 3 of Proposition \ref{fcs-erg-decomp}, 
\begin{equation}\label{conv-erg-gen}
    \rho_{R}=\sum_{\alpha=1}^J \lambda_\alpha\rho^\alpha_{R},
\end{equation}
where $\rho^\alpha_{R}=\mathrm{tr}_M\Phi^{|R|}(\sigma_\alpha)$, $\sigma_\alpha=\tilde{\Pi}_\alpha\sigma\tilde{\Pi}_\alpha/\mathrm{tr}(\tilde{\Pi}_\alpha\sigma\tilde{\Pi}_\alpha)$, and $\lambda_\alpha=\mathrm{tr}(\tilde{\Pi}_\alpha\sigma\tilde{\Pi}_\alpha)$. We observe that each ergodic component is manifestly a FCS. One can think of ergodic components being, in the sense of items 2 and 3, independent of each other. 

The structure of ergodic components can be analyzed further. We summarize their properties in the proposition below, which is a restatement of Proposition 3.3 of Ref. \onlinecite{fcs} in the language of density operators
\begin{proposition}[Proposition 3.3, Ref. \onlinecite{fcs}]\label{ergpgfc} 
For an ergodic FCS, generated by $(\Phi,\sigma)$, with the induced quantum channel $\mathcal{E}$, we may assume that the following properties hold
    \begin{enumerate}
        \item The peripheral spectrum of $\mathcal{E}$ consists of $p\in\mathbb{N}$ non-degenerate eigenvalues $\{\exp(\frac{2\pi i}{p}k)\}\;|\;k=0,\cdot\cdot\cdot,{p-1}\}$. Here $p$ is referred to as the \emph{period} of the state.
        \item There exist $p$ orthogonal projectors $\Pi_k:\mathcal{H}_M\rightarrow\mathcal{H}_M$, $\sum_{k=0}^{p-1}\Pi_k=\mathbb{1}_M$, such that $\mathcal{E}(\Pi_k\mathcal{B}(\mathcal{H}_M)\Pi_k)\subseteq\Pi_{k+1}\mathcal{B}(\mathcal{H}_M)\Pi_{k+1}$ with the convention $\Pi_{p}:=\Pi_{0}$, which leads to $\mathcal{E}^p(\Pi_k\mathcal{B}(\mathcal{H}_M)\Pi_k)\subseteq\Pi_{k}\mathcal{B}(\mathcal{H}_M)\Pi_{k}$.
        \item The density operator $\sigma$ is block-diagonal with respect to the decomposition $\mathcal{H}_M=\bigoplus_{k=0}^{p-1}\Pi_k\mathcal{H}_M$, i.e.,
        $\sigma=\bigoplus_{k=0}^{p-1}\Pi_k\sigma\Pi_k$. Moreover $\mathcal{E}(\Pi_k\sigma\Pi_k)=\Pi_{k+1}\sigma\Pi_{k+1}$, which leads to $\mathrm{tr}(\Pi_k\sigma\Pi_k)={1}/{p}$ for any $k$.
        \item The restriction of $\mathcal{E}^p$ to $\Pi_k\mathcal{B}(\mathcal{H}_M)\Pi_k$ has a trivial peripheral spectrum with $\Pi_k\sigma\Pi_k$ the fixed point.
    \end{enumerate}
\end{proposition}
We observe similarities between the statements of Proposition \ref{fcs-erg-decomp} and Proposition \ref{ergpgfc}. In particular, item 1 of Proposition \ref{fcs-erg-decomp} implies that the algebra of operators on the memory space, $\mathcal{B}(\mathcal{H}_M)$, contains $J$ orthogonal subalgebras invariant under $\mathcal{E}$. By item 2 of Proposition \ref{ergpgfc}, each of these subalgebras is supported on the memory space $\tilde{\Pi}_\alpha\mathcal{H}_M$ of a corresponding ergodic component, and further contains $p$ orthogonal subalgebras that are cyclically permuted by $\mathcal{E}$.

\subsubsection{Purely Generated FCS}
A subcollection of FCS, called \emph{purely generated FCS}, is characterized by a pure quantum channel $\Phi,$ of the form $\Phi(X)=VXV^\dagger$, where $V:\mathcal{H}_M\rightarrow\mathcal{H}_s\otimes\mathcal{H}_M$ is an isometry.\cite{fcs, pgfcs} We will say that pgFCS is induced, or generated, by the pair $(V,\sigma)$. Equivalently these states are characterized by vanishing mean entropy, i.e., for a continuous region $R$ of $|R|$ spins, $\lim_{|R|\rightarrow +\infty}S(\rho_R)/|R|=0$. \cite{pgfcs} 

In the language of MPS, a pgFCS generated by $(V,\sigma)$ corresponds to the reduced density operator of the properly defined translationally invariant limit of the MPS $|\Psi\rangle$ as $n\rightarrow +\infty$\cite{mps-repr} (see Appendix \ref{app-dictionary} for the diagrammatic illustration),
\begin{align*}
    |\Psi\rangle=\sum_{s_{-n},\dots,s_{n}=1}^{d_s} & \left(L\right| M^{s_{-n}}M^{s_{-n+1}}\cdot\cdot\cdot M^{s_{n}}\left|R\right) \\ \nonumber
    &\times|s_{-n}\rangle\otimes|s_{-n+1}\rangle\otimes\cdot\cdot\cdot\otimes|s_{n}\rangle,
\end{align*}
where $M^{s}$ is a $d_M\times d_M$ complex matrix, defined by its elements, 
\begin{equation*}
    M^{s}_{ij}:=(\langle s|\otimes \langle i|) V |j\rangle,
\end{equation*}
and $|L),|R)\in\mathbb{C}^{d_M}$ are boundary factors. In the limit of the infinite chain $\mathrm{tr}_M$ and $\sigma$ play a similar role as the left and right boundary factors $|L)$ and $|R)$, respectively.

For the case of pgFCS the induced channel $\mathcal{E}$ defined in (\ref{E-gen-fcs}) becomes
\begin{equation}\label{E-pgfcs}
    \mathcal{E}(X)=\mbox{tr}_s\left(VXV^\dagger\right).
\end{equation} 

To formulate the study of recoverability and QCMI we subdivide a finite region $R$ into continuous adjacent subregions $A$, $B$, and $C$, as discussed in the introduction. A pgFCS defined on this region may be expressed as
\begin{equation}\label{pgfcs-form}
\rho_{ABC} = \mbox{tr}_M \left(V_C V_B V_A \sigma V_A^{\dagger} V_B^{\dagger} V_C^{\dagger} \right).
\end{equation}
The isometry $V_A:\mathcal{H}_M\rightarrow\mathcal{H}_A\otimes\mathcal{H}_M$, generating all spins in the region $A$, is the $|A|$-fold product of $V$, $V_A:=VV\cdot\cdot\cdot V$, in the sense of (\ref{prod-isom-def}). The isometries $V_B$ and $V_C$ generate all spins in the regions $B$ and $C$, respectively, and are defined in an analogous manner. The products of these isometries are also defined by (\ref{prod-isom-def}). 

Note that while ergodic FCS and pgFCS form distinct subsets of FCS, their intersection, the ergodic pgFCS, is non-empty and plays an important role in this study. A pgFCS can be decomposed into a convex sum (\ref{conv-erg-gen}) of ergodic pgFCS $\rho^\alpha_{ABC}$, as shown in Lemma \ref{pgfcs-erg-pgfcs} below,\cite{fcs, pgfcs}
\begin{equation}\label{conv-erg}
    \rho_{ABC}=\sum_{\alpha=1}^J \lambda_\alpha\rho^\alpha_{ABC}.
\end{equation}
\begin{lemma}\label{pgfcs-erg-pgfcs}
A FCS is a pgFCS if and only if all its ergodic components are pgFCS.
\end{lemma}
\begin{proof}
As was shown in Ref. \onlinecite{pgfcs}, a state is a pgFCS if and only if it has vanishing mean entropy. Note the well-known bounds on the convexity of the quantum entropy, $$\sum_{\alpha=1}^J\lambda_\alpha S(\rho^\alpha_{ABC})\leq S(\rho_{ABC})\leq\sum_{\alpha=1}^J\lambda_\alpha S(\rho^\alpha_{ABC})-\sum_{\alpha=1}^J\lambda_\alpha\ln\lambda_\alpha.$$ Since $$\lim_{|ABC|\rightarrow +\infty}\sum_{\alpha=1}^J\lambda_\alpha\ln\lambda_\alpha/|ABC|=0,$$ then $$\lim_{|ABC|\rightarrow +\infty}S(\rho_{ABC})/|ABC|=\sum_{\alpha=1}^J\lambda_\alpha \lim_{|ABC|\rightarrow +\infty}S(\rho^\alpha_{ABC})/|ABC|.$$ Hence $\lim_{|ABC|\rightarrow +\infty}S(\rho_{ABC})/|ABC|=0$ if and only if $\lim_{|ABC|\rightarrow +\infty}S(\rho^\alpha_{ABC})/|ABC|=0$ for each $\alpha$, which implies that every ergodic component is a pgFCS.
\end{proof}

Since each ergodic component is a pgFCS, we may express $\rho_{ABC}^\alpha$ in the form
\begin{equation}\label{ergpgfcs-alpha}
    \rho^\alpha_{ABC} =\mbox{tr}_{M_\alpha}\left(V_C^\alpha V_B^\alpha V_A^\alpha \sigma_\alpha V_A^{\alpha\dagger} V_B^{\alpha\dagger} V_C^{\alpha\dagger}\right),
\end{equation} 
where all the objects are defined as in (\ref{pgfcs-form}), and each $\sigma_\alpha$ is defined on a distinct memory space $\mathcal{H}_{M_\alpha}$, $\alpha=1,\cdot\cdot\cdot,J$. Naturally, we denote the elementary isometry inducing the ergodic component by $V_\alpha$, i.e., $V_{A}^{\alpha}$ is the $|A|$-fold product $V_{A}^{\alpha} := V_\alpha V_\alpha \cdot \cdot \cdot V_\alpha$ in the sense of (\ref{prod-isom-def}). Since $\rho^\alpha_{ABC}$ is ergodic, the corresponding quantum channel $\mathcal{E}_\alpha:\mathcal{B}(\mathcal{H}_{M_\alpha})\rightarrow\mathcal{B}(\mathcal{H}_{M_\alpha})$, defined as in (\ref{E-pgfcs}), $\mathcal{E}_{\alpha}(X)=\mbox{tr}_{s}(V_\alpha X V_\alpha^\dagger)$, and $\sigma_{\alpha}$ have the properties listed in Proposition \ref{ergpgfc}. Now we rewrite Proposition \ref{fcs-erg-decomp} for the case of pgFCS, making a connection  between $(V,\sigma)$ generating the pgFCS $\rho_{ABC}$ and the collection of $(V_\alpha,\sigma_\alpha)$, $\alpha=1,\cdot\cdot\cdot,J$, generating its ergodic pgFCS components $\rho^\alpha_{ABC}$. We add an additional property, item 5, whose interpretation is that ergodic components are not proportional to each other; in Appendix \ref{app-d} we show that we may assume, without loss of generality, that this property holds.  
\begin{proposition}[Ref.\,\onlinecite{fcs}]\label{gen-pgfcs}
    Without loss of generality, we may assume that a pgFCS is generated by a pair $(V,\sigma)$, with the induced quantum channel $\mathcal{E}$, and the following properties hold
    \begin{enumerate}
        \item $\mathcal{H}_M=\bigoplus_{\alpha=1}^J\mathcal{H}_{M_\alpha}$, where $J$ is the number of ergodic components in $\rho_{ABC}$, and $\mathcal{H}_{M_\alpha}:=\tilde{\Pi}_\alpha\mathcal{H}_M$. 
        \item $V=\sum_{\alpha=1}^{J}V_\alpha\tilde{\Pi}_\alpha$, where $\tilde{\Pi}_\alpha$, ${\alpha=1,\cdot\cdot\cdot,J}$ is the orthogonal projector onto $\mathcal{H}_{M_\alpha}$, and $V_\alpha:\mathcal{H}_{M_\alpha}\rightarrow\mathcal{H}_s\otimes\mathcal{H}_{M_\alpha}$ is an isometry.
        \item $\sigma=\bigoplus_{\alpha=1}^{J}\lambda_\alpha\sigma_\alpha$, where $\sigma_\alpha={\tilde{\Pi}_\alpha\sigma\tilde{\Pi}_{\alpha}}/{\mathrm{tr}(\tilde{\Pi}_\alpha\sigma\tilde{\Pi}_{\alpha})}$, $\lambda_\alpha=\mathrm{tr}(\tilde{\Pi}_\alpha\sigma\tilde{\Pi}_{\alpha})$.
        \item $(V_{\alpha},\sigma_{\alpha})$ generates an ergodic pgFCS in the form (\ref{ergpgfcs-alpha}), with the induced quantum channel $\mathcal{E}_{\alpha}(X)\equiv\mathcal{E}(\tilde{\Pi}_\alpha X\tilde{\Pi}_\alpha)=\mbox{tr}_s(V_\alpha X V^\dagger_\alpha)$, for which Proposition \ref{ergpgfc} applies.
        \item For $\alpha\neq\beta$ there is no unitary $U:\mathcal{H}_M\rightarrow\mathcal{H}_M$ and $\phi\in\mathbb{R}$, such that $V_\alpha=e^{i\phi}UV_\beta U^\dagger$.
    \end{enumerate}
\end{proposition}

\noindent\textbf{Remarks}:\ 
\begin{enumerate}
\item Proposition \ref{gen-pgfcs} contains all the pgFCS properties that we use in this paper. 
\item The decomposition (\ref{conv-erg}) into ergodic components is not necessarily the finest one that can be performed on a pgFCS -- each ergodic component may be further decomposed into a sum of ergodic pgFCS components of period $1$. \cite{fcs} The resulting components may, however, involve adjustment of the memory Hilbert space, and blocking of spins -- the procedure under which several consecutive spins are treated as one. We prefer to deal directly with the decomposition of a pgFCS into ergodic pgFCS components of arbitrary period. In Section \ref{sec-proof} it will be convenient to consider an ergodic pgFCS of period $1$ first, in order to establish the arguments in our proof.
\end{enumerate}

\subsection{Quantum Markov chains}
Quantum Markov chains are defined as those states for which QCMI vanishes, $I(A:C|B):=S(\rho_{AB})+S(\rho_{BC})-S(\rho_{ABC})-S(\rho_{B})=0$, \cite{petz-recovery} and are fully characterized by the following theorem,
\begin{theorem}[Ref.\,\onlinecite{petz-recovery}]\label{qmc-structure}
Let $\rho_{ABC}\in\mathcal{D}(\mathcal{K}_{A}\otimes\mathcal{K}_{B}\otimes\mathcal{K}_{C})$. The following three statements are equivalent:
\begin{enumerate}
    \item $\rho_{ABC}$ is a quantum Markov chain, \emph{i.e.}, $I(A:C|B)=0$.
    \item There exists a quantum channel $\mathcal{R}_{B\rightarrow BC}:\mathcal{B}(\mathcal{K}_B)\rightarrow\mathcal{B}(\mathcal{K}_B\otimes\mathcal{K}_C)$, such that $\rho_{ABC}=\mathcal{R}_{B\rightarrow BC}(\rho_{AB})$. On $\mathcal{B}(\mathrm{supp}(\rho_B))$ this channel can be taken to be the \emph{Petz recovery map} $\mathcal{P}_{B\rightarrow BC}(X):=\rho_{BC}^{\frac{1}{2}}\rho_B^{-\frac{1}{2}} X \rho_B^{-\frac{1}{2}}\rho_{BC}^{\frac{1}{2}}, $ although this choice is not unique.
    \item There is a decomposition $\mathrm{supp}(\rho_B)\cong\bigoplus_{k=1}^{k_{\mathrm{max}}}\mathcal{K}_{b^l_k}\otimes\mathcal{K}_{b^r_k}$, \emph{i.e.}, there exists a unitary isomorphism $I:\mathrm{supp}(\rho_B)\rightarrow\bigoplus_{k=1}^{k_{\mathrm{max}}}\mathcal{K}_{b^l_k}\otimes\mathcal{K}_{b^r_k}$, such that $I\rho_{ABC}I^\dagger=\bigoplus_{k=1}^{k_{\mathrm{max}}}\lambda_k\rho_{Ab^l_k}\otimes\rho_{b^r_kC}$, where $\rho_{Ab^l_k}\in\mathcal{D}(\mathcal{K}_A\otimes\mathcal{K}_{b^l_k})$, $\rho_{b^r_kC}\in\mathcal{D}(\mathcal{K}_{b^r_k}\otimes\mathcal{K}_C)$, $\lambda_k>0$ and $\sum_{k=1}^{k_{\mathrm{max}}}\lambda_k=1$.
\end{enumerate}
\end{theorem}
For a pgFCS $\rho_{ABC}$ we will construct an approximating state $\tilde{\rho}_{ABC}$, and show that it is a quantum Markov chain by proving that it satisfies the property $3$, whence properties $1$ and $2$ as well. 

The theorem has a useful corollary, whose proof is omitted, which we will employ in the next section,
\begin{corollary}\label{qmc-structure-cor}
    If $\mathcal{K}_B=\mathcal{K}_{B_1}\oplus\mathcal{K}_{B_2},$ and $\rho_{AB_1C}\in\mathcal{D}(\mathcal{K}_{A}\otimes\mathcal{K}_{B_1}\otimes\mathcal{K}_{C})$ and $\rho_{AB_2C}\in\mathcal{D}(\mathcal{K}_{A}\otimes\mathcal{K}_{B_2}\otimes\mathcal{K}_{C})$ are quantum Markov chains, then $\lambda\rho_{AB_1C}+(1-\lambda)\rho_{AB_2C}$, where $0\leq \lambda \leq 1$, is a quantum Markov chain on $\mathcal{D}(\mathcal{K}_A\otimes\mathcal{K}_B\otimes\mathcal{K}_C)$.
\end{corollary}

\section{Proof of theorem \ref{main-result}}\label{sec-proof}
In this section we prove the main result, Theorem \ref{main-result}. The proof relies on two steps: (i) approximating pgFCS $\rho_{ABC}$ by another state $\tilde{\rho}_{ABC}$ and (ii) the proof that $\tilde{\rho}_{ABC}$ is a quantum Markov chain. Conveniently, the approximation step is the same for any pgFCS, and we outline it immediately below. The proof that the constructed $\tilde{\rho}_{ABC}$ is a quantum Markov chain, proceeds in several steps, which we develop in subsequent subsections. We first consider the simplest case of an ergodic pgFCS of period $1$, before generalizing to arbitrary period $p.$ Finally, by taking convex sums, we deal with a general case of pgFCS. 

We now deal with the approximation step, assuming that $\tilde{\rho}_{ABC}$ is a quantum Markov chain. This is sufficient to prove the Theorem \ref{main-result}.

Any pgFCS has the form (\ref{pgfcs-form}), 
\begin{equation*}
\rho_{ABC} = \mbox{tr}_M \left(V_C V_B V_A \sigma V_A^{\dagger} V_B^{\dagger} V_C^{\dagger} \right).
\end{equation*}
We introduce an isometry $\tilde{V}_B:\mathcal{H}_A \otimes\mathcal{H}_M\rightarrow\mathcal{H}_A\otimes\mathcal{H}_B\otimes\mathcal{H}_M$ which acts trivially on the space $\mathcal{H}_A$, and use it to construct the state
\begin{equation}\label{pgfcs-tilde-sec}
\tilde{\rho}_{ABC} = \mbox{tr}_M \left(V_C \tilde{V}_B V_A \sigma V_A^{\dagger} \tilde{V}_B^{\dagger} V_C^{\dagger} \right).
\end{equation}
Note that the isometry $\tilde{V}_B$ generates all the spins in the region $B$, but unlike $V_B$, it is not a concatenation of isometric factors generating one spin at a time. It is a key step to choose $\tilde{V}_B$ to be of a particular form that ensures the state $\tilde{\rho}_{ABC}$ is simultaneously close (in trace norm distance) to $\rho_{ABC}$ and is an exact quantum Markov chain. 

The trace distance $\|\rho_{ABC}-\tilde{\rho}_{ABC}\|_1$ is bounded by twice the error in approximating $V_{B}$ by $\tilde{V}_{B}$ in operator norm,
\begin{align}\label{1-1-leq-opnorm}
    \|\rho_{ABC}-\tilde{\rho}_{ABC}\|_1 &= \|\mbox{tr}_M \left(V_C V_B V_A \sigma V_A^{\dagger} V_B^{\dagger} V_C^{\dagger} \right)\\ \nonumber
    &\quad-\mbox{tr}_M \left(V_C \tilde{V}_B V_A \sigma V_A^{\dagger} \tilde{V}_B^{\dagger} V_C^{\dagger} \right)\|_1 \\ \nonumber
    &\leq \|V_B V_A \sigma V_A^{\dagger} V_B^{\dagger}-\tilde{V}_B V_A \sigma V_A^{\dagger} \tilde{V}_B^{\dagger}\|_1 \\ \nonumber
    &\leq 2\|V_B-\tilde{V}_B\|,
\end{align}
where in the last line we have used the inequalities $\|XY\|_1\leq 
\|X\|\|Y\|_1$, $\|XZX^\dagger-YZY^\dagger\|_1\leq 2\|X-Y\|\|Z\|_1$, and the fact that $\|V_A\sigma V_A^\dagger\|_1=1$.  

For $\tilde{\rho}_{ABC}$ a quantum Markov chain, according to condition 2 of Theorem \ref{qmc-structure} there exists a recovery channel $\mathcal{R}_{B\rightarrow BC}$, such that $\mathcal{R}_{B\rightarrow BC}(\tilde{\rho}_{AB})=\tilde{\rho}_{ABC}$. We may select $\mathcal{R}_{B\rightarrow BC}$ to be the Petz recovery map, defined in statement 2 of Theorem \ref{qmc-structure}, or one of the alternatives that we discuss in Section \ref{sec-rec-maps}. If we use such an $\mathcal{R}_{B\rightarrow BC}$ to approximately recover $\rho_{ABC}$ from $\rho_{AB}$, then the recovery error is given by
\begin{align}\label{error-leq-1-1}
&\nonumber\|\rho_{ABC}-\mathcal{R}_{B\rightarrow BC}\left(\rho_{AB}\right)\|_1 =\|\rho_{ABC}-\mathcal{R}_{B\rightarrow BC}\circ \mathrm{tr}_C\left(\rho_{ABC}\right)\|_1\\ \nonumber
&= \|\rho_{ABC}-\tilde{\rho}_{ABC}+\mathcal{R}_{B\rightarrow BC}\circ \mathrm{tr}_C\left(\tilde{\rho}_{ABC}-\rho_{ABC}\right)\|_1  \\ \nonumber
&\leq \|\rho_{ABC}-\tilde{\rho}_{ABC}\|_1+\|\mathcal{R}_{B\rightarrow BC}\circ \mathrm{tr}_C\left(\rho_{ABC}-\tilde{\rho}_{ABC}\right)\|_1 \\ 
&\leq 2\|\rho_{ABC}-\tilde{\rho}_{ABC}\|_1 
\end{align}
where we have used the quantum Markov property, the triangle inequality, and the property that the trace norm is non-increasing under actions of quantum channels.

Combining (\ref{1-1-leq-opnorm}) and (\ref{error-leq-1-1}), we obtain a bound on the recovery error, 
\begin{equation}\label{error-leq-op} 
    \|\rho_{ABC}-\mathcal{R}_{B\rightarrow BC}\left(\rho_{AB}\right)\|_1\leq 4\|V_B-\tilde{V}_B\|.
\end{equation}

Our goal is to choose the isometry $\tilde{V}_B$ so that $\|V_B-\tilde{V}_B\|$ is exponentially small in the size of the region $B$. Consider the $|B|$-fold composition of $\mathcal{E}$, defined in (\ref{E-pgfcs}), $\mathcal{E}^{|B|}$. Since $\mathcal{E}$ is dilated by the elementary isometry $V$, then the $|B|$-fold composition of $V$, that is $V_B$, dilates $\mathcal{E}^{|B|}$.
From $\mathcal{E}$ we construct the quantum channel $\tilde{\mathcal{E}}$ via Definition \ref{e-tilde-e}, and then approximate $\mathcal{E}^{|B|}$ by $\tilde{\mathcal{E}}^{|B|}$, the $|B|$-fold composition of $\tilde{\mathcal{E}}$. As discussed in Section \ref{subsec-quant-chan}, the maximum dimension of the dilating space of $\tilde{\mathcal{E}}^{|B|}$ is equal to $\dim\mathcal{B}(\mathcal{H}_M)=d_M^2$. We will assume that $\dim\mathcal{H}_B\geq d_M^2$, so that the dilation space of $\tilde{\mathcal{E}}^{|B|}$ can be embedded into $\mathcal{H}_B$. Under this condition, we define $\tilde{V}_B$ to be an isometry dilating $\tilde{\mathcal{E}}^{|B|}$. Note that any isometry of the form $U_B\tilde{V}_B$, where $U_B:\mathcal{H}_B\rightarrow\mathcal{H}_B$ is unitary, will also be a dilating isometry for $\tilde{\mathcal{E}}^{|B|}$. We will use this freedom of choice to minimize $\|V_B-\tilde{V}_B\|$. 

By the continuity of Stinespring's dilation, \cite{infdisturb} we can relate the operator norm distance between $V_B$ and $\tilde{V}_B$ to the distance between channels $\mathcal{E}^{|B|}$ and $\mathcal{\tilde{E}}^{|B|}$ through the inequalities  
\begin{equation}\label{inf-dist}
    \inf_{U_{B}}\|V_{B}-U_{B}\tilde{V}_{B}\|^2\leq\|\mathcal{E}^{|B|}-\tilde{\mathcal{E}}^{|B|}\|_\diamond\leq 2\inf_{U_{B}}\|V_{B}-U_{B}\tilde{V}_{B}\|,
\end{equation}
where the infimum is taken over the set of all unitaries $U_B:\mathcal{H}_B\rightarrow\mathcal{H}_B.$ (Recall that the diamond norm is defined by $\|\mathcal{E}^{|B|}-\tilde{\mathcal{E}}^{|B|}\|_\diamond:=\sup_{n\in\mathbb{N}}\|\mbox{id}_n\otimes(\mathcal{E}^{|B|}-\tilde{\mathcal{E}}^{|B|})\|_{1\rightarrow 1}$.\cite{kitaev}) As the dimension of $\mathcal{H}_B$ is finite, the infimum is attained for some $U_B$, which we take to be the identity, $U_{B}=\mathbb{1}_{B}$, so that our choice of $\tilde{V}_B$ is optimal. Using the first inequality in (\ref{inf-dist}), we bound the distance between the isometries as 
\begin{equation}\label{opnorm-leq-diam}
    \|V_B-\tilde{V}_{B}\|\leq \sqrt{\|\mathcal{E}^{|B|}-\tilde{\mathcal{E}}^{|B|}\|_\diamond}.
\end{equation}

To bound the diamond norm $\|\mathcal{E}^{|B|}-\tilde{\mathcal{E}}^{|B|}\|_\diamond$ we take advantage of the $2\rightarrow 2$ norm, which is easier to estimate. (The $2\rightarrow 2$ norm is defined by $\|\mathcal{E}\|_{2\rightarrow 2}:=\sup_{\|X\|_2=1}\|\mathcal{E}(X)\|_2$, where the 2-norm, or Hilbert-Schmidt norm, is given by $\|X\|_2 :=\sqrt{\mbox{tr}(X^\dagger X)}$.) Since the Hilbert space $\mathcal{H}_M$ is finite dimensional, the supremum $\sup_{n\in\mathbb{N}}\|\mbox{id}_n\otimes(\mathcal{E}^{|B|}-\tilde{\mathcal{E}}^{|B|})\|_{1\rightarrow 1}$ is achieved when $n=d_M$,\cite{kitaev} hence $\|\mathcal{E}^{|B|}-\tilde{\mathcal{E}}^{|B|}\|_\diamond=\|\mbox{id}_{d_M}\otimes(\mathcal{E}^{|B|}-\tilde{\mathcal{E}}^{|B|})\|_{1\rightarrow 1}$. Then, using the relations $\|X\|_1\leq\mathrm{rank}(X)^{1/2}\|X\|_2$ and $\|\mbox{id}_n\otimes(\mathcal{E}^{|B|}-\tilde{\mathcal{E}}^{|B|})\|_{2\rightarrow 2}=\|\mathcal{E}^{|B|}-\tilde{\mathcal{E}}^{|B|}\|_{2\rightarrow 2}$ (tensoring with the identity map does not change the $2\rightarrow 2$ norm), and bounding the rank with the space dimension $\mathrm{rank}(X)\leq d_M^2$, we obtain the well-known estimate
\begin{equation}\label{diam-leq-2-2}
    \|\mathcal{E}^{|B|}-\tilde{\mathcal{E}}^{|B|}\|_\diamond\leq d_M\|\mathcal{E}^{|B|}-\tilde{\mathcal{E}}^{|B|}\|_{2\rightarrow 2}.
\end{equation}
From Lemma \ref{b-2} applied to the quantum channels $\mathcal{E}$ and $\tilde{\mathcal{E}}$ it follows that,
\begin{equation}\label{2-2-leq-exp}
    \|\mathcal{E}^{|B|}-\tilde{\mathcal{E}}^{|B|}\|_{2\rightarrow 2}\leq c\nu^{|B|},
\end{equation}
with $c>0$ and $\nu_{\mathrm{gap}}<\nu<1$, where $\nu_{\mathrm{gap}}:=\max_{|\nu_i|<1}\{|\nu_i|\}$, and $\nu_i$ are the eigenvalues of $\mathcal{E}$. Combining together (\ref{opnorm-leq-diam}), (\ref{diam-leq-2-2}), and (\ref{2-2-leq-exp}), we obtain
\begin{equation}\label{isonorm-leq-exp}
\|V_B-\tilde{V}_B\|\leq \sqrt{d_M c} \; \nu^{|B|/2},
\end{equation}
and, from (\ref{error-leq-op}), 
\begin{align}\label{error-bound}
    \|\rho_{ABC}-\tilde{\rho}_{ABC}\|_1&\leq 2\sqrt{d_Mc} \ \ \nu^{|B|/2}\\ \nonumber
    \|\rho_{ABC}-\mathcal{R}_{B\rightarrow BC}\left(\rho_{AB}\right)\|_1 &\leq 4\sqrt{d_Mc} \ \ \nu^{|B|/2}.
\end{align}
Observe that statements $1$ and $2$ of Theorem \ref{main-result} hold with the constants $K=4\sqrt{d_Mc}$ and $q=\ln\left(\nu^{-1}\right)/2$. The bound (\ref{error-bound}) is meaningful while $4\sqrt{d_Mc}\nu^{|B|/2}\leq 2$, since the trace distance between any two states cannot exceed $2$. The statement $3$ of the theorem follows from the relation (\ref{qcmi-1-1-norm}) in Appendix \ref{app-bound} with $\tilde{K}=\sqrt{d_Mc}\left(2|A|\ln d_s+2-2\ln(2\sqrt{d_Mc}\nu^{{|B|}/{2}})\right)$. Notice that $\ln(2\sqrt{d_Mc}\nu^{{|B|}/{2}})\leq 0$ when the bound (\ref{error-bound}) is useful. This completes the proof. 

It remains to prove that $\tilde{\rho}_{ABC}$ given in (\ref{pgfcs-tilde-sec}), and defined in terms of the isometry $\tilde{V}_B$, is a quantum Markov chain.

\subsection{Ergodic pgFCS of period 1}\label{subsec-pgfcs}

We first prove that $\tilde{\rho}_{ABC}$ is a quantum Markov chain for an ergodic pgFCS with $p=1.$ 
Since $\mathcal{E}$ has a trivial peripheral spectrum and $\sigma$ is its fixed point, then $\tilde{\mathcal{E}}$ constructed according to Definition \ref{e-tilde-e} sends any input into the state proportional to $\sigma$, $\tilde{\mathcal{E}}(X)=\mathrm{tr}(X)\sigma$. It follows that for any $n\in\mathbb{N}$, $\tilde{\mathcal{E}}^n=\tilde{\mathcal{E}}$, in particular $\tilde{\mathcal{E}}^{|B|}=\tilde{\mathcal{E}}$. We observe that the isometry $\tilde{V}'_{B}=\sum_{i,j=1}^{d_M}\sqrt{\sigma_i} |\xi_{ij}\rangle \otimes |i\rangle\langle j|$ is a dilation of $ \tilde{\mathcal{E}}^{|B|}=\tilde{\mathcal{E}}$, where we have introduced an arbitrary orthonormal set of vectors $\{| \xi_{ij}\rangle \in \mathcal{H}_B\;|\; i,j = 1 \cdot \cdot \cdot d_M \}$. The vectors $|i\rangle\in\mathcal{H}_M$ and $\sigma_i>0$ are the eigenvectors and corresponding eigenvalues of $\sigma$, respectively, i.e., $\sigma=\sum_{i=1}^{d_M}\sigma_i |i\rangle\langle i|$. Indeed,
\begin{align*}
    \mathrm{tr}_B(\tilde{V}'_B X \tilde{V}'^\dagger_B)&=\sum_{i,j,i',j'}\sqrt{\sigma_i\sigma_{i'}} \langle \xi_{i'j'}|\xi_{ij}\rangle \langle j|X|j'\rangle\, |i\rangle\langle i'|\\ \nonumber
    &=\sum_{j=1}^{d_M} \langle j|X|j\rangle\, \sum_{i=1}^{d_M}\sigma_i|i\rangle\langle i|\\ \nonumber
    &=\mbox{tr}(X)\sigma.
\end{align*}
Since $\tilde{V}'_{B}$ dilates $\tilde{\mathcal{E}}$, which is also dilated by $\tilde{V}_B$, there is unitary $U_B:\mathcal{H}_B\rightarrow\mathcal{H}_B$, such that $\tilde{V}_B=U_B\tilde{V}'_B$, and we observe that this amounts to a unitary change of the orthonormal basis $\{|\xi_{ij}\rangle\}$. Thus, considering the freedom of choice of the basis vectors $|\xi_{ij}\rangle$, we can assume without loss of generality that $\tilde{V}_B=\tilde{V}'_B$,
\begin{equation}\label{tilde-iso-form}
    \tilde{V}_{B}=\sum_{i,j=1}^{d_M}\sqrt{\sigma_i} |\xi_{ij}\rangle \otimes |i\rangle\langle j|.
\end{equation}

Substituting (\ref{tilde-iso-form}) into (\ref{pgfcs-tilde-sec}), $\tilde{\rho}_{ABC}$ may be written explicitly,
\begin{align}\label{tilde-rho-explic}
    \tilde{\rho}_{ABC}=\sum_{i,i',j,j'} \langle j|V_A \sigma V_A^\dagger|j'\rangle &\otimes\sqrt{\sigma_i}|\xi_{ij}\rangle\langle\xi_{i'j'}|\sqrt{\sigma_{i'}}\\ \nonumber &\otimes\mbox{tr}_M\left(V_C|i\rangle\langle i'|V_C^\dagger\right).
\end{align}
We will show by direct computation that $\tilde{\rho}_{ABC}$ satisfies the condition 3 of Theorem \ref{qmc-structure} and hence is a quantum Markov chain. We denote $\mathcal{H}_b:=\mathrm{span}\{|\xi_{ij}\rangle\;|\;i,j=1,\cdot\cdot\cdot,d_M \rangle\}$, and note that since their dimensions are equal, $\mathcal{H}_b$ and 
$\mathcal{H}_M\otimes\mathcal{H}_M$ are isomorphic. We construct the unitary isomorphism $I:\mathcal{H}_b\rightarrow\mathcal{H}_{M}\otimes\mathcal{H}_{M}$, defined by $I|\xi_{ij}\rangle=|j\rangle\otimes|i\rangle$. Under this isomorphism, a straightforward calculation gives, 
\begin{equation*}
    I\tilde{\rho}_{ABC}I^\dagger=V_A\sigma V^\dagger_{A}\otimes\mbox{tr}_M\left(V_C\sigma^{\frac{1}{2}}|+\rangle\langle +|\sigma^{\frac{1}{2}}V^\dagger_C\right),
\end{equation*}
where $|+\rangle=\sum_{i=1}^{d_M}|i\rangle\otimes|i\rangle\in\mathcal{H}_M\otimes\mathcal{H}_{M}$ is a non-normalized, maximally entangled state vector, and the partial trace $\mathrm{tr}_M$ is over the second factor. If we write $\mathcal{H}_b=\mathcal{H}_{b^l}\otimes\mathcal{H}_{b^r}$, where $\mathcal{H}_{b^l}$ and $\mathcal{H}_{b^r}$ are isomorphic to $\mathcal{H}_M$, we observe that $\tilde{\rho}_{ABC}$ satisfies condition 3 of Theorem \ref{qmc-structure} with
\begin{align*}
    &\mathrm{supp}(\rho_B)\cong\mathcal{H}_b=\mathcal{H}_{b^l}\otimes\mathcal{H}_{b^r},\\ \nonumber
    &I\tilde{\rho}_{ABC}I^\dagger=\tilde{\rho}_{Ab^l} \otimes \tilde{\rho}_{b^rC},\\ \nonumber
    &\tilde{\rho}_{Ab^l}:=V_A\sigma V^\dagger_A, \\ \nonumber
    &\tilde{\rho}_{b^rC}:=\mbox{tr}_M(V_C\sigma^{\frac{1}{2}}|+\rangle\langle +|\sigma^{\frac{1}{2}}V^\dagger_C),
\end{align*}
and is thus a quantum Markov chain.

\subsection{Ergodic pgFCS of period $p$}\label{subsec-erg}
For an ergodic pgFCS with $p \geq 1$, Proposition \ref{ergpgfc} states that the peripheral spectrum of $\mathcal{E}$ contains $p$ eigenvalues, and thus for $p>1$ the proof that $\tilde{\rho}_{ABC}$ is a quantum Markov chain must be modified.

The properties of $\mathcal{E}$ listed in Proposition \ref{ergpgfc} allow us to prove Lemma \ref{subsec-erg-lemma} below, stating that (assuming $\dim\mathcal{H}_B\geq d_M^2$) $\tilde{V}_B$ can be decomposed into $\tilde{V}_B=\sum_{k=0}^{p-1}\tilde{V}^k_{B}$, where each $\tilde{V}_B^{k}:\mathcal{H}_M\rightarrow\mathcal{H}_B\otimes\mathcal{H}_M$ is a partial isometry with $\mathrm{supp}(\tilde{V}_B^{k})=\Pi_k\mathcal{H}_M$, and $\tilde{V}_B^{k\dagger}\tilde{V}_B^{k}=\Pi_k$. The range of $\tilde{V}^k_B$ satisfies $$\mathrm{Range}(\tilde{V}_B^k)\subseteq\mathcal{H}_{B_k}\otimes\Pi_{k+|B|}\mathcal{H}_M\subseteq \mathcal{H}_B\otimes\mathcal{H}_M,$$ where the subspaces $\mathcal{H}_{B_k}\subseteq\mathcal{H}_{B}$ are mutually orthogonal. Below, each $\tilde{V}^k_{B}$ will be seen to have the form (\ref{tilde-isom-r-form}) similar to (\ref{tilde-iso-form}). The operator $\sigma$ is block diagonal with respect to the set of all projectors $\Pi_k,$ $k=0,\cdot \cdot \cdot, p-1.$

Anticipating the validity of Lemma \ref{subsec-erg-lemma}, below, we may express $\tilde{\rho}_{ABC}$ in the form 
\begin{widetext}
\begin{align}\label{ergpgfs-direct}
    \tilde{\rho}_{ABC} =\mbox{tr}_M \left(V_C \left(\sum_{k_1=0}^{p-1}\tilde{V}^{k_1}_{B}\right) V_A \sigma V_A^{\dagger} \left(\sum_{k_2=0}^{p-1}\tilde{V}^{k_2\dagger}_{B}\right)V_C^{\dagger} \right)&=\mbox{tr}_M \left(V_C \left(\sum_{k_1=0}^{p-1}\tilde{V}^{k_1}_{B}\right)\left(\sum_{k=0}^{p-1}\Pi_k V_A \sigma V_A^{\dagger} \Pi_k\right)\left(\sum_{k_2=0}^{p-1}\tilde{V}^{k_2\dagger}_{B}\right)V_C^{\dagger} \right)\\ \nonumber
    &=\bigoplus_{k=0}^{p-1} \mbox{tr}_M \left(V_C \tilde{V}^{k}_{B} V_A \sigma V_A^{\dagger}\tilde{V}_{B}^{k\dagger} V_C^{\dagger} \right),
\end{align}
\end{widetext}
where in the first line we substitute the decomposition of $\tilde{V}_B$ into partial isometries and make the replacement $V_A\sigma V^\dagger_A=\sum_{k=1}^p \Pi_kV_A\sigma V^\dagger_A\Pi_k$ (see Appendix \ref{app-c} for proof), and $\tilde{V}^{k'}_{B}\Pi_k=\delta_{kk'}\tilde{V}^{k}_{B}$ collapses the expression to a single sum in the second line. Since the ranges of $\tilde{V}^k_B$ are mutually orthogonal, the summands corresponding to different $k$ are supported on orthogonal subspaces, so that the sum is direct. 

Each term in the direct sum can be expressed as in (\ref{tilde-rho-explic}), and by the arguments of Section \ref{subsec-pgfcs} is shown to be (up to normalization) a quantum Markov chain. Hence by Corollary \ref{qmc-structure-cor}, $\tilde{\rho}_{ABC}$ is also a quantum Markov chain. Thus we can use the bound for the recovery error (\ref{error-bound}), and the statement of Theorem \ref{main-result} holds with the same constants $q$, $K$ and $\tilde{K}$ as in Section \ref{subsec-pgfcs}. 

We now prove that $\tilde{V}_B$ has the required decomposition.
\begin{lemma}\label{subsec-erg-lemma}
For an ergodic pgFCS of period $p$, the isometry $\tilde{V}_B:\mathcal{H}_M\rightarrow\mathcal{H}_B\otimes\mathcal{H}_M$ defined in (\ref{pgfcs-tilde-sec}) can be constructed as the sum $\tilde{V}_B=\sum_{k=0}^{p-1}\tilde{V}_B^k$, where 
\begin{enumerate}
    \item $\tilde{V}_{B}^k$ is a partial isometry, $\tilde{V}_B^{k\dagger}\tilde{V}_B^k=\Pi_k$, and $\tilde{V}^k_B\Pi_{k'}=\delta_{kk'}\tilde{V}_B^k$. 
    \item The range of $\tilde{V}_B^k$ satisfies $\mathrm{Range}(\tilde{V}_B^k)\subseteq\mathcal{H}_{B_k}\otimes\Pi_{k+|B|}\mathcal{H}_M\subseteq \mathcal{H}_B\otimes\mathcal{H}_M$. The subspaces $\mathcal{H}_{B_k}$ are mutually orthogonal.
    \item The density operator $p\,\mathrm{tr}_M(V_C\tilde{V}_B^kV_A\sigma V^\dagger_A\tilde{V}_B^{k\dagger}V_C^\dagger)$ is a quantum Markov chain. 
\end{enumerate}
\end{lemma}

\noindent\textbf{Remarks}:\ 
\begin{enumerate}
\item Explicitly, each $\tilde{V}_{B}^k$, $k=0,\cdot\cdot\cdot,p-1$, has the form 
    \begin{equation}\label{tilde-isom-r-form}
        \tilde{V}_B^k=\sqrt{p}\sum_{(i,j)\in\mathcal{O}_k}\sqrt{\sigma_i}|\xi_{ij}\rangle\otimes|i\rangle\langle j|, 
    \end{equation}
where $\{|\xi_{ij}\rangle\in\mathcal{H}_B\;|\;i,j=1,\cdot\cdot\cdot,d_M\}$ is a set of orthonormal vectors, $\langle\xi_{i'j'}|\xi_{ij}\rangle=\delta_{ii'}\delta_{jj'}$, the vectors $|i\rangle,|j\rangle\in\mathcal{H}_M$ and the eigenvalues $\sigma_i>0$ are defined by the spectral decomposition $\sigma=\sum_{i=1}^{d_M}\sigma_i|i\rangle\langle i|$, and $\mathcal{O}_k:=\{(i,j)\;|\;|i\rangle\in\Pi_{k+|B|}\mathcal{H}_M, |j\rangle\in\Pi_{k}\mathcal{H}_M\}$. Notice that if $p>1$, then the set $\bigcup_{k=0}^{p-1}\mathcal{O}_k$ does not contain all pairs $(i,j)$.
\item The subspaces $\mathcal{H}_{B_k}$ are defined as $\mathcal{H}_{B_k}:=\mathrm{span}\{|\xi_{ij}\rangle\in\mathcal{H}_B\;|\; (i,j)\in\mathcal{O}_k\}$. It follows that the subspaces $\mathcal{H}_{B_k}$ are mutually orthogonal, and implies that for $X\in\mathcal{B}(\mathcal{H}_M)$ and $k\neq k'$, the sum $\tilde{V}^k_BX\tilde{V}^{k\dagger}_B+\tilde{V}^{k'}_{B}X\tilde{V}^{k'\dagger}_B$ is direct, i.e., $\tilde{V}^k_BX\tilde{V}^{k\dagger}_B\in\mathcal{B}(\mathcal{H}_{B_k})\otimes\mathcal{B}(\mathcal{H}_M)$, $\tilde{V}^{k'}_BX\tilde{V}^{k'\dagger}_B\in\mathcal{B}(\mathcal{H}_{B_{k'}})\otimes\mathcal{B}(\mathcal{H}_M)$ and $\tilde{V}^k_BX\tilde{V}^{k\dagger}_B+\tilde{V}^{k'}_{B}X\tilde{V}^{k'\dagger}_B\in\left(\mathcal{B}(\mathcal{H}_{B_k})\oplus\mathcal{B}(\mathcal{H}_{B_{k'}})\right)\otimes\mathcal{B}(\mathcal{H}_M)$.
\end{enumerate}

\begin{proof}
As we show in Appendix \ref{app-erg-e-tilde}, the quantum channel $\tilde{\mathcal{E}}$, obtained via Definition \ref{e-tilde-e} from $\mathcal{E}$ with the properties listed in Proposition \ref{ergpgfc}, has the form
\begin{equation}\label{E-tilde}
    \tilde{\mathcal{E}}(X)=p\sum_{k=0}^{p-1} \mbox{tr}\left(\Pi_k X \Pi_k\right)\Pi_{k+1}\sigma\Pi_{k+1}.
\end{equation}
In $\Pi_{k+k'}$ the addition in the subscript is interpreted modulo $p$. Then,
\begin{equation*}
    \tilde{\mathcal{E}}^{|B|}(X)=p\sum_{k=0}^{p-1} \mbox{tr}\left(\Pi_k X \Pi_k\right)\Pi_{k+|B|}\sigma\Pi_{k+|B|}.
\end{equation*}
Notice that for each $k$ and $X\in\mathcal{B}(\mathcal{H}_M)$ the map, $X\mapsto p\mbox{tr}\left(\Pi_k X \Pi_k\right)\Pi_{k+|B|}\sigma\Pi_{k+|B|}$, while not trace-preserving, is completely positive, and has an associated dilation $\tilde{V}_B^k:\mathcal{H}_M\rightarrow\mathcal{H}_B\otimes\mathcal{H}_M$, defined in (\ref{tilde-isom-r-form}).

It is easy to check that $\tilde{V}_B^k$ is a partial isometry, $$\tilde{V}_B^{k\dagger}\tilde{V}_B^k=\sum_{|j\rangle\in\Pi_k\mathcal{H}_M}|j\rangle\langle j|=\Pi_k.$$ We define the subspace $\mathcal{H}_{B_k}:=\mathrm{span}\{|\xi_{ij}\rangle\in\mathcal{H}_B\;|\; (i,j)\in\mathcal{O}_k\}$ that has the dimension $\dim\mathcal{H}_{B_k}=\mathrm{rank}(\Pi_k)\mathrm{rank}(\Pi_{k+|B|})$.

Now we prove that the assumption $\dim\mathcal{H}_B\geq d_M^2$ is consistent with mutual orthogonality of the subspaces $\mathcal{H}_{B_k}$. For $\mathcal{H}_B$ to contain $\bigoplus_{k=0}^{p-1}\mathcal{H}_{B_k}$, we require   $\sum_{k=0}^{p-1}\dim\mathcal{H}_{B_k}\leq\dim\mathcal{H}_B$, and since $\sum_{k=0}^{p-1} \mathrm{rank}(\Pi_k)\mathrm{rank}(\Pi_{k+|B|})\leq (\sum_{k=0}^{p-1} \mathrm{rank}(\Pi_k))^2=d^2_M$, this condition is satisfied. 

The channel $\tilde{\mathcal{E}}^{|B|}$ has the associated dilation isometry $\tilde{V}'_B:=\sum_{k=0}^{p-1} \tilde{V}^k_B$, since $$\tilde{V}_B^{'\dagger}\tilde{V}'_B=\sum_{k_1,k_2=0}^{p-1} \tilde{V}^{k_1\dagger}_B\tilde{V}^{k_2}_B=\sum_{k=0}^{p-1} \tilde{V}^{k\dagger}_B\tilde{V}^k_B=\sum_{k=0}^{p-1}\Pi_k=\mathbb{1}_B,$$ and
\begin{align*}
\mathrm{tr}_{B}(\tilde{V}'_BX\tilde{V}^{'\dagger}_B)&=\sum_{k_1,k_2=0}^{p-1}\mathrm{tr}_{B}(\tilde{V}^{k_1}_BX\tilde{V}^{k_2\dagger}_B)\\
&=\sum_{k=0}^{p-1}\mathrm{tr}_{B}(\tilde{V}^{k}_BX\tilde{V}^{k\dagger}_B)\\
&=p\sum_{k=0}^{p-1}\mbox{tr}\left(\Pi_k X \Pi_k\right)\Pi_{k+|B|}\sigma\Pi_{k+|B|}\\
&=\tilde{\mathcal{E}}^{|B|}(X).
\end{align*}
We have used mutual orthogonality of $\mathcal{H}_{B_k}$ in both calculations. Since both $\tilde{V}'_B,\tilde{V}_B:\mathcal{H}_M\rightarrow\mathcal{H}_B\otimes\mathcal{H}_M$ dilate $\tilde{\mathcal{E}}^B$, there exists unitary $U_B:\mathcal{H}_B\rightarrow\mathcal{H}_B$, such that $\tilde{V}_B=U_B\tilde{V}'_B$. Using the same reasoning as in Section \ref{subsec-pgfcs} we can, without loss of generality, take $\tilde{V}'_B=\tilde{V}_B$,
\begin{equation}\label{tilde-isom-erg-form}
    \tilde{V}_B=\sum_{k=0}^{p-1} \tilde{V}_B^k=\sqrt{p}\sum_{k=0}^{p-1}\sum_{(i,j)\in\mathcal{O}_k}\sqrt{\sigma_i}|\xi_{ij}\rangle\otimes|i\rangle\langle j|.
\end{equation}

Notice that the case $p=1$ corresponds to a single projector $\Pi_0=\mathbb{1}_M$, so that (\ref{E-tilde}) reduces to $\tilde{\mathcal{E}}^{|B|}(X)=\mathrm{tr}(X)\sigma$, and (\ref{tilde-isom-erg-form}) reduces to (\ref{tilde-iso-form}).

Comparing (\ref{tilde-isom-r-form}) to (\ref{tilde-iso-form}), we realize that the we can use the same arguments as in Section \ref{subsec-pgfcs} to prove that $p\mathrm{tr}_M(V_C\tilde{V}_B^kV_A\sigma V_A\tilde{V}_B^{k\dagger}V_C^\dagger)$ is a quantum Markov chain. The normalization constant comes from the observation that 
\begin{align*}
    \mathrm{tr}(V_C\tilde{V}^k_BV_A\sigma V^\dagger_A\tilde{V}_B^{k\dagger}V_C^\dagger)&=\mathrm{tr}(\Pi_kV_A\sigma V^\dagger_A\Pi_k)\\ \nonumber
    &=\mathrm{tr}\left(\Pi_{k}\mathcal{E}^{|A|}(\sigma)\Pi_k\right)\\ \nonumber
    &=\mathrm{tr}\left(\Pi_{k}\mathcal{E}^{|A|}\left(\bigoplus_{k'=0}^{p-1}\Pi_{k'}\sigma\Pi_{k'}\right)\Pi_k\right)\\ \nonumber
    &=\mathrm{tr}(\Pi_{k}\bigoplus_{k'=0}^{p-1}\Pi_{k'+|A|}\sigma\Pi_{k'+|A|}\Pi_k)& \\ \nonumber
    &=\mathrm{tr}(\Pi_{k}\sigma\Pi_k)=\frac{1}{p}.
\end{align*}
\end{proof}

\subsection{The general case of pgFCS}\label{sect-fcs}
We generalize the proof to the case of a convex sum of ergodic pgFCS considered in the previous subsection. The properties of pgFCS are given in Proposition \ref{gen-pgfcs}, whose notation and definitions we employ below.

We will prove that $\tilde{\rho}_{ABC}$, as defined in (\ref{pgfcs-tilde-sec}), is a quantum Markov chain by anticipating the Lemma \ref{gen-pgfcs-isom-decomp}, which states that the isometry $\tilde{V}_B$ can be decomposed as a sum of partial isometries corresponding to ergodic components, $\tilde{V}_B=\sum_{\alpha=1}^J\tilde{V}^\alpha_{B}$, where $\tilde{V}^\alpha_B:\mathcal{H}_M\rightarrow\mathcal{H}_{B}\otimes\mathcal{H}_M$ is a partial isometry with $\tilde{\Pi}_\alpha\mathcal{H}_{M}$ the orthogonal complement of its kernel, and where $\mathrm{Range}(\tilde{V}_B^\alpha)=\mathcal{H}_{B_\alpha}\otimes\tilde{\Pi}_\alpha \mathcal{H}_M$, with subspaces $\mathcal{H}_{B_\alpha}\subseteq\mathcal{H}_{B}$ being mutually orthogonal. The projectors $\tilde{\Pi}_\alpha$ are defined in Proposition \ref{gen-pgfcs}, and correspond to different ergodic components of $\rho_{ABC}$. Moreover, $\tilde{V}^\alpha_B$ are such that 
\begin{equation*}
\tilde{\rho}^\alpha_{ABC} := \mbox{tr}_M \left(V^\alpha_C \tilde{V}^\alpha_B V^\alpha_A \sigma_\alpha V_A^{\alpha\dagger} \tilde{V}_B^{\alpha\dagger} V_C^{\alpha\dagger} \right)
\end{equation*}
separately approximates the corresponding $\rho^\alpha_{ABC}$, defined in Proposition \ref{gen-pgfcs}. Then $\tilde{\rho}_{ABC}$ may be decomposed,
\begin{align*}
    \tilde{\rho}_{ABC}&=\mbox{tr}_M \left(V_C \left(\sum_{\alpha=1}^J\tilde{V}^\alpha_{B}\right) V_A \sigma V_A^{\dagger} \left(\sum_{\beta=1}^J\tilde{V}^{\beta\dagger}_{B}\right) V_C^{\dagger} \right)\\ \nonumber
    &=\bigoplus_{\alpha=1}^J\lambda_\alpha\mbox{tr}_{M_\alpha}\left(V^\alpha_C \tilde{V}^\alpha_{B}V^\alpha_A \sigma_\alpha V_A^{\alpha\dagger} \tilde{V}^{\alpha\dagger}_{B} V_C^{\alpha\dagger} \right)\\ \nonumber
    &=\bigoplus_{\alpha=1}^J\lambda_\alpha\tilde{\rho}^\alpha_{ABC},
\end{align*}
where we have used $V_{C}\tilde{V}^\alpha_{B}=V_C^\alpha\tilde{V}^\alpha_{B}$ and $\tilde{V}^\alpha_B V_A\sigma V^\dagger_A \tilde{V}_B^\beta=\delta_{\alpha\beta}\lambda_\alpha\tilde{V}^\alpha_B V^\alpha_A\sigma_\alpha V^{\alpha\dagger}_A\tilde{V}_B^\beta$. By the proof in Section \ref{subsec-erg}, each $\tilde{\rho}^\alpha_{ABC}$ is a quantum Markov chain, and by Corollary \ref{qmc-structure-cor} $\tilde{\rho}_{ABC}$ is a quantum Markov chain. This is sufficient to complete the proof. 

The following lemma shows that $\tilde{V}_B$ indeed possesses the requisite decomposition. We repeatedly use Proposition \ref{gen-pgfcs} in the proof.
\begin{lemma}\label{gen-pgfcs-isom-decomp}
For a pgFCS, the isometry $\tilde{V}_B:\mathcal{H}_M\rightarrow\mathcal{H}_B\otimes\mathcal{H}_M$ defined in (\ref{pgfcs-tilde-sec}) can be constructed as the sum $\tilde{V}_B=\sum_{\alpha=1}^{J}\tilde{V}_B^\alpha$, where $J$ is the number of ergodic components, and
\begin{enumerate}
    \item $\tilde{V}_{B}^\alpha$ is a partial isometry, $\tilde{V}_B^{\alpha\dagger}\tilde{V}_B^\alpha=\tilde{\Pi}_\alpha$, and $\tilde{V}^\alpha_B\tilde{\Pi}_{\beta}=\delta_{\alpha\beta}\tilde{V}_B^\alpha$. 
    \item $\mathrm{Range}(\tilde{V}_B^\alpha)=\mathcal{H}_{B_\alpha}\otimes\tilde{\Pi}_{\alpha}\mathcal{H}_M\subseteq \mathcal{H}_B\otimes\mathcal{H}_M$. The subspaces $\mathcal{H}_{B_\alpha}$ are mutually orthogonal.
    \item $\tilde{\rho}^\alpha_{ABC}=\mathrm{tr}_M(V^\alpha_C\tilde{V}_B^{\alpha}V^\alpha_A\sigma V^{\alpha\dagger}_A\tilde{V}_B^{\alpha\dagger}V_C^{\alpha\dagger})$ is a quantum Markov chain. 
\end{enumerate}
\end{lemma}
\begin{proof}
We express $\mathcal{E}$ as 
\begin{align}\label{ab-channels}
    \mathcal{E}(X)=\sum_{\alpha,\beta=1}^{J} \mathcal{E}\left(\tilde{\Pi}_{\alpha}X\tilde{\Pi}_\beta\right)=\sum_{\alpha,\beta=1}^{J}\mbox{tr}_{s}\left(V_\alpha \tilde{\Pi}_{\alpha}X\tilde{\Pi}_\beta V_{\beta}^{\dagger}\right),
\end{align}
where we have used $V\tilde{\Pi}_\alpha=V_\alpha\tilde{\Pi}_\alpha$. Since the range of $V_\alpha$ is $\mathcal{H}_s\otimes\tilde{\Pi}_\alpha\mathcal{H}_M$, then $\mbox{tr}_{s}\left(V_\alpha \tilde{\Pi}_{\alpha}X\tilde{\Pi}_\beta V_{\beta}^{\dagger}\right)\in\tilde{\Pi}_{\alpha}\mathcal{B}(\mathcal{H}_M)\tilde{\Pi}_\beta$. 
We decompose $\tilde{\mathcal{E}}$ similarly,
\begin{equation*}
    \tilde{\mathcal{E}}(X)=\sum_{\alpha,\beta=1}^{J} \tilde{\mathcal{E}}\left(\tilde{\Pi}_{\alpha}X\tilde{\Pi}_\beta\right).
\end{equation*}
We now show that the last equation reduces to the direct sum, $\tilde{\mathcal{E}}(X)=\bigoplus_{\alpha=1}^J\tilde{\mathcal{E}}(\tilde{\Pi}_{\alpha}X\tilde{\Pi}_\alpha)$. We observe that the map, $X\mapsto\tilde{\mathcal{E}}(\tilde{\Pi}_{\alpha}X\tilde{\Pi}_\beta)$, is obtained via Definition \ref{e-tilde-e} from the map $X\mapsto\mbox{tr}_{s}(V_\alpha \tilde{\Pi}_{\alpha}X\tilde{\Pi}_\beta V_{\beta}^{\dagger})$. Recall that for $\alpha\neq\beta$ there is no unitary $U:\mathcal{H}_M\rightarrow\mathcal{H}_M$ and $\phi\in\mathbb{R}$, such that $V_\alpha=e^{i\phi}UV_\beta U^\dagger$, which implies by Lemma \ref{b-1} that the eigenvalues of $X\mapsto\mbox{tr}_{s}(V_\alpha \tilde{\Pi}_{\alpha}X\tilde{\Pi}_\beta V_{\beta}^{\dagger})$ have magnitudes less than $1$. Then, following the construction in Definition \ref{e-tilde-e}, the map $X\mapsto\mbox{tr}_{s}(V_\alpha \tilde{\Pi}_{\alpha}X\tilde{\Pi}_\beta V_{\beta}^{\dagger})$ is identically zero, $\tilde{\mathcal{E}}(\tilde{\Pi}_{\alpha}X\tilde{\Pi}_\beta)=0$ for $\alpha\neq\beta$. Thus, $\tilde{\mathcal{E}}(X)=\sum_{\alpha=1}^J\tilde{\mathcal{E}}(\tilde{\Pi}_{\alpha}X\tilde{\Pi}_\alpha)$. Since $\mathcal{E}(\tilde{\Pi}_{\alpha}X\tilde{\Pi}_\alpha)\in\tilde{\Pi}_{\alpha}\mathcal{B}(\mathcal{H}_M)\tilde{\Pi}_\alpha$, from Definition \ref{e-tilde-e} follows that $\tilde{\mathcal{E}}(\tilde{\Pi}_{\alpha}X\tilde{\Pi}_\alpha)\in\tilde{\Pi}_{\alpha}\mathcal{B}(\mathcal{H}_M)\tilde{\Pi}_\alpha$. Since $\tilde{\Pi}_\alpha$ are orthonormal projectors, the sum is direct, $\tilde{\mathcal{E}}(X)=\bigoplus_{\alpha=1}^J\tilde{\mathcal{E}}(\tilde{\Pi}_{\alpha}X\tilde{\Pi}_\alpha)$, and
\begin{equation*}
    \tilde{\mathcal{E}}^{|B|}(X)=\bigoplus_{\alpha=1}^J\tilde{\mathcal{E}}^{|B|}(\tilde{\Pi}_{\alpha}X\tilde{\Pi}_\alpha).
\end{equation*}
We define $\tilde{V}_B^\alpha:\mathcal{H}_{M_\alpha}\rightarrow\mathcal{H}_{B_\alpha}\otimes\mathcal{H}_{M_\alpha}$ to be an isometric dilation of $X\mapsto\tilde{\mathcal{E}}(\tilde{\Pi}_\alpha X\tilde{\Pi}_{\alpha})$. Here $\mathcal{H}_{M_\alpha}:=\tilde{\Pi}_\alpha\mathcal{H}_M$, and $\mathcal{H}_{B_\alpha}$ is the dilation space embedded in $\mathcal{H}_B$ of dimension $\dim\mathcal{H}_{B_\alpha}\leq(\dim\mathcal{H}_{M_\alpha})^2=d^2_{M_\alpha}$. It immediately follows that $\tilde{V}_B^{\alpha\dagger}\tilde{V}_B^\alpha=\tilde{\Pi}_\alpha$, proving the statement 1. From property 4 of Proposition \ref{gen-pgfcs} we know that $(V_\alpha,\sigma_\alpha)$ generates an ergodic pgFCS with the induced quantum channel $X\mapsto\mathcal{E}(\tilde{\Pi}_\alpha X\tilde{\Pi}_\alpha)$. As mentioned before, $X\mapsto\tilde{\mathcal{E}}(\tilde{\Pi}_{\alpha}X\tilde{\Pi}_\alpha)$ is exactly the map obtained via Definition \ref{e-tilde-e} from the map $X\mapsto\mathcal{E}(\tilde{\Pi}_\alpha X\tilde{\Pi}_\alpha)$. Then, separately for each $\alpha$, we can repeat the reasoning of Section \ref{subsec-erg}, showing that each  $\tilde{V}_B^\alpha$, the dilating isometry of $X\mapsto\tilde{\mathcal{E}}^{|B|}(\tilde{\Pi}_\alpha X\tilde{\Pi}_\alpha)$, possesses the properties which ensure that $\tilde{\rho}^\alpha_{ABC}$ is a quantum Markov chain, proving the statement 3. 

Now we show that we can safely choose the subspaces $\mathcal{H}_{B_\alpha}$ to be mutually orthogonal. Since $\mathcal{H}_{M}=\bigoplus_{\alpha=1}^J\mathcal{H}_{M_\alpha}$, and since we assume $\dim\mathcal{H}_{B}\geq d_M^2$, then $\dim\mathcal{H}_B\geq(\sum_{\alpha=1}^Jd_{M_\alpha})^2\geq \sum_{\alpha=1}^Jd^2_{M_\alpha}\geq\sum_{\alpha=1}^J\dim\mathcal{H}_{B_\alpha}$, which implies that the subspaces $\mathcal{H}_{B_\alpha}$ may be chosen to be orthogonal, proving the statement 2. Then we can construct the isometry $\tilde{V}'_B=\sum_{\alpha=1}^J\tilde{V}^\alpha_B$. This isometry clearly dilates $\tilde{\mathcal{E}}^{|B|}$, which is also dilated by $\tilde{V}_B$. By the same reasoning as in Section \ref{subsec-pgfcs} and Section \ref{subsec-erg} we may take $\tilde{V}'_B=\tilde{V}_B$,
\begin{equation*}
    \tilde{V}_B=\sum_{\alpha=1}^J\tilde{V}^\alpha_B.
\end{equation*}
Thus, $\tilde{V}_B$ possesses all the required properties.
\end{proof}

\section{Recovery maps}\label{sec-rec-maps}
In this section we consider the recovery maps that can be used to restore $\tilde{\rho}_{ABC}$, defined in (\ref{pgfcs-tilde-sec}), from $\tilde{\rho}_{AB}$. The obvious choice is the Petz recovery map, defined on $\mathcal{B}(\mathrm{supp}(\tilde{\rho}_B))$ in (\ref{petz-map}), $\mathcal{P}_{B\rightarrow BC}(X)=\tilde{\rho}_{BC}^{\frac{1}{2}}\tilde{\rho}^{-\frac{1}{2}}_BX\tilde{\rho}^{-\frac{1}{2}}_B\tilde{\rho}_{BC}^{\frac{1}{2}}$, and which can be extended by a trivial embedding $\mathrm{supp}(\tilde{\rho}_B)^\perp\rightarrow\mathrm{supp}(\tilde{\rho}_B)^\perp\otimes\mathcal{H}_C$ to the map $\mathcal{P}_{B\rightarrow BC}:\mathcal{B}(\mathcal{H}_B)\rightarrow\mathcal{B}(\mathcal{H}_B)\otimes\mathcal{B}(\mathcal{H}_C)$.\cite{petz-1, petz-2, petz-recovery} The Petz recovery map can be defined for any state that is a quantum Markov chain.

The Petz recovery map is, however, not the only map that can exactly reconstruct a quantum Markov chain $\tilde{\rho}_{ABC}$. From Theorem \ref{qmc-structure}, we note that there is another obvious choice of the recovery channel, which relies on the isomorphism $I:\mathrm{supp}(\tilde{\rho}_B)\rightarrow\bigoplus_{k=1}^{k_{\mathrm{max}}}\mathcal{H}_{b^l_k}\otimes\mathcal{H}_{b^r_k}$ defined in Theorem \ref{qmc-structure}, and which provides the decomposition $I\tilde{\rho}_{ABC}I^\dagger=\bigoplus_{k=1}^{k_{\mathrm{max}}}\lambda_k\tilde{\rho}_{Ab^l_k}\otimes\tilde{\rho}_{b^r_k C}$. We gave an example of such an isomorphism in Section \ref{subsec-pgfcs} for the case of an ergodic pgFCS of period $1$. 
Using $I$, for a general quantum Markov chain, we can construct the map $\mathcal{R}_{B\rightarrow BC}:\mathcal{B}(\mathcal{H}_B)\rightarrow\mathcal{B}(\mathcal{H}_B\otimes\mathcal{H}_C)$,  
\begin{equation*}
\mathcal{R}_{B\rightarrow BC}(X):=I^\dagger \sum_{k=1}^{k_{\mathrm{max}}}\mathrm{tr}_{b^r_k}\left(\mathbb{P}_k (IXI^\dagger)\mathbb{P}_k\right)\otimes \tilde{\rho}_{b^r_k C} I,
\end{equation*}
where $\mathbb{P}_k$ is a projector onto $\mathcal{H}_{b^l_k}\otimes\mathcal{H}_{b^r_k}$. We verify that this map is indeed a recovery channel by direct computation,
\begin{align*}
    &\mathcal{R}_{B\rightarrow BC}(\tilde{\rho}_{AB})=I^\dagger \sum_{k=1}^{k_{\mathrm{max}}}\mathrm{tr}_{b^r_k}\left(\mathbb{P}_k (I\tilde{\rho}_{AB}I^\dagger)\mathbb{P}_k\right)\otimes \tilde{\rho}_{b^r_k C} I\\ \nonumber
    &=I^\dagger\sum_{k=1}^{k_{\mathrm{max}}}\mathrm{tr}_{b^r_k}(\mathbb{P}_k \left(\bigoplus_{k'=1}^{k_{\mathrm{max}}}\lambda_{k'}\tilde{\rho}_{Ab^l_{k'}}\otimes\tilde{\rho}_{b^r_{k'} }\right)\mathbb{P}_k)\otimes \tilde{\rho}_{b^r_k C} I\\ \nonumber
    &=I^\dagger\bigoplus_{k'=1}^{k_{\mathrm{max}}}\lambda_{k'}\mathrm{tr}_{b^r_{k'}} (\tilde{\rho}_{Ab^l_{k'}}\otimes\tilde{\rho}_{b^r_{k'} })\otimes \rho_{b^r_k C} I\\ \nonumber
    &=I^\dagger\bigoplus_{k'=1}^{k_{\mathrm{max}}}\lambda_{k'}\tilde{\rho}_{Ab^l_{k'}}\otimes\tilde{\rho}_{b^r_{k'} C} I\\ \nonumber
    &=\tilde{\rho}_{ABC},
\end{align*}
where we have used $I\tilde{\rho}_{AB}I^\dagger=I\mathrm{tr}_C\tilde{\rho}_{ABC}I^\dagger=\mathrm{tr}_C(I\tilde{\rho}_{ABC}I^\dagger)$ (recall that $I$ acts non-trivially only on $\mathcal{H}_B$), the structure of the quantum Markov chain from Theorem \ref{qmc-structure} in the second line, and $\mathbb{P}_k\rho_{Ab^l_{k'}}\otimes\rho_{b^r_{k'} C}\mathbb{P}_k=\delta_{{k'} k}\rho_{Ab^l_{k'}}\otimes\rho_{b^r_{k'} C}$ in the third line.

Unlike the Petz recovery map that is defined in terms of the reduced density operators $\tilde{\rho}_{BC}$ and $\tilde{\rho}_B$, the recovery map $\mathcal{R}_{B\rightarrow BC}$ requires knowledge of the isomorphism $I$, the structure of which is not obvious in the case of a general quantum Markov chain. In the case of a quantum Markov chain approximating a pgFCS, and possessing the form (\ref{pgfcs-tilde-sec}), we can construct the isomorphism $I$. For simplicity of argument, we illustrate this claim using the case of ergodic pgFCS of period $1$. Recall that $\tilde{\rho}_{ABC}$ has the form (\ref{pgfcs-tilde-sec}), and $\tilde{V}_B$ has the form (\ref{tilde-iso-form}), $\tilde{V}_B=\sum_{i,j=1}^{d_M}\sqrt{\sigma_i}|\xi_{ij}\rangle\otimes|i\rangle\langle j|$, where $|i\rangle$ and $|j\rangle$ are the eigenvectors of $\sigma$, and $|\xi_{ij}\rangle$ is an orthonormal basis set. This allows us to express $\tilde{\rho}_{ABC}$ as in (\ref{tilde-rho-explic}) and define $I$ in terms of the orthonormal vectors $|\xi_{ij}\rangle$, $I=\sum_{i,j=1}^{d_M}(|i\rangle\otimes|j\rangle)\langle\xi_{ij}|$. The required set of orthonormal vectors $\{|\xi_{ij}\rangle\in\mathcal{H}_B\;|\; i,j=1,\cdot\cdot\cdot,d_M\}$ is the solution of the minimization problem for the norm $\|V_B-\tilde{V}_{B}\|$ over unitary change of $d_M^2$ basis vectors in $\mathcal{H}_B$.

Interestingly, we can construct another recovery channel, that has a slightly less optimal bound in terms of the pre-exponential factor, but for which the vectors $|\xi_{ij}\rangle$ can be constructed more explicitly. Again we use the case of ergodic pgFCS of period $1$ for illustration. 

For $\tilde{V}_B$ in the form (\ref{tilde-iso-form}), \emph{any} choice of the $|\xi_{ij}\rangle$ induces an exact quantum Markov chain $\tilde{\rho}_{ABC}$ of the form (\ref{pgfcs-tilde-sec}). A particular choice of $|\xi_{ij}\rangle$ influences the bound on the recovery error. Expressing
\begin{align*}
   \|V_B-\tilde{V}_B\|&=\|\sum_{i,j=1}^{d_M^2} (\langle i|V_B|j\rangle-\sqrt{\sigma_i}|\xi_{ij}\rangle)\otimes|i\rangle\langle j|\|.
\end{align*}
we observe that any choice of $|\xi_{ij}\rangle$ close to $\sigma^{-\frac{1}{2}}_{i}\langle i|V_B|j\rangle$ is expected to be a good one. Let us construct the positive semidefinite matrix 
\begin{equation}\label{gram-m}
    Q_{i'j';ij}:=\langle j'|V^\dagger_B\left(\mathbb{1}_B\otimes|i'\rangle\langle i|\right) V_B|j\rangle,
\end{equation}
which is a Gramian matrix for the set of vectors $\{\langle i|V_B| j\rangle\; | \; i,j=1,\cdot\cdot\cdot, d_M\}$. For large $|B|$, the vectors in the latter set are almost orthogonal to each other, and $\|\langle i|V_B| j\rangle\|$ is close to $\sigma_i^{1/2}$. To prove this, notice that the map $$X\mapsto V^\dagger_B(\mathbb{1}_B\otimes X)V_B$$ with $X\in\mathcal{B}(\mathcal{H}_M)$ is adjoint (with respect to the Hilbert-Schmidt inner product) to the quantum channel $\mathcal{E}^{|B|}(X)=\mathrm{tr}_B(V_B X V^\dagger_B)$, which converges to $\tilde{\mathcal{E}}(X)=\mathrm{tr}(X)\sigma$ in the limit $|B|\rightarrow\infty$. Hence the map $X\mapsto V^\dagger_B(\mathbb{1}_B\otimes X)V_B$ converges to the adjoint of $\tilde{\mathcal{E}}$, the map $X\mapsto \mathrm{tr}(\sigma X)\mathbb{1}_M$. Thus, $Q_{i'j';ij}$ converges to $\langle j'|\mbox{tr}_M\left(\sigma|i'\rangle\langle i|\right)\mathbb{1}_M V_B|j\rangle=\sigma_i\delta_{ii'}\delta_{jj'}$. 

Then we define $|\xi_{ij}\rangle$ as
\begin{equation}\label{xi-constr}
    |\xi_{ij}\rangle=\sum_{m,n=1}^{d_M}(Q^{-\frac{1}{2}})_{mn;ij}\langle m|V_B|n\rangle.
\end{equation}
(As we show in Appendix \ref{app-alt-rec}, $Q$ is guaranteed to be invertible, if the region $B$ is large enough). We verify that the set of $|\xi_{ij}\rangle$ is orthonormal,
\begin{align*}
    \langle\xi_{i'j'}|\xi_{ij}\rangle&=\sum_{\mathclap{m,n,m',n'}}(Q^{-\frac{1}{2}})_{i'j';m'n'}\langle n'|V^\dagger_B|m'\rangle\langle m|V_B|n\rangle (Q^{-\frac{1}{2}})_{mn;ij}\\ \nonumber
    &=\sum_{\mathclap{m,n,m',n'}}(Q^{-\frac{1}{2}})_{i'j';m'n'}Q_{m'n';mn}(Q^{-\frac{1}{2}})_{mn;ij}\\ \nonumber
    &=\delta_{ii'}\delta_{jj'}.
\end{align*}
In Appendix \ref{app-alt-rec} we show that for this choice of $|\xi_{ij}\rangle$,  
\begin{equation}\label{alt-op-norm-bound-sec}
    \|V_B-\tilde{V}_B\|\leq\sqrt{2(1+2^{-\frac{11}{2}})c}\;d_M\nu^{\frac{|B|}{2}}, 
\end{equation}
with $c$ and $\nu$ defined in Lemma \ref{b-2}, leading to the recovery error
\begin{equation}\label{alt-rec-error}
    \|\mathcal{R}_{B\rightarrow BC}(\rho_{AB})-\rho_{ABC}\|_1\leq 4\sqrt{2(1+2^{-\frac{11}{2}})c}\;d_M\nu^{\frac{|B|}{2}},
\end{equation}
which is worse than the recovery error (\ref{error-bound}) by the factor $\sqrt{2(1+2^{-\frac{11}{2}})d_M}$. 

In Appendix \ref{app-alt-rec} we present a similar construction for the general case of pgFCS, which is based on the same idea, but is more delicate.

\section{Conclusion}\label{sec-concl}

Quantum conditional mutual information (QCMI) is non-negative by virtue of the strong subadditivity of von Neumann entropy. We have studied the QCMI for a homogeneous quantum spin chain described by a purely generated finitely correlated state (pgFCS) $\rho_{ABC}$. Explicitly, by separating the chain regions A and C by a domain B, and provided the size $| B |$ of region B is large enough, we show there exists a quantum Markov chain which differs from $\rho_{ABC}$ in trace distance by an error exponentially small in $| B |$. It follows that $\rho_{ABC}$ may be recovered, within the stated error, from the reduced density operator $\rho_{AB}$ by a quantum channel acting only on the domain B. As a consequence this implies that the QCMI, denoted $I(A:C | B),$ is also exponentially small in $|B|.$ We have presented quantum channels that, in principle, can perform the state recovery. Besides the obvious choice of the Petz recovery map, we give two other examples derived from the particular structure of the pgFCS. We have carried out numerical experiments on the decay of QCMI, the results of which are consistent with the bound presented in Theorem \ref{main-result}. The details will be reported elsewhere. 

A natural next step would be to prove the exponential decay of QCMI for generic FCS, as conjectured in Ref. \onlinecite{fcs-mpdo}. Both the results presented here and in Ref. \onlinecite{fcs-mpdo} point in this direction. As discussed further in Appendix \ref{app-comparison}, the pgFCS considered here and the class of states considered in Ref. \onlinecite{fcs-mpdo} are in general distinct. 

Generalization of our result to generic FCS seems to require an approach different from the one we used in this paper. It remains an open question whether a FCS exhibiting an exponential decay of QCMI is necessarily close to a quantum Markov chain in trace distance.  

\begin{acknowledgments}
The authors would like to thank Shivan Mittal for helpful discussions and the referee for helpful suggestions.
\end{acknowledgments}

\section*{Author declarations}
\subsection*{Conflict of interests}
The authors have no conflicts to disclose.

\section*{Data Availability Statement}
Data sharing is not applicable to this article as no new data were created or analyzed in this study.

\appendix

\section{}\label{app-dictionary}
We present in the form of Table \ref{table1} the dictionary for conversion from the language of pgFCS defined by the pair $(V,\sigma)$ (see Section \ref{subsec-fcs-itrod}) to the language of matrix product states (in diagrammatic notation).

\begin{table*}[tb]
\centering
\begin{tabular}{|c|c|}
\hline
\rule{0pt}{2em}
pgFCS notation & Diagrammatic notation 
\rule[-1.5em]{0pt}{1em}
\\
\hline\hline

$\left(\langle s|\otimes\langle i|\right)V|j\rangle$ &
\rule{0pt}{2em}
\raisebox{-.45\height}{
\begin{tikzpicture}
\node (isom) [roundnode] {M};
          \draw (isom.east) to [edge label=$j$] ++(3mm,0)                                 
                (isom.west) to [edge label'=$i$] ++ (-3mm,0)
                (isom.north) to [edge label=$s$] ++(0,3mm);
\end{tikzpicture}
}
\rule[-1.9em]{0pt}{1em}
\\
\hline

$V^\dagger V=\mathbb{1}_M$ &
\rule{0pt}{3.7em}
\raisebox{-.45\height}{
\begin{tikzpicture}
\node (up) [roundnode] {$M$};
\path ($(up)+(0,-12mm)$) node (dn) [roundnode] {$M^\dagger$};
\path ($0.5*(up.east)+0.5*(dn.east)+(8mm,0)$) node (eq) {{\LARGE $=$}};
\path ($(eq.east)+(4mm,0)$) coordinate (id);
\foreach \x in {up,dn} {\draw[-] (\x.west)--++(-3mm,0);}
\draw[-] (up) -- (dn);
\draw[-] (up.east) -- ++(3mm,0) |- (dn.east);
\draw[-] let \p1 = ($(dn)-(id)$) in (id) -- ++(0,\y1) -- ++(-2mm,0);
\draw[-] let \p1 = ($(up)-(id)$) in (id) -- ++(0,\y1) -- ++(-2mm,0);
\end{tikzpicture}
}
\rule[-3.2em]{0pt}{1em}
\\
\hline

$\langle i|\sigma|j\rangle$ &
\rule{0pt}{2em}
\raisebox{-.45\height}{
\begin{tikzpicture}
\node (sigma) [sigmanode] {$\sigma$};
\draw (sigma.north) to [edge label=$i$] ++(0,3mm)
      (sigma.south) to [edge label'=$j$] ++(0,-3mm);
\end{tikzpicture}
}
\rule[-1.7em]{0pt}{1em}
\\
\hline

$\mathrm{tr}_M(V\sigma V^\dagger)=\mathcal{E}(\sigma)=\sigma$ &
\rule{0pt}{3.7em}
\raisebox{-.45\height}{
\begin{tikzpicture}
\node (up1) [roundnode] {$M$};
\path ($(up1)+(0,-12mm)$) node (dn1) [roundnode] {$M^\dagger$};
\path ($0.5*(up1.west)+0.5*(dn1.west)+(-3mm,0)$) node (sigma)  [sigmanode] {$\sigma$};
\path ($0.5*(up1.east)+0.5*(dn1.east)+(10mm,0)$) node (eq1) {{\LARGE $=$}};
\path ($(eq1.east)+(3mm,0)$) node (sigma2) [anchor=west,sigmanode] {$\sigma$};
\draw[-] (sigma)|-(up1);
\draw[-] (sigma)|-(dn1);
\draw[-] (up1.east)--++(3mm,0);
\draw[-] (dn1.east)--++(3mm,0);
\draw[-] (up1) -- (dn1);
\draw[-] let \p1 = ($(dn1)-(sigma.south)$) in (sigma2.south) -- ++(0,\y1) -- ++(3mm,0);
\draw[-] let \p1 = ($(up1)-(sigma.north)$) in (sigma2.north) -- ++(0,\y1) -- ++(3mm,0);
\end{tikzpicture}
}
\rule[-3.2em]{0pt}{1em}
\\
\hline

$\rho_{ABC} = \mbox{tr}_M \left(V_C V_B V_A \sigma V_A^{\dagger} V_B^{\dagger} V_C^{\dagger} \right)$ &
\rule{0pt}{4.7em}
\raisebox{-.5\height}{
\begin{tikzpicture}
\matrix (rabc) [matrix of nodes, nodes in empty cells, nodes=roundnode, row sep=8mm,column sep=2mm] {
$M$\& $M$\&$M$ \& $M$\&$M$ \&$M$ \& $M$ \\
$M^\dagger$\& $M^\dagger$\& $M^\dagger$\& $M^\dagger$\& $M^\dagger$\& $M^\dagger$\& $M^\dagger$ \\
};

\path ($0.5*(rabc-1-1.west)+0.5*(rabc-2-1.west)+(-5mm,0)$) node (sigma)  [sigmanode] {$\sigma$};

\foreach \x [evaluate=\x as \xprev using \x-1] in {2,3,...,7} {
                            \draw (rabc-1-\x.north) -- ++(0,2mm);
                            \draw (rabc-2-\x.south) -- ++(0,-2mm);
                            \draw (rabc-1-\x) -- (rabc-1-\xprev);
                            \draw (rabc-2-\x) -- (rabc-2-\xprev);}
\draw[-] (rabc-1-1.north) -- ++(0,2mm)
         (rabc-2-1.south) -- ++(0,-2mm)
         (rabc-1-7.east) -- ++ (3mm,0) |- (rabc-2-7.east);
\foreach \x in {1,2} {\draw[-] (sigma)|-(rabc-\x-1.west);}

\draw [decorate,decoration={brace,amplitude=3mm}, thick] ($(rabc-1-2.south east)+(2mm,-1mm)$) --             ($(rabc-1-1.south west)+(-2mm,-1mm)$) node [black,midway,below,yshift=-3mm] {$A$};
\draw [decorate,decoration={brace,amplitude=3mm}, thick] ($(rabc-1-5.south east)+(2mm,-1mm)$) -- ($(rabc-1-3.south west)+(-2mm,-1mm)$) node [black,midway,below,yshift=-3mm] {$B$};
\draw [decorate,decoration={brace,amplitude=3mm}, thick] ($(rabc-1-7.south east)+(2mm,-1mm)$) -- ($(rabc-1-6.south west)+(-2mm,-1mm)$) node [black,midway,below,yshift=-3mm] {$C$};
\end{tikzpicture}
}
\rule[-4.7em]{0pt}{1em}
\\
\hline

\end{tabular}
\caption{The dictionary for conversion from the language of pgFCS defined by the pair $(V,\sigma)$ to the language of matrix product states (in diagrammatic notation)}
\label{table1}
\end{table*}

\section{}\label{app-bound}
In this appendix we derive the right-hand side of the inequality (\ref{qcmi-rec}). In Refs. \onlinecite{universal-rec-1,universal-rec-2, universal-rec-3, kim-mpdo}, the bound is based on the Alicki-Fannes inequality, \cite{alicki-fannes} which is the extension of the Fannes inequality. \cite{fannes-ineq} Here we employ an improved version of the Alicki-Fannes inequality from Ref. \onlinecite{tight-bounds} (we refer the reader to this paper for further references on the topic). 

In order to apply the results of Ref. \onlinecite{tight-bounds}, we first express the QCMI in terms of \emph{quantum relative entropy}, $D(\rho||\sigma):=\mathrm{tr}\rho(\ln\rho-\ln\sigma)$,
\begin{align*}
    I(A:C|B)&=S(\rho_{BC})-S(\rho_{ABC})+S(\rho_{AB})-S(\rho_{B})\\ \nonumber
    &=D(\rho_{ABC}||\rho_{BC})-D(\rho_{AB}||\rho_B).
\end{align*}
By the monotonicity of quantum relative entropy under the action of quantum channels, \cite{rel-entr-mon, watrouse} $D(\mathcal{R}_{B\rightarrow BC}(\rho_{AB})||\mathcal{R}_{B\rightarrow BC}(\rho_B))\leq D(\rho_{AB}||\rho_B)$, thus we can estimate
\begin{equation*}
    I(A:C|B)\leq D(\rho_{ABC}||\rho_{BC})-D\left(\mathcal{R}_{B\rightarrow BC}(\rho_{AB})||\mathcal{R}_{B\rightarrow BC}(\rho_B)\right).
\end{equation*}
To make the result of Ref. \onlinecite{tight-bounds} directly applicable, we recast the above inequality in terms of \emph{conditional entropy}, $S(R_1|R_2)_\rho:=S(\rho_{R_1R_2})-S(\rho_{R_2})$ (we adopt the notation of Ref. \onlinecite{tight-bounds}). We notice that $D(\rho_{ABC}||\rho_{BC})=-S(A|BC)_\rho$ and $D(\mathcal{R}_{B\rightarrow BC}(\rho_{AB})||\mathcal{R}_{B\rightarrow BC}(\rho_B))=-S(A|BC)_{\rho'}$, where $\rho_{ABC}'=\mathcal{R}_{B\rightarrow BC}(\rho_{AB})$. Then
\begin{equation*}
    I(A:C|B)\leq S(A|BC)_{\rho'}-S(A|BC)_{\rho}.
\end{equation*}
We apply Lemma 2 of Ref. \onlinecite{tight-bounds} to obtain the bound
\begin{equation*}
    I(A:C|B)\leq \epsilon'\ln\dim\mathcal{H}_{A}+(1+\epsilon')h\left(\frac{\epsilon'}{1+\epsilon'}\right),
\end{equation*}
where $h(p)=-p\ln p-(1-p)\ln(1-p)$ is the binary entropy and $\epsilon'=\frac{1}{2}\|\rho_{ABC}-\mathcal{R}_{B \rightarrow BC}(\rho_{AB})\|_1$. This bound, unlike the one that follows from the Alicki-Fannes inequality, is applicable for any $0\leq\epsilon'\leq 1$. Notice that we expressed the inequality (\ref{qcmi-rec}) in terms of $\epsilon=\|\rho_{ABC}-\mathcal{R}_{B\rightarrow BC}(\rho_{AB})\|_1=2\epsilon'$. Reverting to our notation,  
\begin{equation*}
    I(A:C|B)\leq \frac{1}{2}\epsilon |A|\ln d_s+(1+\frac{1}{2}\epsilon)h\left(\frac{\epsilon}{2+\epsilon}\right),
\end{equation*}
as presented in (\ref{qcmi-rec}), where we have substituted $\dim\mathcal{H}_A=d_s^{|A|}$.
We can also derive the simplified bound by estimating
\begin{align*}
    (1+\epsilon')h\left(\frac{\epsilon'}{1+\epsilon'}\right)&=-\epsilon'\ln\left(\frac{\epsilon'}{1+\epsilon'}\right)+\ln(1+\epsilon')\\ \nonumber
    &=-\epsilon'\ln\epsilon'+(1+\epsilon')\ln(1+\epsilon')\\ \nonumber
    &\leq \epsilon'(2-\ln\epsilon'),
\end{align*}
where we have used $0\leq\epsilon'\leq 1$ and $\ln(1+\epsilon')\leq\epsilon'$. This leads to the bound for the QCMI,
\begin{equation}\label{qcmi-1-1-norm}
    I(A:C|B)\leq \epsilon\left(\frac{1}{2}|A|\ln d_s+1-\frac{1}{2}\ln\frac{\epsilon}{2}\right),
\end{equation}
where $0< \epsilon <2$, which we use in the proof of Theorem \ref{main-result}.

\section{}\label{disk-setup}
In this appendix we extend Theorem \ref{main-result} for the configuration of the regions $A$, $B$, and $C$ depicted in Figure \ref{fig1}(b). In this case $\rho_{ABC}$ has the form 
\begin{equation*}
    \rho_{ABC}=\mathrm{tr}_M\left(V_{A_2}V_{B_2}V_CV_{B_1}V_{A_1}\sigma V_{A_1}^\dagger V_{B_1}^\dagger V^\dagger_CV^\dagger_{B_2}V^\dagger_{A_2}\right).
\end{equation*}
We introduce approximating isometries $\tilde{V}_{B_1}:\mathcal{H}_M\rightarrow\mathcal{H}_{B_1}\otimes\mathcal{H}_M$ and $\tilde{V}_{B_2}:\mathcal{H}_M\rightarrow\mathcal{H}_{B_2}\otimes\mathcal{H}_M$ in exactly the same way as in Section \ref{sec-proof}, which we use to construct the state
\begin{equation}\label{tilde-rho-disk}
    \tilde{\rho}_{ABC}=\mathrm{tr}_M\left(V_{A_2}\tilde{V}_{B_2}V_C\tilde{V}_{B_1}V_{A_1}\sigma V_{A_1}^\dagger\tilde{V}_{B_1}^\dagger V^\dagger_C\tilde{V}^\dagger_{B_2}V^\dagger_{A_2}\right).
\end{equation}
Then, repeating the calculations in (\ref{error-leq-1-1}) and (\ref{1-1-leq-opnorm}) and using the triangle inequality, we bound the recovery error for the recovery channel $\mathcal{R}_{B\rightarrow BC}:\mathcal{B}(\mathcal{H}_B)\rightarrow\mathcal{B}(\mathcal{H}_B\otimes\mathcal{H}_C)$, for which $\mathcal{R}_{B\rightarrow BC}(\tilde{\rho}_{AB})=\tilde{\rho}_{ABC}$, as
\begin{equation*} 
    \|\rho_{ABC}-\mathcal{R}_{B\rightarrow BC}\left(\rho_{AB}\right)\|_1\leq 4\|V_{B_1}-\tilde{V}_{B_1}\|+4\|V_{B_2}-\tilde{V}_{B_2}\|.
\end{equation*}
For both terms we apply the bound (\ref{isonorm-leq-exp}), which leads to 
\begin{align*}
\|\rho_{ABC}-\mathcal{R}_{B\rightarrow BC}\left(\rho_{AB}\right)\|_1&\leq 4\sqrt{d_Mc} \ \ (\nu^{|B_1|/2}+\nu^{|B_2|/2})\\ \nonumber &\leq 8\sqrt{d_Mc}\nu^{d(A,C)/2},
\end{align*}
where $d(A,C)=\min\{|B_1|,|B_2|\}$. 

We can prove that $\tilde{\rho_{ABC}}$ is a quantum Markov chain using the same reasoning as in Sections \ref{subsec-pgfcs}, \ref{subsec-erg}, and \ref{sect-fcs}. We present the proof only for the case of ergodic pgFCS of period 1. The extensions of the argument to the general case of pgFCS are analogous to the ones of  Sections \ref{subsec-erg} and \ref{sect-fcs}.

Both $\tilde{V}_{B_1}$ and $\tilde{V}_{B_2}$ have the form (\ref{tilde-iso-form}),
\begin{align}\label{isom-disk}
    \tilde{V}_{B_1}&=\sum_{i,j=1}^{d_M}\sqrt{\sigma_i} |\xi_{ij}\rangle \otimes |i\rangle\langle j|, \\ \nonumber
    \tilde{V}_{B_2}&=\sum_{m,n=1}^{d_M}\sqrt{\sigma_m} |\zeta_{mn}\rangle \otimes |m\rangle\langle n|.
\end{align}
The vectors $|\xi_{ij}\rangle$ and $|\zeta_{mn}\rangle$ define the subspaces $\mathcal{H}_{B_1}\supseteq\mathcal{H}_{b_1}:=\mathrm{span}\{|\xi_{ij}\rangle\in\mathcal{H}_{B_1}\;|\;i,j=1,\cdot\cdot\cdot,d_M\}$ and $\mathcal{H}_{B_2}\supseteq\mathcal{H}_{b_2}:=\mathrm{span}\{|\zeta_{mn}\rangle\in\mathcal{H}_{B_1}\;|\;m,n=1,\cdot\cdot\cdot,d_M\}$. We observe that we should demand $\dim\mathcal{H}_{B_1}\geq d_M^2$ and $\dim\mathcal{H}_{B_2}\geq d_M^2$ for the construction to be possible.
Substituting (\ref{isom-disk}) into (\ref{tilde-rho-disk}), we explicitly express
\begin{widetext}
\begin{align*}
    \tilde{\rho}_{ABC}=\sum_{all\, indices=1}^{d_M}&\langle j|V_{A_1}\sigma V_{A_1}^\dagger|j'\rangle\otimes\sqrt{\sigma_i} |\xi_{ij}\rangle|\xi_{i'j'}\rangle\sqrt{\sigma_{i'}}\otimes \langle n|V_C|i\rangle\rangle i'| V^\dagger_C|n'\rangle\otimes\sqrt{\sigma_m} |\zeta_{mn}\rangle\langle\zeta_{m'n'}|\sqrt{\sigma_{m'}}\otimes\mathrm{tr}_M(V_{A_2}|m\rangle\langle m'|V^\dagger_{A_2}).
\end{align*}
\end{widetext}
We introduce the isomorphism $I:\mathcal{H}_{b_1}\otimes\mathcal{H}_{b_2}\rightarrow\mathcal{H}_{M}\otimes\mathcal{H}_{M}\otimes\mathcal{H}_{M}\otimes\mathcal{H}_{M}$, defined by $I|\xi_{ij}\rangle\otimes|\zeta_{mn}\rangle=|j\rangle\otimes|i\rangle\otimes|n\rangle\otimes|m\rangle$, which maps $\tilde{\rho}_{ABC}$ to
\begin{align*}
    I\tilde{\rho}_{ABC}I^\dagger=V_{A_1}\sigma V_{A_1}^\dagger &\otimes V_C\sqrt{\sigma}|+\rangle\langle +|\sqrt{\sigma} V^\dagger_C \\ \nonumber &\otimes \mathrm{tr}_M\left(V_{A_2}\sqrt{\sigma}|+\rangle\langle +|\sqrt{\sigma}V^\dagger_{A_2}\right),
\end{align*}
where $|+\rangle=\sum_{i=1}^{d_M}|i\rangle\otimes|i\rangle\in\mathcal{H}_M\otimes\mathcal{H}_M$. By Proposition \ref{qmc-structure} this state is a quantum Markov chain.

\section{}\label{qmi-examples}
In this appendix we construct explicit examples of pgFCS, for which QMI converges to a finite limit as the size of the region $B$ is increased, $\lim_{|B|\rightarrow+\infty}I(A:C)>0$. By contrast, Theorem \ref{main-result} proves that QCMI converges to zero for any pgFCS.

To construct the examples we use two isometries. The first one is 
$V_1:\mathbb{C}^2\rightarrow\mathbb{C}^3\otimes\mathbb{C}^2$, defined by 
\begin{align}\label{v1-expl}
    V_1=&\frac{1}{\sqrt{2}}|1\rangle\otimes|-\rangle\langle+|-\frac{1}{\sqrt{2}}|0\rangle\otimes|+\rangle\langle+|\\ \nonumber &+\frac{1}{\sqrt{2}}|0\rangle\otimes|-\rangle\langle -|-\frac{1}{\sqrt{2}}|-1\rangle\otimes|+\rangle\langle -|,
\end{align}
where $\{|-1\rangle,|0\rangle,|1\rangle\}$ is an orthonormal basis in $\mathbb{C}^3$, and $\{|-\rangle,|+\rangle\}$ is an orthonormal basis in $\mathbb{C}^2$.
This isometry is of the type defined in Ref. \onlinecite{fcs}, that generates the ground state of the generalized AKLT model. We define $\mathcal{E}_1:\mathcal{B}(\mathbb{C}^2)\rightarrow\mathcal{B}(\mathbb{C}^2)$ by $\mathcal{E}_1(X)=\mathrm{tr}_3(V_1XV_1^\dagger)$, where $\mathrm{tr}_3$ is the partial trace over $\mathbb{C}^3$. For the current discussion it is important that the pair $(V_1,\mathbb{1}_2/2)$ generates an ergodic pgFCS of period 1. Hence, $\tilde{\mathcal{E}}_1$, obtained from $\mathcal{E}_1$ using Definition \ref{e-tilde-e} satisfies $\tilde{\mathcal{E}}_1(X)=\mathrm{tr}(X)\mathbb{1}_2/2$. We obtain the second isometry, $V_2:\mathbb{C}^2\rightarrow\mathbb{C}^3\otimes\mathbb{C}^2$, from $V_1$ by performing the cyclic permutation $|-1\rangle\mapsto|0\rangle\mapsto|1\rangle\mapsto|-1\rangle$, 
\begin{align}\label{v2-expl}
    V_2=&\frac{1}{\sqrt{2}}|-1\rangle\otimes|-\rangle\langle+|-\frac{1}{\sqrt{2}}|1\rangle\otimes|+\rangle\langle+|\\ \nonumber 
    &+\frac{1}{\sqrt{2}}|1\rangle\otimes|-\rangle\langle -|-\frac{1}{\sqrt{2}}|0\rangle\otimes|+\rangle\langle -|.
\end{align}
We observe $\mathrm{tr}_3(V_2XV_2^\dagger)=\mathrm{tr}_3(V_1\mathbb{1}_2V_1^\dagger)=\mathcal{E}_1(X)$ for any $X\in\mathcal{B}(\mathcal{H}_M)$.

Our first example is a pgFCS with two ergodic components. For a Hilbert space of the spin $\mathcal{H}_s:=\mathbb{C}^3$ and a memory Hilbert space $\mathcal{H}_M:=\mathbb{C}^4\cong\mathbb{C}^2\oplus\mathbb{C}^2$, the quantum state is generated by
\begin{equation*}
    V=\begin{pmatrix}
            V_1 & 0\\
            0 & V_2
        \end{pmatrix},
    \quad    
    \sigma=\frac{1}{4}\begin{pmatrix}
                        \mathbb{1}_2 & 0\\
                        0 & \mathbb{1}_2
                        \end{pmatrix}.
\end{equation*}
It is obvious that this state by construction has two ergodic components (see Proposition \ref{gen-pgfcs}). We choose the regions $A$ and $C$ to each consist of a single spin, i.e., $\mathcal{H}_A,\mathcal{H}_C=\mathbb{C}^2$. In the limit $|B|\rightarrow+\infty$, 
\begin{align*}
    \tilde{\rho}_{AC}&:=\lim_{|B|\rightarrow+\infty}\mathrm{tr}_M(V\mathrm{tr}_B(V_B V\sigma V^\dagger V_B^\dagger)V^\dagger)\\ \nonumber
    &=\mathrm{tr}_M(V\tilde{\mathcal{E}}(V\sigma V^\dagger)V^\dagger)\\ \nonumber
    &=\frac{1}{8}\mathrm{tr}_2(V_1V^\dagger_1)\otimes\mathrm{tr}_2(V_1V^\dagger_1)+\frac{1}{8}\mathrm{tr}_2(V_2V^\dagger_2)\otimes\mathrm{tr}_2(V_2V^\dagger_2),
\end{align*}
where $\mathrm{tr}_2$ is the partial trace over $\mathbb{C}^2$. In the second line $\tilde{\mathcal{E}}$ is the channel obtained from $\mathcal{E}:\mathbb{C}^4\rightarrow\mathbb{C}^4$, $\mathcal{E}(X)=\mathrm{tr}_3(VXV^\dagger)$, via Definition \ref{e-tilde-e}, which acts as 
\begin{align*}
    \tilde{\mathcal{E}}\begin{pmatrix}
                            X_{11} & X_{12}\\ 
                            X_{21} & X_{22}
                       \end{pmatrix}
                       &=
                       \begin{pmatrix}
                            \tilde{\mathcal{E}}_1(X_{11}) & 0\\ 
                            0 & \tilde{\mathcal{E}}_1(X_{22})
                       \end{pmatrix}
                       \\ \nonumber
                       &=
                       \begin{pmatrix}
                            \frac{1}{2}\mathrm{tr}_2(X_{11})\mathbb{1}_2 & 0\\ 
                            0 & \frac{1}{2}\mathrm{tr}_2(X_{22})\mathbb{1}_2
                       \end{pmatrix}.
\end{align*}
Then 
\begin{align*}
    \tilde{\rho}_{A}\otimes\tilde{\rho}_{C}=\rho_{A}\otimes\rho_{C}=\frac{1}{16}&\left(\mathrm{tr}_2(V_1V^\dagger_1)+\mathrm{tr}_2(V_2V^\dagger_2)\right)\\ \nonumber 
    &\otimes \left(\mathrm{tr}_2(V_1V^\dagger_1)+\mathrm{tr}_2(V_2V^\dagger_2)\right).
\end{align*}
Using the explicit expressions (\ref{v1-expl}) and (\ref{v2-expl}), we calculate
\begin{align*}
    \tilde{\rho}_{AC}=&\frac{1}{32}\begin{pmatrix}
                            1 & 0 & 0\\
                            0 & 2 & 0\\
                            0 & 0 & 1
                          \end{pmatrix}
                          \otimes
                          \begin{pmatrix}
                            1 & 0 & 0\\
                            0 & 2 & 0\\
                            0 & 0 & 1
                          \end{pmatrix} \\ \nonumber
                          &+
                          \frac{1}{32}
                          \begin{pmatrix}
                            1 & 0 & 0\\
                            0 & 1 & 0\\
                            0 & 0 & 2
                          \end{pmatrix}
                          \otimes
                          \begin{pmatrix}
                            1 & 0 & 0\\
                            0 & 1 & 0\\
                            0 & 0 & 2
                          \end{pmatrix}, \\ \nonumber
    \tilde{\rho}_A\otimes\tilde{\rho}_C=&\frac{1}{64}\begin{pmatrix}
                            2 & 0 & 0\\
                            0 & 3 & 0\\
                            0 & 0 & 3
                          \end{pmatrix}
                          \otimes
                          \begin{pmatrix}
                            2 & 0 & 0\\
                            0 & 3 & 0\\
                            0 & 0 & 3
                          \end{pmatrix}.
\end{align*}
We observe that $\tilde{\rho}_{AC}\neq\tilde{\rho}_A\otimes\tilde{\rho}_C$, which implies that $I(A:C)\neq 0$. We verify this by explicitly calculating $I(A:C)={17}\ln 2/{16} -9 \ln 3/8+5 \ln 5/16\approx 0.0035$.

Now we construct an ergodic pgFCS of period $2$ for which QMI does not converge to zero. We use the same setup as above, changing only the elementary isometry generating the state to
\begin{equation*}
    V=\begin{pmatrix}
            0 & V_1\\
            V_2 & 0
        \end{pmatrix},
\end{equation*}
where $V_1$ and $V_2$ are the same isometries (\ref{v1-expl}) and (\ref{v2-expl}), respectively. One can check that the only eigenvalues of absolute value $1$ of $\mathcal{E}(X)=\mathrm{tr}_3(VXV^\dagger)$ are $1$ and $-1$, corresponding to $\left(\begin{smallmatrix}\mathbb{1}_2/4 & 0\\ 0 & \mathbb{1}_2/4 \end{smallmatrix}\right)$ and $\left(\begin{smallmatrix}\mathbb{1}_2/4 & 0\\ 0 & -\mathbb{1}_2/4 \end{smallmatrix}\right)$, respectively. Hence the generated state is indeed an ergodic pgFCS with period 2. To determine the asymptotic behavior of $\rho_{AC}$ we explicitly express it in terms of $V_1$ and $V_2$,
\begin{equation*}
    \rho_{ABC}=
    \begin{cases}
        \begin{aligned}[t]
        &\frac{1}{4}\mathrm{tr}_2(V_1V_2\cdot\cdot\cdot V_2V_1 \frac{\mathbb{1}_2}{2}V^\dagger_1V^\dagger_2\cdot\cdot\cdot V^\dagger_2V^\dagger_1) \\
        &+\frac{1}{4}\mathrm{tr}_2(V_2V_1\cdot\cdot\cdot V_1V_2 \frac{\mathbb{1}_2}{2}V^\dagger_2V^\dagger_1\cdot\cdot\cdot V^\dagger_1V^\dagger_2),\\ &\quad\quad\quad \text{if} \ \ |B|\,\mbox{mod}\, 2 \equiv 1,
        \end{aligned}   \\ 
        \begin{aligned}[t]
        &\frac{1}{4}\mathrm{tr}_2(V_2V_1\cdot\cdot\cdot V_2V_1 \frac{\mathbb{1}_2}{2}V^\dagger_1V^\dagger_2\cdot\cdot\cdot V^\dagger_1V^\dagger_2)\\
        &+\frac{1}{4}\mathrm{tr}_2(V_1V_2\cdot\cdot\cdot V_1V_2 \frac{\mathbb{1}_2}{2}V^\dagger_2V^\dagger_1\cdot\cdot\cdot V^\dagger_2V^\dagger_1),\\  
        &\quad\quad\quad \text{if} \ \ |B|\,\mbox{mod}\, 2 \equiv 0.  \end{aligned} 
    \end{cases}
\end{equation*}
Since $V_2=U_sV_1$ for a unitary $U_s:\mathbb{C}^3\rightarrow \mathbb{C}^3$, then $\mathrm{tr}_3(V_1 X V_1^\dagger)=\mathrm{tr}_3(V_2 X V^\dagger_2)$ for any $X\in\mathcal{B}(\mathcal{H}_M)$. Notice that for any choice of the indices $\{i_k\in\{1,2\}\;|\;k=1,2,\cdot\cdot\cdot,|B|\}$, we have $\lim_{|B|\rightarrow +\infty}\mathrm{tr}_B({V_{i_{|B|}}}\cdot\cdot\cdot V_{i_2}V_{i_1}X V_{i_1}^\dagger V_{i_2}^\dagger\cdot\cdot\cdot V^\dagger_{i_{|B|}})=\lim_{|B|\rightarrow +\infty}\mathcal{E}_1^{|B|}(X)=\tilde{\mathcal{E}}_1^{|B|}(X)=\mathrm{tr}(X)\mathbb{1}_2/2$. This implies that in the asymptotic limit $|B|\rightarrow +\infty$, the state $\rho_{AC}$ oscillates between two states, 
\begin{align*}
    \tilde{\rho}^{(1)}_{AC}&=\frac{1}{8}\mathrm{tr}_2(V_1V^\dagger_1)\otimes\mathrm{tr}_2(V_1V^\dagger_1)+\frac{1}{8}\mathrm{tr}_2(V_2V^\dagger_2)\otimes\mathrm{tr}_2(V_2V^\dagger_2), \\ \nonumber
    \tilde{\rho}^{(2)}_{AC}&=\frac{1}{8}\mathrm{tr}_2(V_1V^\dagger_1)\otimes\mathrm{tr}_2(V_2V^\dagger_2)+\frac{1}{8}\mathrm{tr}_2(V_2V^\dagger_2)\otimes\mathrm{tr}_2(V_1V^\dagger_1).
\end{align*}
For both of these states,
\begin{align*}
    &\tilde{\rho}^{(1)}_{A}\otimes\tilde{\rho}^{(1)}_{C}=\tilde{\rho}^{(2)}_{A}\otimes\tilde{\rho}^{(2)}_{C} \\ \nonumber
    &=\frac{1}{16}\left(\mathrm{tr}_2(V_1V^\dagger_1)+\mathrm{tr}_2(V_2V^\dagger_2)\right)\otimes \left(\mathrm{tr}_2(V_1V^\dagger_1)+\mathrm{tr}_2(V_2V^\dagger_2)\right).
\end{align*}
Using the explicit expressions (\ref{v1-expl}) and (\ref{v2-expl}), we obtain 
\begin{align*}
    &\tilde{\rho}^{(1)}_{AC}=\frac{1}{32}\begin{pmatrix}
                            1 & 0 & 0\\
                            0 & 2 & 0\\
                            0 & 0 & 1
                          \end{pmatrix}
                          \otimes
                          \begin{pmatrix}
                            1 & 0 & 0\\
                            0 & 2 & 0\\
                            0 & 0 & 1
                          \end{pmatrix}\\ \nonumber
                          &\quad\quad\quad +
                          \frac{1}{32}
                          \begin{pmatrix}
                            1 & 0 & 0\\
                            0 & 1 & 0\\
                            0 & 0 & 2
                          \end{pmatrix}
                          \otimes
                          \begin{pmatrix}
                            1 & 0 & 0\\
                            0 & 1 & 0\\
                            0 & 0 & 2
                          \end{pmatrix}, \\ \nonumber
    &\tilde{\rho}^{(2)}_{AC}=\frac{1}{32}\begin{pmatrix}
                            1 & 0 & 0\\
                            0 & 2 & 0\\
                            0 & 0 & 1
                          \end{pmatrix}
                          \otimes
                          \begin{pmatrix}
                            1 & 0 & 0\\
                            0 & 1 & 0\\
                            0 & 0 & 2
                          \end{pmatrix}\\ \nonumber
                          &\quad\quad\quad+
                          \frac{1}{32}
                          \begin{pmatrix}
                            1 & 0 & 0\\
                            0 & 1 & 0\\
                            0 & 0 & 2
                          \end{pmatrix}
                          \otimes
                          \begin{pmatrix}
                            1 & 0 & 0\\
                            0 & 2 & 0\\
                            0 & 0 & 1
                          \end{pmatrix}, \\ \nonumber
    &\tilde{\rho}^{(1)}_{A}\otimes\tilde{\rho}^{(1)}_{C}=\tilde{\rho}^{(2)}_{A}\otimes\tilde{\rho}^{(2)}_{C}=\frac{1}{64}
                        \begin{pmatrix}
                            2 & 0 & 0\\
                            0 & 3 & 0\\
                            0 & 0 & 3
                          \end{pmatrix}
                          \otimes
                          \begin{pmatrix}
                            2 & 0 & 0\\
                            0 & 3 & 0\\
                            0 & 0 & 3
                          \end{pmatrix}.
\end{align*}
An explicit calculation shows that $\tilde{\rho}^{(1)}_{AC}$ and $\tilde{\rho}^{(2)}_{AC}$ have the same QMI, calculated in the previous example, $I(A:C)={17}\ln 2/{16} -9\ln 3/8+5\ln 5/16\approx 0.0035$. Thus, while for this ergodic pgFCS with period 2 QMI actually converges, it still does not converge to zero.

\section{}\label{app-b}
In this appendix we present a pair of lemmata for completeness.

The first lemma is a restatement of the Lemma A.2 from Ref.\onlinecite{mpdsrenorm} in a slightly different language.
\begin{lemma}[Lemma A.2 of Ref. \onlinecite{mpdsrenorm}]\label{b-1}
Let isometries $V_1:\mathcal{H}_M\rightarrow\mathcal{H}_s\otimes\mathcal{H}_M$ and $V_{2}:\mathcal{H}_M\rightarrow\mathcal{H}_s\otimes\mathcal{H}_M$ be such that full-rank density operators $\sigma_1, \sigma_2\in\mathcal{B}(\mathcal{H}_M)$ are the only fixed points of the quantum channels $X\mapsto\mathrm{tr}_s(V_1 X V^{\dagger}_1)$ and $X\mapsto\mathrm{tr}_s(V_2X V^{\dagger}_2)$, $X\in\mathcal{B}(\mathcal{H}_M)$, respectively, and that these quantum channels have no other eigenvalues of magnitude 1. Then the eigenvalues $\nu_i$ of the map $X\mapsto\mathrm{tr}_s(V_1 X V^{\dagger}_2)$ are such that $|\nu_i|\leq 1.$ Furthermore, if there exists the eigenvalue $\nu_i$ of magnitude $1$, then $V_1=e^{-i\phi}(\mathbb{1}_s\otimes W^\dagger)V_2W$ for some $\phi\in\mathbb{R}$ and unitary $W:\mathcal{H}_M\rightarrow\mathcal{H}_M$.
\end{lemma}
\begin{proof}
Let $\nu_i$ be an eigenvalue, so that for some $X$:
\begin{equation*}
    \mbox{tr}_s(V_{1}X V^{\dagger}_{2})=\nu_i X.
\end{equation*}
Using the fact that $\sigma_1$ is invertible, consider:
\begin{align}\label{eigenvleq1}
    &\left|\nu_i\mbox{tr}(X^\dagger \sigma_1^{-1}X)\right|^2=\left|\mbox{tr}(X^\dagger \sigma_1^{-1} \mbox{tr}_s(V_1X V^{\dagger}_2))\right|^2\\ \nonumber
    &=\left|\mbox{tr}((\mathbb{1}_s\otimes X^\dagger \sigma_1^{-1}) V_1{\sigma_1}^{1/2}\sigma_1^{-1/2}X V^{\dagger}_2 )\right|^2 \\ \nonumber
    &\leq\mbox{tr}(X^\dagger\sigma_1^{-1}\mbox{tr}_s(V_1{\sigma_1}V^{\dagger}_1)\sigma_1^{-1}X)\mbox{tr}(V_2 X^\dagger \sigma_1^{-1}X V^{\dagger}_2 )\\ \nonumber
    &\leq \mbox{tr}(X^\dagger\sigma_1^{-1} X)^2,
\end{align}
where we have used the Cauchy–Schwarz inequality for the Hilbert-Schmidt inner product $|\mathrm{tr}(A^\dagger B)|^2\leq\mathrm{tr}(A^\dagger A)\mathrm{tr}(B^\dagger B)$ to get from the second line to the third. Thus $|\nu_i|\leq 1$. Moreover $|\nu_i|=1$ only in the case of equality in the Cauchy–Schwarz inequality, which implies: 
\begin{equation*}
    \sigma_1^{-1/2}X V^{\dagger}_2=e^{i\phi}{\sigma_1}^{1/2} V^{\dagger}_1(\mathbb{1}_s\otimes \sigma_1^{-1}X),
\end{equation*}
where $\phi\in\mathbb{R}$.
Using that $\sigma_1^{-1/2}$ is invertible and denoting $W=\sigma_1^{-1}X$, we get:
\begin{equation*}
    W V^{\dagger}_2=e^{i\phi}V^{\dagger}_1(\mathbb{1}_s\otimes W).
\end{equation*}
Using that $V_2^\dagger V_2=\mathbb{1}_M$, we obtain:
\begin{equation*}
    WW^\dagger=V_1^\dagger(\mathbb{1}_s\otimes WW^\dagger)V_1.
\end{equation*}
This means that $WW^\dagger$ is the eigenvector with eigenvalue $1$ of the map $X\mapsto V_1^\dagger(\mathbb{1}_s \otimes X)V_1$. The latter map is adjoint to the quantum channel $X\mapsto\mathrm{tr}_s(V_1XV_1^\dagger),$ with respect to the Hilbert-Schmidt inner product, therefore it has the same spectrum. Hence the adjoint map has only one eigenvector with eigenvalue $1$.  Since $V_1$ is an isometry, $V_1^\dagger(\mathbb{1}_s\otimes \mathbb{1}_M)V_1=\mathbb{1}_M$, the identity operator $\mathbb{1}_M$ is this eigenvector, and $W^\dagger W=\mathbb{1}_M$. Replacing $\mathrm{tr}(X^\dagger\sigma_1^{-1}X)$ with  $\mathrm{tr}(X\sigma_2^{-1}X^\dagger)$ in (\ref{eigenvleq1}) leads through the same reasoning to $W W^\dagger=\mathbb{1}_M$. Thus $W$ is unitary and
\begin{equation*}
    V_2=e^{-i\phi}(\mathbb{1}_s\otimes W^\dagger)V_1W.
\end{equation*}
\end{proof}

The second lemma is a restatement of the result used in Ref.\onlinecite{fcs}, the formal proof of which can be found in Ref.\onlinecite{wolf-szehr}. We present the proof using our notation.
\begin{lemma}[Theorem III.2 of Ref. \onlinecite{wolf-szehr}, Ref. \onlinecite{fcs}]\label{b-2-app}
Let $\mathcal{E}:\mathcal{B}(\mathcal{K})\rightarrow\mathcal{B}(\mathcal{K})$ be a map with the spectral radius $1$, and let $\tilde{\mathcal{E}}$ be the map obtained from $\mathcal{E}$ as described in Definition \ref{e-tilde-e}. Then for any $\nu\in\mathbb{R}$, such that $\nu_{\mathrm{gap}}<\nu<1$, where $\nu_{\mathrm{gap}}:=\max_{|\nu_i|<1}|\nu_i|$, there exist the constant $c>0$, depending on $\nu$, such that
$$\|\mathcal{E}^n-\tilde{\mathcal{E}}^n\|_{2\rightarrow 2}\leq c\nu^n.$$
\end{lemma}
\begin{proof}
Following the procedure described in Appendix \ref{app-vec}, we represent the maps $\mathcal{E}$ and $\tilde{\mathcal{E}}$ with the operators $E$ and $\tilde{E}$, respectively, that inherit the spectra of the maps. 

As is shown in Appendix \ref{app-vec},
\begin{equation}\label{app-b-norm}
    \|\mathcal{E}^n-\tilde{\mathcal{E}}^n\|_{2\rightarrow 2}=\|E^n-\tilde{E}^n\|.
\end{equation}
Then to prove the statement of the lemma we can estimate $\|E^n-\tilde{E}^n\|$.

We make the Jordan decomposition of $E$,
\begin{equation*}
    E=\sum_{i}(\nu_iP_{\nu_i}+N_{\nu_i}),
\end{equation*}
where $\nu_i$ is an eigenvalue, $P_{\nu_i}$ is the projector onto the subspace corresponding to the eigenvalue $\nu_i$, and $N_{\nu_i}$ is a nilpotent operator with index $(K_i+1)$, $N_{\nu_i}^{K_i+1}=0$. The Jordan decomposition for $\tilde{E}$ can be obtained by setting all $\nu_i$, except for the ones of magnitude 1, to 0,
\begin{equation*}
    \tilde{E}=\sum_{|\nu_i|=1}(\nu_iP_{\nu_i}+N_{\nu_i}),
\end{equation*}

From the properties of the operators $P_{\nu_i}$ and $N_{\nu_i}$ in the Jordan decomposition, $P_{\nu_i}P_{\nu_j}=\delta_{ij}P_{\nu_i}$, $P_{\nu_i}N_{\nu_j}=N_{\nu_j}P_{\nu_i}=\delta_{ij}N_{\nu_i}$, and $N_{\nu_i}N_{\nu_k}=\delta_{ik}N_{\nu_i}^2$, it follows that we can express $E^n$ as 
\begin{equation*}
    E^n=\sum_{i} ({\nu_i}P_{\nu_i}+N_{\nu_i})^n.
\end{equation*}

For all peripheral eigenvalues $\nu_i$, \emph{i.e.} such that $|\nu_i|=1$, $N_{\nu_i}=0$ (see, for example, the proof of Proposition 3.1 in Ref.\onlinecite{fcs}). Thus, we can write the expression for $E^n$ as
\begin{equation*}
    E^n=\sum_{|\nu_i|=1}\nu^n_i P_{\nu_i}+\sum_{|{\nu_i}|<1} ({\nu_i}P_{\nu_i}+N_{\nu_i})^n.
\end{equation*}
It immediately follows that 
\begin{equation*}
    \tilde{E}^n=\sum_{|\nu_i|=1}\nu^n_i P_{\nu_i}.
\end{equation*}
Then we can estimate
\begin{align}\label{app-b-n}
    &\|E^n-\tilde{E}^n\|^2=\|\sum_{|\nu_i|<1}({\nu_i} P_{\nu_i}+N_{\nu_i})^{n}\|^2\\ \nonumber
    &=\sup_{\| |\psi \rangle \|=1}\langle\psi|\sum_{|\nu_i|<1}(\nu_i^* P_{\nu_i}+N_{\nu_i}^\dagger)^{n}({\nu_i} P_{\nu_i}+N_{\nu_i})^{n}|\psi\rangle \\ \nonumber
    &=\sum_{|\nu_i|<1}\sup_{ \| |\psi \rangle \|=1} \langle\psi|P_{\nu_i}(\nu_i^* P_{\nu_i}+N_{\nu_i}^\dagger)^{n}({\nu_i} P_{\nu_i}+N_{\nu_i})^{n}P_{\nu_i}|\psi\rangle\\ \nonumber
    &\leq\sum_{|\nu_i|<1}\|(\nu_i P_{\nu_i}+N_{\nu_i})^n\|^2 \sup_{\| |\psi \rangle \|=1}\langle\psi|P_{\nu_i}^2|\psi\rangle \\ \nonumber
    &\leq \max_{0<|\nu_i|<1}\|(\nu_i P_{\nu_i}+N_{\nu_i})^n\|^2 \sup_{ \| |\psi \rangle \|=1} \langle\psi|\sum_{|\nu_i|<1}P_{\nu_i}|\psi\rangle\\ \nonumber
    &\leq \max_{0<|\nu_i|<1}\|(\nu_i P_{\nu_i}+N_{\nu_i})^n\|^2 \sup_{\| |\psi \rangle \|=1} \langle\psi|\sum_{\nu_i}P_{\nu_i}|\psi\rangle\\ \nonumber
    &\leq \max_{0<|\nu_i|<1}\|(\nu_i P_{\nu_i}+N_{\nu_i})^{n}\|^2.
\end{align}
We have used the definition of the operator norm and orthogonality of the terms $\nu_iP_{\nu_i}+N_{\nu_i}$ to get the second line. The third line follows from the identity $\nu_iP_{\nu_i}+N_{\nu_i}=(\nu_iP_{\nu_i}+N_{\nu_i})P_{\nu_i},$ and the bound $\langle\psi|A^\dagger A|\psi\rangle=\|A|\psi\rangle\|^2\leq\|A\|^2\langle\psi|\psi\rangle$ is used in the fourth line. The sixth and seventh lines follow from properties of the orthogonal projectors $\sum_{|\nu_i|<1}P_{\nu_i} \le \sum_{\nu_i}P_{\nu_i}=\mathbb{1}.$ 

Now, using the fact that $N_{\nu_i}$ is nilpotent with the index $K_i+1$, we estimate
\begin{align*}
    \|(\nu_i P_{\nu_i}+N_{\nu_i})^{n}\|&=\|P_{\nu_i}\sum_{k=0}^{K_i} \binom{n}{k}\nu_i^{n-k}N_{\nu_i}^k\|\\ \nonumber
    &\leq|\nu_i|^n\sum_{k=0}^{K_i} \binom{n}{k} \frac{\|N_{\nu_i}\|^k}{|\nu_i|^k}
    \\ \nonumber
    &\leq |\nu_i|^n\sum_{k=0}^{K_i} \frac{n^k}{k!}\frac{\|N_{\nu_i}\|^k}{|\nu_i|^k}\\ \nonumber
    &\leq |\nu_i|^n\sum_{k=0}^{K_i} \binom{K_i}{k}\left(\frac{n\|N_{\nu_i}\|}{|\nu_i|}\right)^k\\ \nonumber
    &= |\nu_i|^n\left(\frac{n\|N_{\nu_i}\|}{|\nu_i|}+1\right)^{K_i},
\end{align*}
where we have used the triangle inequality and the inequality $\binom{n}{k}\leq \frac{n^k}{k!}\leq\binom{K_i}{k}n^k$ for $k\leq K_i$. Denoting $c_1:=\max_{0<|\nu_i|<1}\left\{\|N_{\nu_i}\|/|\nu_i|)\right\}$, and $K_{\mathrm{max}}=\max_{0<|\nu_i|<1}K_i$, we bound
\begin{align}\label{app-b-max}
    \max_{0<|\nu_i|<1}\|(\nu_i P_{\nu_i}+N_{\nu_i})^{n}\|\leq\nu_{\mathrm{gap}}^n \left(c_1n+1\right)^{K_{\mathrm{max}}}.
\end{align}
Combining (\ref{app-b-norm}), (\ref{app-b-n}), and (\ref{app-b-max}) together,
\begin{equation*}
\|\mathcal{E}^n-\tilde{\mathcal{E}}^n\|_{2\rightarrow 2} \leq \nu_{\mathrm{gap}}^n \left(c_1n+1\right)^{K_{\mathrm{max}}},
\end{equation*}
manifesting an almost exponential decay. Using a slightly weaker bound, we can make the decay to be strictly exponential. We choose some $\nu_{\mathrm{gap}}<\nu<1$, expand
\begin{equation*}
    \nu_{\mathrm{gap}}^n \left(c_1n+1\right)^{K_{\mathrm{max}}}=\left(\frac{\nu_{\mathrm{gap}}}{\nu}\right)^n \left(c_1n+1\right)^{K_{\mathrm{max}}}\nu^n,
\end{equation*}
and notice that first two factors can be bounded by a constant $c$, depending on $\nu$ 
\begin{equation*}
    \left(\frac{\nu_{\mathrm{gap}}}{\nu}\right)^n \left(c_1n+1\right)^{K_{\mathrm{max}}}\leq c,
\end{equation*}
which leads to the stated estimate,
\begin{equation*}
\|\mathcal{E}^n-\tilde{\mathcal{E}}^n\|_{2\rightarrow 2} \leq c\nu^n .
\end{equation*}

\end{proof}

\section{}\label{app-vec}
In this appendix we review the matrix representation of a quantum channel.\cite{watrouse}

The map $\mathcal{E}:\mathcal{B}(\mathcal{H}_M)\rightarrow\mathcal{B}(\mathcal{H}_M)$ can be represented by the operator $E:\mathcal{H}_M\otimes\mathcal{H}_M\rightarrow\mathcal{H}_M\otimes\mathcal{H}_M$ through the process of vectorization.\cite{watrouse} For the chosen orthonormal basis $\{|i\rangle : i=1\cdot \cdot \cdot d_M \}$ of $\mathcal{H}_M$ the invertible map, $\mathrm{vec}:\mathcal{B}(\mathcal{H}_M)\rightarrow\mathcal{H}_M\otimes\mathcal{H}_M$, is defined through its action on the operators $|i\rangle\langle j|$, which constitute an orthonormal basis of $\mathcal{B}(\mathcal{H}_M)$ (with respect to the Hilbert-Schmidt inner product),
\begin{equation}\label{map-vec}
    \mbox{vec}(|i\rangle\langle j|)=|i\rangle\otimes|j\rangle, \ \ i,j=1\cdot \cdot \cdot d_M.
\end{equation}
The inner product on the Hilbert space $\mathcal{H}_M\otimes\mathcal{H}_M$ may then be related to the Hilbert-Schmidt inner product on the Hilbert space $\mathcal{B}(\mathcal{H}_M)$, by
\begin{equation}\label{a-1}
    \langle \mbox{vec} (X)|\mbox{vec} (Y)\rangle_{\mathcal{H}_M\otimes\mathcal{H}_M} =\mbox{tr}(X^{\dagger} Y) = \langle X, Y \rangle_{\mathcal{B}(\mathcal{H}_M)},
\end{equation}
for $X, Y \in\mathcal{B}(\mathcal{H}_M).$
It follows immediately that 
\begin{equation}\label{a-2}
    \|X\|_2=\|\mathrm{vec}(X)\|,
\end{equation}
where $\||\Psi\rangle\|=\langle\Psi|\Psi\rangle^{\frac{1}{2}}_{\mathcal{H}_M\otimes\mathcal{H}_M}$ is the norm of the vector.

Define the operator $E:\mathcal{H}_M\otimes\mathcal{H}_M\rightarrow\mathcal{H}_M\otimes\mathcal{H}_M$ by
\begin{equation}\label{map-to-op}
    E :=\mbox{vec}\circ\mathcal{E}\circ\mbox{vec}^{-1},
\end{equation}
leading to the explicit form of the matrix elements of $E$,
\begin{equation*}
    \langle m|\otimes\langle n|E|i\rangle\otimes|j\rangle=\langle m| \mathcal{E}(|i\rangle\langle j|) |n \rangle, \quad i,j,m,n=1\cdot \cdot \cdot d_M.    
\end{equation*}

The relations (\ref{a-1}) and (\ref{a-2}) imply that the $2\rightarrow 2$ norm of the map $\mathcal{E}$ can be expressed through the operator norm of $E$,
\begin{align*}
    \|\mathcal{E}\|_{2\rightarrow 2}& :=\sup_{\|X\|_2=1}\|\mathcal{E}(X)\|_2\\ \nonumber
    &=\sup_{\|\mbox{vec}(X)\|=1}\|\mbox{vec}\circ\mathcal{E}\circ\mbox{vec}^{-1}\circ\mbox{vec}(X)\|\\ \nonumber
    &=\sup_{\| |x \rangle\|=1}\|E |x \rangle\|\\ \nonumber
    &=\|E\|.
\end{align*}

The spectra of the map $\mathcal{E}$ and the operator $E$ coincide: if $\mathcal{E}(X)=\lambda X$ for some $X\in\mathcal{B}(\mathcal{H}_M)$, then $E\mathrm{vec}(X)=\lambda\cdot\mathrm{vec}(X)$. Conversely, if $E|\psi\rangle=\lambda|\psi\rangle$ for some $|\psi\rangle\in\mathcal{H}_M\otimes\mathcal{H}_M$, then $\mathcal{E}\left(\mathrm{vec}^{-1}(|\psi\rangle)\right)=\lambda\cdot\mathrm{vec}^{-1}(|\psi\rangle)$.

\section{}\label{app-d}
In this appendix we show that without loss of generality, we can assume that the decomposition (\ref{conv-erg}) of a pgFCS into the sum of ergodic components does not contain terms that are essentially equal in the sense described below.

Let us suppose there exists $\phi \in \mathbb R$ and a unitary $W:\mathcal{H}_M\rightarrow\mathcal{H}_M$, such that the isometries $V_{\alpha}$ and $V_{\beta}$ satisfy
\begin{equation}\label{c-2}
V_{\beta}=e^{-i\phi}(\mathbb{1}_{s}\otimes W^\dagger)V_{\alpha}W.
\end{equation}
This condition defines an equivalence relation and the isometries $V_{\alpha}$ and $V_{\beta}$ are then said to be equivalent.

In this appendix we show that the general convex sum representation of the FCS may be written as
\begin{equation}\label{c-1}
    \rho_{ABC} = \sum_{\alpha=1}^{J}\lambda_{\alpha} \rho_{ABC}^{\alpha},
\end{equation}
where the $\rho_{ABC}^\alpha$ define a collection of $J$ ergodic pgFCS constructed as in (\ref{pgfcs-form}), with a collection of inequivalent isometries $V_\alpha$ and full-rank density operators $\sigma_\alpha.$

Suppose the converse is true, then
\begin{align}\label{c-3}
    \rho^\beta_{ABC}&=\mbox{tr}_M \left(V^\beta_C V^\beta_B V^\beta_A \sigma_\beta V_A^{\beta\dagger} V_B^{\beta\dagger} V_C^{\beta\dagger} \right)\\ \nonumber
    &=\mbox{tr}_M \left((\mathbb{1}_{ABC}\otimes W^\dagger)V^\alpha_C V^\alpha_B V^\alpha_A W\sigma_\beta W^\dagger\times\right.\\ \nonumber 
    &\qquad\qquad\qquad\qquad
    \left. V_A^{\alpha\dagger} V_B^{\alpha\dagger} V_C^{\alpha\dagger} (\mathbb{1}_{ABC}\otimes W)\right)\\ \nonumber
    &=\mbox{tr}_M \left(V^\alpha_C V^\alpha_B V^\alpha_A W\sigma_\beta W^\dagger V_A^{\alpha\dagger} V_B^{\alpha\dagger} V_C^{\alpha\dagger} \right).
\end{align}
Now we show that the density operator $W\sigma_\beta W^\dagger$ is a fixed point of the quantum channel $X\mapsto \mathrm{tr}_{s}(V_\alpha X V^\dagger_\alpha)$,
\begin{align*}
    \mbox{tr}_{s}(V_\alpha W\sigma_\beta W^\dagger V^\dagger_\alpha)&=W\mbox{tr}_{s}\left((\mathbb{1}\otimes W^\dagger)V_\alpha W\sigma_\beta W^\dagger\times\right.\\ \nonumber
    &\qquad\qquad\qquad
    \left. V^\dagger_\alpha(\mathbb{1}\otimes W)\right)W^\dagger \\ \nonumber
    &=W\mbox{tr}_{s}(V_\beta\sigma_\beta V^\dagger_\beta))W^\dagger \\ \nonumber
    &=W\sigma_\beta W^\dagger,
\end{align*}
where we have used the fact that $\sigma_\beta$ is a fixed point of the map $X\mapsto \mathrm{tr}_{s}(V_\beta X V^\dagger_\beta)$. Thus, $W\sigma_\beta W^\dagger$ is the fixed point of the quantum channel $X\mapsto \mathrm{tr}_s(V_\alpha X V^\dagger_\alpha)$. Since $\rho_{ABC}^\alpha$ is ergodic, the density operator $\sigma_\alpha$ is the only fixed point of this map. Then $W \sigma_\beta W^\dagger=\sigma_\alpha$. Combining this fact with (\ref{c-3}), we obtain that $\rho^\alpha_{ABC}=\rho^\beta_{ABC}$. Thus, $\lambda_\alpha\rho^\alpha_{ABC}+\lambda_\beta\rho^\beta_{ABC}=(\lambda_\alpha+\lambda_\beta)\rho^\alpha_{ABC}$. We can combine in the convex sum (\ref{c-1}) all the terms satisfying the relation (\ref{c-2}) together.

\section{}\label{app-c}
In this appendix we present a technical result, required in the proof of Theorem \ref{main-result} in Section \ref{subsec-erg},
$$V_A\sigma V^\dagger_A\in \mathcal{B}(\mathcal{H}_A)\otimes\bigoplus_{k=0}^{p-1} \Pi_k\mathcal{B}(\mathcal{H}_M)\Pi_k.$$
We closely follow the reasoning of the proof of Proposition 3.3 of Ref. \onlinecite{fcs}.
Here $V_A:\mathcal{H}_M\rightarrow\mathcal{H}_A\otimes\mathcal{H}_M$ is the isometry generating the spins in the region $A$ as in (\ref{pgfcs-form}), and  $\sigma\in\mathcal{D}(\mathcal{H}_M)$ and the projectors $\{\Pi_k\in\mathcal{B}(\mathcal{H}_M)\;|\;k=0,\cdot\cdot\cdot,p-1\}$, $\sum_{k=0}^{p-1}\Pi_k=\mathbb{1}_M$, are defined in Proposition \ref{ergpgfc}. 

Let us express $$V_A\sigma V_A^\dagger=\sum_{n,m=1}^{d_s^{|A|}}|\phi_n\rangle\langle \phi_m|\otimes\langle\phi_n|V_A\sigma V_A^\dagger|\phi_m\rangle,$$ where $\{|\phi_n\rangle\in\mathcal{H}_A|\;n=1,\cdot\cdot\cdot,d_s^{|A|}\}$ is an orthonormal basis. It is clear that $$\langle\phi_n|V_A\sigma V_A^\dagger|\phi_m\rangle\in\mathrm{span}\{\langle \zeta|V_A\sigma V^\dagger_A|\psi\rangle \;|\; |\zeta\rangle,|\psi\rangle\in\mathcal{H}_A\}.$$ Notice that $\langle \zeta|V_A\sigma V^\dagger_A|\psi\rangle=\mathrm{tr}_A(|\psi\rangle\langle\zeta|V_A\sigma V^\dagger_A)$ and the vector space $\mathrm{span}\{|\psi\rangle\langle\zeta|\;|\; |\zeta\rangle,|\psi\rangle\in\mathcal{H}_A\}=\mathcal{B}(\mathcal{H}_A)$ can be spanned by positive semi-definite operators $O_A\in\mathcal{B}(\mathcal{H}_A)$. Hence, 
\begin{align*}
    \mathcal{B}(\mathcal{H}_M)=\mathrm{span}\{\mathrm{tr}_A(O_AV_A\sigma V^\dagger_A)\;|\; O_A\geq 0,\, O_A \in\mathcal{B}(\mathcal{H}_A)\}.
\end{align*} 
Since any $O_A\geq 0$ satisfies $O_A\leq\|O_A\|\mathbb{1}_A$, it follows that 
$$\mathrm{tr}_A(O_AV_A\sigma V_A^\dagger)\leq\|O_A\|\mathrm{tr}_A(V_A\sigma V^\dagger_A)=\|O_A\|\sigma,$$ where we have used $\mathrm{tr}_A(V_A \sigma V^\dagger_A)=\sigma$. The operators $\mathrm{tr}_A(O_AV_A\sigma V_A^\dagger)$ and $\|O_A\|\sigma$ can be ordered in this way only if $\mathrm{tr}_A(O_AV_A\sigma V_A^\dagger)\in\bigoplus_{k=0}^{p-1}\Pi_k\mathcal{B}(\mathcal{H}_M)\Pi_k$, since $\sigma\in\bigoplus_{k=0}^{p-1}\Pi_k\mathcal{B}(\mathcal{H}_M)\Pi_k$. Then 
\begin{align*}
    \mathrm{span}\{\mathrm{tr}_A(O_AV_A\sigma V^\dagger_A)\;|\; O_A\geq 0,\, O_A\in\mathcal{B}(\mathcal{H}_A)\}\\ \in\bigoplus_{r=1}^p\Pi_k\mathcal{B}(\mathcal{H}_M)\Pi_k,
\end{align*}from which it follows that $\langle\zeta|V_A\sigma V^\dagger_A|\psi\rangle\in \bigoplus_{k=0}^{p-1} \Pi_k\mathcal{B}(\mathcal{H}_M)\Pi_k$, and $$V_A\sigma V^\dagger_A\in \mathcal{B}(\mathcal{H}_A)\otimes\bigoplus_{k=0}^{p-1} \Pi_k\mathcal{B}(\mathcal{H}_M)\Pi_k.$$ This leads to the identity $V_A\sigma V_A=\sum_{k=0}^{p-1}\Pi_k V_A\sigma V_A^\dagger \Pi_k$. 

\section{}\label{app-erg-e-tilde}
In this appendix we prove that the quantum channel $\tilde{\mathcal{E}}$ obtained from $\mathcal{E}$ by Definition \ref{e-tilde-e} of Section \ref{subsec-quant-chan}, and used in Section \ref{subsec-erg}, takes the form 
\begin{equation*}
    \tilde{\mathcal{E}}(X)=p\sum_{r=0}^{p-1} \mbox{tr}\left(\Pi_r X \Pi_r\right)\Pi_{r+1}\sigma\Pi_{r+1}.
\end{equation*}
We use the notation established in Section \ref{subsec-erg}.

Let $E\in\mathcal{B}(\mathcal{H}_M\otimes\mathcal{H}_M)$ and $\tilde{E}\in\mathcal{B}(\mathcal{H}_M\otimes\mathcal{H}_M)$ be the operators representing $\mathcal{E}$ and $\tilde{\mathcal{E}}$ as described in Appendix \ref{app-vec}. The maps adjoint to $\mathcal{E}$ and $\tilde{\mathcal{E}}$, with respect to Hilbert-Schmidt inner product, are represented by the operators $E^\dagger$ and $\tilde{E}^\dagger$, respectively. From Proposition 3.3 of Ref. \onlinecite{fcs} we know that the peripheral spectrum of $E^\dagger$ consists of $p$ non-degenerate eigenavalues $\exp(\frac{2\pi i}{p}r)$, where $r=0,\cdot\cdot\cdot,p-1$, which correspond to the eigenvectors $u_r$, which are unitary operators. Each of $u_r$ can be decomposed as $u_r=\sum_{k=0}^{p-1}\exp(\frac{2\pi i}{p}rk)\Pi_k$, where $\{\Pi_k\}_{k=0,\cdot\cdot\cdot,p-1}$ is the set of projectors defined in Proposition \ref{ergpgfc}, $\sum_{k=0}^{p-1}\Pi_k=\mathbb{1}_{M}$.

For brevity of notation in this appendix we denote the vectors $\mathrm{vec}(X)\in\mathcal{H}_M\otimes\mathcal{H}_M$, corresponding to the operators $X\in\mathcal{B}(\mathcal{H}_M)$ through the isomorphism $\mathrm{vec}$ defined in Appendix \ref{app-vec}, as $|X\rangle:=\mathrm{vec}(X)$. Then, in this notation,
\begin{equation*}
    |u_r\rangle=\sum_{k=0}^{p-1}\exp\left(\frac{2\pi i}{p}rk\right)|\Pi_k\rangle.
\end{equation*}
Since the matrix with the elements $\{\frac{1}{\sqrt{p}}\exp(\frac{2\pi i}{p}rk)\}_{rk}$ is unitary, we can invert the above relation,
\begin{equation*}
    |\Pi_k\rangle=\frac{1}{p}\sum_{r=0}^{p-1}\exp\left(-\frac{2\pi i}{p}rk\right)|u_r\rangle.
\end{equation*}
Notice that since $E^\dagger|u_r\rangle=\exp(\frac{2\pi i}{p}r)|u_r\rangle$, it follows that 
\begin{equation*}
    E^\dagger|\Pi_k\rangle=|\Pi_{k-1}\rangle,
\end{equation*}
as was obtained in the proof of Proposition 3.3 of Ref. \onlinecite{fcs}.

As $\tilde{E}$ constructed using Definition \ref{e-tilde-e} is defined by the peripheral spectrum of $E$, it follows that $\tilde{E}^\dagger|u_r\rangle=\exp(\frac{2\pi i}{p}r)|u_r\rangle$, $\tilde{E}^\dagger|\Pi_k\rangle=|\Pi_{k-1}\rangle$ and $\tilde{E}^\dagger=\sum_{r=0}^{p-1}\exp\left(\frac{2\pi i}{p}r\right)P_r$,
where $P_r$ is the projection onto the subspace spanned by $|u_r\rangle$. We can represent $P_r=|u_r\rangle\langle v_r|$, where $| v_r\rangle$ satisfy $\langle v_{r'}|u_r\rangle=\delta_{rr'}$. It immediately follows that 
\begin{equation*}
    \tilde{E}=\sum_{r=0}^{p-1}\exp\left(-\frac{2\pi i}{p}r\right)|v_r\rangle\langle u_r|,
\end{equation*}
and that $|v_r\rangle$ are eigenvectors of $E$ corresponding to the eigenvalues $\exp(-\frac{2\pi i}{p}rk)$. We introduce the vectors $|\omega_k\rangle$, defined by
\begin{equation*}
    |\omega_k\rangle:=\frac{1}{p}\sum_{r=0}^{p-1}\exp\left(-\frac{2\pi i}{p}rk\right)|v_r\rangle.
\end{equation*}
It is easy to see that $\tilde{E}|\omega_k\rangle=|\omega_{k+1}\rangle$ and $\langle \omega_{k'}|\Pi_k\rangle=\delta_{kk'}$. Moreover, we can express $\tilde{E}$ in terms of $|\omega_k\rangle$ and $|\Pi_k\rangle$,
\begin{align*}
    \tilde{E}&=\sum_{r=0}^{p-1}\exp\left(-\frac{2\pi i}{p}r\right)|v_r\rangle\langle u_r| \\ \nonumber
    &=\sum_{k,k',r=0}^{p-1}\exp\left(\frac{2\pi i}{p}r(k-k'-1)\right)|\omega_k\rangle\langle \Pi_{k'}|\\ \nonumber
    &=\sum_{k=0}^{p-1}|\omega_k\rangle\langle \Pi_{k-1}|.
\end{align*}

Since $\tilde{E}\sum_{k=0}^{p-1}|\omega_k\rangle=\sum_{k=0}^{p-1}|\omega_{k+1}\rangle=\sum_{k=0}^{p-1}|\omega_k\rangle$, then $\sum_{k=0}^{p-1}|\omega_k\rangle$ is the fixed point of $\tilde{E}$. In the case of the ergodic state the only fixed point of $\mathcal{E}$ is $\sigma=\sum_{k=0}^{p-1}\Pi_k\sigma\Pi_k$. Then the only fixed point of $E$, and hence of $\tilde{E}$ as well, is $|\sigma\rangle=\sum_{k=0}^{p-1}|\Pi_k\sigma\Pi_k\rangle$ with $|\Pi_k\sigma\Pi_k\rangle:=\mathrm{vec}(\Pi_k\sigma\Pi_k)$. Then $|\sigma\rangle=|\omega\rangle$. Moreover since $\langle\Pi_{k'}|\Pi_k\sigma\Pi_k\rangle=\delta_{kk'}$, then $|\omega_k\rangle=|\Pi_k\sigma\Pi_k\rangle$. Then $\tilde{E}=p\sum_{k=0}^{p-1}|\Pi_{k+1}\sigma\Pi_{k+1}\rangle\langle \Pi_{k}|$, which corresponds to the quantum channel
\begin{equation*}
    \tilde{\mathcal{E}}(X)=p\sum_{k=0}^{p-1}\mbox{tr}(\Pi_k X\Pi_k)\Pi_{k+1}\sigma\Pi_{k+1},
\end{equation*}
as required.

\section{}\label{app-alt-rec}
In this appendix we extend the construction of the recovery quantum channel defined in equations (\ref{gram-m})-(\ref{alt-rec-error}) to the general case of pgFCS.

We use the intuition built in Section \ref{sec-proof}. For a general pgFCS $\rho_{ABC}$ (\ref{pgfcs-form}) (which can be decomposed into a convex sum (\ref{conv-erg})) we build the approximating state $\tilde{\rho}_{ABC}$ (\ref{pgfcs-tilde-sec}) induced by 
\begin{equation}\label{isom-gen-form}
    \tilde{V}_B=\sum_{\alpha=1}^J\sum_{r_\alpha=0}^{p_\alpha-1}\sum_{(i_\alpha,j_\alpha)\in\mathcal{O}_{r_\alpha}}\sqrt{\lambda_\alpha p_\alpha\sigma_{\alpha,i}}|\xi^\alpha_{ij}\rangle\otimes|i_\alpha\rangle\langle j_\alpha|,
\end{equation}
where the index $\alpha$ distinguishes the components in the convex sum decomposition (\ref{conv-erg}). The terms corresponding to different $\alpha$ are the isometries inducing ergodic pgFCS and possessing the form (\ref{tilde-isom-erg-form}). The sets $\mathcal{O}_{r_\alpha}$ are defined in (\ref{tilde-isom-r-form}). From the Sections \ref{subsec-pgfcs}, \ref{subsec-erg}, and \ref{sect-fcs} we know that $\tilde{V}_B$ of the form (\ref{isom-gen-form}) (with any set of orthonormal vectors $|\xi^\alpha_{ij}\rangle$) guarantees that $\tilde{\rho}_{ABC}$ is a quantum Markov chain.

At this point it is convenient to restrict ourselves to the case of $J=1$, since the further arguments can be extended to the general case in an obvious manner. Thus in the sequel we deal with an ergodic pgFCS with period $p\geq 1$, for which 
\begin{equation*}
    \tilde{V}_B=\sum_{r=0}^{p-1}\sqrt{p}\sum_{(i,j)\in\mathcal{O}_r}\sqrt{\sigma_i}|\xi_{ij}\rangle\otimes|i\rangle\langle j|
\end{equation*}

In contrast to the main text, we here present an explicit procedure to choose $|\xi_{ij}\rangle$, which leads to a recovery error exponentially small in the size of the region $B$. These will be derived with the use of the $d_M^2\times d_M^2$ matrix $Q$. 

We define $Q$ together with the auxiliary matrix $\tilde{Q}$ of the same dimension by their matrix elements,
\begin{align}\label{m-matr-def}
    Q_{i'j';ij}=\langle j'|V_B^\dagger|i' \rangle\langle i|V_B|j\rangle, \\ \nonumber
    \tilde{Q}_{i'j';ij}=\langle j'|\tilde{V}_B^\dagger|i' \rangle\langle i|\tilde{V}_B|j\rangle,
\end{align}
where $i,j=1,\cdot\cdot\cdot,d_M$. If the region $B$ is large, these matrices are close to each other in matrix norm. Notice that as Gramian matrices, both $Q$ and $\tilde{Q}$ are positive semidefinite. We can explicitly express $\tilde{Q}$, 
\begin{equation}\label{tilde-m-expl}
    \tilde{Q}_{i'j';ij}=p\sigma_i\delta_{ii'}\delta_{jj'}\chi_{\mathcal{O}_r}\left((i,j)\right),
\end{equation}
where $\chi_{\mathcal{O}_r}$ is a characteristic function of the set $\mathcal{O}_r$ defined in (\ref{tilde-isom-r-form}). We observe that $\tilde{Q}$ is diagonal and positive semidefinite, but not full-rank if $p>1$ (since $\chi_{\mathcal{O}_r}((i,j))=0$ for some pairs $(i,j)$). Meanwhile, $Q$ differs from $\tilde{Q}$ by a small correction, and in general can have different rank. We want to construct the number of vectors $|\xi_{ij}\rangle$ equal to the cardinality of the set $\mathcal{O}_r$, which is equal to the rank of $\tilde{Q}$, hence we want to force the ranks of $Q$ and $\tilde{Q}$ to be equal. But first let us rigorously determine how close $Q$ and $\tilde{Q}$ are.

Matrices $Q$ and $\tilde{Q}$ are partial transposes of $E$ and $\tilde{E}$, respectively, defined in Appendix \ref{app-b}. The latter are related to the quantum channels $\mathcal{E}$ and $\tilde{\mathcal{E}}$, respectively, defined in Section \ref{sec-proof}. Then from Ref. \onlinecite{tomiyama},
\begin{equation}\label{m-closeness}
    \|Q-\tilde{Q}\|\leq d_M\|E-\tilde{E}\| \leq cd_M\nu^{|B|},
\end{equation}
where we have used Lemma \ref{b-2} in the second inequality. Now we can represent $Q$ in the form
\begin{equation}\label{m-repr}
    Q=\tilde{Q}+cd_M\nu^{|B|}Z,
\end{equation}
where $Z$ is a Hermitian matrix with $\|Z\|\leq 1$.

From (\ref{tilde-m-expl}) it is clear that the smallest non-zero eigenvalue of $\tilde{Q}$ is equal to $\sigma_{\mathrm{min}}$, the smallest eigenvalue of $\sigma$. The matrices $Q$ and $\tilde{Q}$ are close if the region $B$ is large enough. By dropping from $Q$ the part of the spectrum below $\sigma_{\mathrm{min}}$, we obtain from $Q$ the matrix $Q'$, which approximates to both $Q$ and $\tilde{Q}$ and has the same rank as the latter. We first make the spectral decomposition, $Q=W\Lambda W^\dagger$. We now define $\Pi_{>\epsilon}$, the projector onto the part of the spectrum with eigenvalues greater than $\epsilon=cd_M\nu^{|B|}$, and further define $Q':=W\Lambda_{>\epsilon}W^\dagger$, where $\Lambda_{>\epsilon}=\Pi_{>\epsilon}\Lambda\Pi_{>\epsilon}$. As follows from (\ref{m-closeness}), if $\sigma_{\mathrm{min}}>2cd_M\nu^{|B|}$, then as $|B|\rightarrow\infty$ there is a part of the spectrum of $Q$ with eigenvalues converging to $0$ that is well separated from the rest of the spectrum. We assume that this condition is satisfied. For $Q'$ we have the decomposition analogous to (\ref{m-repr}), $Q'=\tilde{Q}+c\nu^{|B|}Z'$, where $Z'$ is Hermitian and $\|Z'\|\leq 1$.

Since $Q'$ is positive semidefinite, $Q'^{\frac{1}{2}}$ is well-defined. We take the pseudo-inverse of $Q'^{\frac{1}{2}}$, and define $|\xi_{ij}\rangle$,
\begin{equation}\label{xi-def}
    |\xi_{ij}\rangle:=\sum_{m,n=1}^{d_M} (Q'^{-\frac{1}{2}})_{mn;ij}\langle m|V_B|n\rangle. 
\end{equation}

Now we determine the bound on $\|V_B-\tilde{V}_B\|$. By definition of the operator norm,
\begin{widetext}
\begin{align}\label{op-norm-alt}
    \|V_B-\tilde{V}_B\|^2 &=\|\sum_{i,j=1}^{d_M}\left(\langle i|V_B|j\rangle-\sqrt{p\sigma_i}|\xi_{ij}\rangle\right)\otimes|i\rangle\langle j|\|^2=\|\sum_{i,j=1}^{d_M}\left(\langle i|V_B|j\rangle-\sqrt{p\sigma_i}\sum_{m,n=1}^{d_M}(Q'^{-\frac{1}{2}})_{mn;ij}\langle m|V_B|n\rangle\right)\otimes|i\rangle\langle j|\|^2\\ \nonumber
    &=\max_{\sum_j|\eta_j|^2=1}\sum_{j,j'=1}^{d_M}\bar{\eta}_{j'}\sum_{i=1}^{d_M}\left(Q_{ij';ij}+\tilde{Q}_{ij';ij}-2\sqrt{p\sigma_{i}}Q'^{\frac{1}{2}}_{ij';ij}\right)\eta_j=\|\mathrm{tr}_1Q+\mathrm{tr}_1\tilde{Q}-2\mathrm{tr}_1((\sqrt{\sigma}\otimes\mathbb{1})Q'^{\frac{1}{2}})\|,
\end{align}
\end{widetext}
where we have used the definition of $Q$ (\ref{m-matr-def}), the identity $p\sigma_{i}\langle\xi_{ij'}|\xi_{ij}\rangle=\tilde{Q}_{ij';ij}$, which follows from (\ref{m-matr-def}) and (\ref{xi-def}), and the identity $\langle j'|V^\dagger_B|i\rangle\langle m|V_B|n\rangle(Q'^{-\frac{1}{2}})_{mn;ij}=Q_{ij';mn}(Q'^{-\frac{1}{2}})_{mn;ij}=Q'^{\frac{1}{2}}_{ij';ij}$. In the last line $\mathrm{tr}_1$ denotes the partial trace over the first pair of indices.

Notice that 
\begin{align*}
   \sum_{i=1}^{d_M}\tilde{Q}_{ij';ij}&=\sum_{i=1}^{d_M}p\sigma_i\delta_{jj'}\chi_{\mathcal{O}_r}\left((i,j)\right)\\ \nonumber
    &=p\left(\sum_{i=1}^{d_M}\chi_{\mathcal{O}_r}\left((i,j)\right)\sigma_i\right)\delta_{jj'}\\ \nonumber
    &=p\mathrm{tr}(\Pi_{r+|B|}\sigma\Pi_{r+|B|}))\delta_{jj'}\\ \nonumber
    &=p\frac{1}{p}\delta_{jj'}=\delta_{jj'}.
\end{align*}
Thus, 
\begin{equation}\label{tr1-tilde-m}
    \mathrm{tr}_1(\tilde{Q})=\mathbb{1}.
\end{equation}

Now we need to bound $\mathrm{tr}_1 Q$ and $\mathrm{tr}_1((\sqrt{\sigma}\otimes \mathbb{1})Q'^{\frac{1}{2}})$. Using the representation (\ref{m-repr}),
\begin{equation*}
    \mathrm{tr}_1Q=\mathrm{tr}_1\tilde{Q}+cd_M\nu^{|B|}\mathrm{tr}_1Z=\mathbb{1}+cd_M\nu^{|B|}\mathrm{tr}_1Z,
\end{equation*}
where we have used (\ref{tr1-tilde-m}). Since $\|Z\|\leq 1$, then $\|\mathrm{tr}_1Z\|\leq d_M$, and then
\begin{equation}\label{tr1-m-bound}
    (1-cd_M^2\nu^{|B|})\mathbb{1}\leq\mathrm{tr}_1Q\leq(1+cd_M^2\nu^{|B|})\mathbb{1}.
\end{equation}

To bound $\mathrm{tr}_1((\sqrt{\sigma}\otimes \mathbb{1})Q'^{\frac{1}{2}})$, recall that $Q'=\tilde{Q}+cd_M\nu^{|B|}Z'$, hence
\begin{align*}
    Q'^{\frac{1}{2}}&=\sqrt{\tilde{Q}+cd_M\nu^{|B|}Z'}&\\ \nonumber
    &=\tilde{Q}^{\frac{1}{4}}\sqrt{\mathbb{1}+c\nu^{|B|}d_M\tilde{Q}^{-\frac{1}{2}}Z'\tilde{Q}^{-\frac{1}{2}}}\tilde{Q}^{\frac{1}{4}},
\end{align*}
where $\tilde{Q}^{-\frac{1}{2}}$ should be regarded as a pseudo-inverse. As $Z'$ is Hermitian, we can decompose $Z'=Z'_+-Z'_-$, where both $Z'_+$ and $Z'_{-}$ are positive semidefinite. Then, by the monotonicity of matrix square root with respect to partial ordering of non-negative operators, 

\begin{align*}
    &\tilde{Q}^{\frac{1}{4}}\sqrt{\mathbb{1}-cd_M\nu^{|B|}\tilde{Q}^{-\frac{1}{2}}Z'_-\tilde{Q}^{-\frac{1}{2}}}\tilde{Q}^{\frac{1}{4}}\leq Q'^{\frac{1}{2}}\\ \nonumber &\quad\quad\quad \leq\tilde{Q}^{\frac{1}{4}}\sqrt{\mathbb{1}+cd_M\nu^{|B|}\tilde{Q}^{-\frac{1}{2}}Z'_+\tilde{Q}^{-\frac{1}{2}}}\tilde{Q}^{\frac{1}{4}}.
\end{align*}
Using Taylor expansion we can estimate
\begin{widetext}
\begin{align*}
    \sqrt{\mathbb{1}-cd_M\nu^{|B|}\tilde{Q}^{-\frac{1}{2}}Z'_-\tilde{Q}^{-\frac{1}{2}}}&\geq \mathbb{{1}}-\frac{1+2^{-\frac{9}{2}}}{2}cd_M\nu^{|B|}\tilde{Q}^{-\frac{1}{2}}Z'_-\tilde{Q}^{-\frac{1}{2}},\\ \nonumber
    \sqrt{\mathbb{1}+cd_M\nu^{|B|}\tilde{Q}^{-\frac{1}{2}}Z'_+\tilde{Q}^{-\frac{1}{2}}}&\leq\mathbb{1}+\frac{1}{2}cd_M\nu^{|B|}\tilde{Q}^{-\frac{1}{2}}Z'_+\tilde{Q}^{-\frac{1}{2}}.
\end{align*}
\end{widetext}
The former inequality is based on the previously assumed condition $2cd_M\nu^{|B|}\leq\sigma_{\mathrm{min}}$, which leads to $\|cd_M\nu^{|B|}\tilde{Q}^{-\frac{1}{2}}Z'_-\tilde{Q}^{-\frac{1}{2}}\|\leq \frac{cd_M\nu^{|B|}}{\sigma_{\mathrm{min}}}\leq\frac{1}{2}$. We have used the latter bound to estimate the remainder term in the Taylor expansion. We have also used $\|Z'_-\|\leq\|Z\|\leq 1$. Thus,
\begin{align*}
    &\tilde{Q}^{\frac{1}{2}}-\frac{1+2^{-\frac{9}{2}}}{2}cd_M\nu^{|B|}\tilde{Q}^{-\frac{1}{4}}Z'_-\tilde{Q}^{-\frac{1}{4}}\leq Q'^{\frac{1}{2}}\\ \nonumber
    &\quad\quad\quad\leq\tilde{Q}^{\frac{1}{2}}+\frac{1}{2}cd_M\nu^{|B|}\tilde{Q}^{-\frac{1}{4}}Z'_+\tilde{Q}^{-\frac{1}{4}},
\end{align*}
which leads to 
\begin{align*}
    \mathbb{1}-\frac{1+2^{-\frac{9}{2}}}{2}cd_M\nu^{|B|}\mathrm{tr}_1Z'_-&\leq \mathrm{tr}_1((\sqrt{\sigma}\otimes\mathbb{1})Q'^{\frac{1}{2}})\\ \nonumber &\leq \mathbb{1}+\frac{1}{2}cd_M\nu^{|B|}\mathrm{tr}_1Z'_+,
\end{align*}
where we have used $\mathrm{tr}_1((\sqrt{\sigma}\otimes\mathbb{1})\tilde{Q}^{\frac{1}{2}})=\mathrm{tr}_1\tilde{Q}=\mathbb{1}$, and $\mathrm{tr}_1((\sqrt{\sigma}\otimes\mathbb{1})\tilde{Q}^{-\frac{1}{4}}Z'_{\pm}\tilde{Q}^{-\frac{1}{4}}))=\mathrm{tr}_1(\tilde{Q}^{-\frac{1}{4}}\tilde{Q}^{\frac{1}{2}}\tilde{Q}^{-\frac{1}{4}}Z'_{\pm})=\mathrm{tr}_1Z'_{\pm}$. Since $\|\mathrm{tr}_1Z'_\pm\|\leq \|\mathrm{tr}_1(\mathbb{1}\otimes\mathbb{1})\|\leq d_M$, then
\begin{align}\label{tr1-cross-term}
    \left(1-\frac{1+2^{-\frac{9}{2}}}{2}cd_M^2\nu^{|B|}\right)\mathbb{1}&\leq \mathrm{tr}_1((\sqrt{\sigma}\otimes\mathbb{1})Q'^{\frac{1}{2}})\\ \nonumber &\leq \left(1+\frac{1}{2}cd_M^2\nu^{|B|}\right)\mathbb{1}.
\end{align}
Combining (\ref{tr1-tilde-m}), (\ref{tr1-m-bound}), and (\ref{tr1-cross-term}) together,
\begin{equation*}
    \|\mathrm{tr}_1Q+\mathrm{tr}_1\tilde{Q}-2\mathrm{tr}_1((\sqrt{\sigma}\otimes\mathbb{1})Q'^{\frac{1}{2}})\|\leq 2(1+2^{-\frac{11}{2}})cd_M^2\nu^{|B|}.
\end{equation*}
From (\ref{op-norm-alt}) we obtain 
\begin{equation}\label{alt-op-norm-bound}
    \|V_B-\tilde{V}_B\|\leq \sqrt{2(1+2^{-\frac{11}{2}})c}d_M\nu^{\frac{|B|}{2}}.
\end{equation}
This bound is only a factor $\sqrt{2(1+2^{-\frac{11}{2}})d_M}$ worse than the bound (\ref{isonorm-leq-exp}), though it requires the condition $\sigma_{\mathrm{min}}>2cd_M\nu^{|B|}$ to be satisfied. 

Analogous to the bound (\ref{isonorm-leq-exp}) in Section \ref{sec-proof}, (\ref{alt-op-norm-bound}) remains true for a general pgFCS. In the case $J>1$ of the convex decomposition (\ref{conv-erg}), the matrices $Q$, $Q'$ and $\tilde{Q}$ decompose into direct sums over ergodic components, i.e., $Q=\bigoplus_{\alpha=1}^{J}Q_\alpha$. The generalization to this case is straightforward. For the ergodic pgFCS of period $1$, $\tilde{Q}$ is full-rank, hence we have the simplification $Q'=Q$, where $Q$ is guaranteed to be invertible provided $\sigma_{\mathrm{min}}>cd_M\nu^{|B|}$. We have discussed this case in Section \ref{sec-rec-maps}.

\section{}\label{app-comparison}
In this appendix we make some remarks comparing our Theorem \ref{main-result} to Theorems III.1, III.2 (a generalization of III.1), and Proposition III.3 of Ref \onlinecite{fcs-mpdo}. In particular, we show that the classes of states considered in Ref. \onlinecite{fcs-mpdo} are in general distinct from pgFCS. 

In our notation, the states considered in the mentioned theorems and proposition of Ref. \onlinecite{fcs-mpdo} have the form 
\begin{equation}\label{brandao-state-form}
     \rho_{ABM}=\Phi^{|B|}\circ\Phi^{|A|}(\sigma),
\end{equation}
where the compatibility condition $\mathrm{tr}_s\Phi(\sigma)=\sigma$ is not assumed. Note that what is called system $C$ in Ref. \cite{fcs-mpdo}, is the memory system $M$ in our article, hence the indexing in (\ref{brandao-state-form}). The states $\rho_{ABM}$ of (\ref{brandao-state-form}) are related to the states $\rho_{ABC}$ of (\ref{fcs-r}) and (\ref{pgfcs-form}) in the current paper through
\begin{equation}\label{c-from-m}
    \rho_{ABC}=\mathrm{tr}_M\circ\Phi^{|C|}(\rho_{ABM}). 
\end{equation} 

In Theorems III.1 and III.2 of Ref. \onlinecite{fcs-mpdo}, the considered quantum channels $\Phi$ satisfy the condition (in our notation) 
\begin{equation}\label{brandao-cond}
    \Phi(\zeta_M)=\chi_s\otimes\zeta_M,
\end{equation}
for some states $\chi\in\mathcal{D}(\mathcal{H}_s)$ and $\zeta\in\mathcal{D}(\mathcal{H}_M)$, which are maximally mixed states in the setup of Theorem III.1. It is proven that
\begin{align}
    \|\rho_{ABM}-\rho_{AB}\otimes\zeta_M\|_1 &=O(\eta^{|B|}) \label{brandao-mem-aqmc}\\ 
    I(A:M|B)&=O(|B|\eta^{|B|}), \label{brandao-mem-qcmi-1}
\end{align}
where $\eta<1$ under some additional assumptions. We note that, while it is not stated in Ref. \onlinecite{fcs-mpdo}, (\ref{brandao-mem-aqmc}) implies that $\rho_{ABM}$ is approximated in trace norm by a manifestly QMC $\rho_{AB}\otimes\zeta_M$. Since action of a quantum channel on the system $M$ increases neither trace norm, nor QCMI, the state $\rho_{ABC}$ in (\ref{c-from-m}) inherits the properties (\ref{brandao-mem-aqmc}) and (\ref{brandao-mem-qcmi-1}),
\begin{align}
    \|\rho_{ABC}-\tilde{\rho}_{ABC}\|_1 &=O(\eta^{|B|})\\ 
    I(A:C|B)&=O(|B|\eta^{|B|}),
\end{align}
where $\tilde{\rho}_{ABC}=\rho_{AB}\otimes\mathrm{tr}_M\circ\Phi^{|C|}(\zeta_{M})$ is a QMC. In these bounds our result is similar, however the pgFCS that we consider are generally an independent class of states. Indeed, for pgFCS, the generating channel $\Phi$ has the form $\Phi(X)=VXV^\dagger$, where $V:\mathcal{H}_M\rightarrow \mathcal{H}_s\otimes\mathcal{H}_M$ is an isometry, consistent with the condition (\ref{brandao-cond}) only under the restrictive choice $\chi=|\phi\rangle\langle\phi|$ and $V=|\phi\rangle\otimes\mathbb{1}_M$ for some normalized vector $|\phi\rangle\in\mathcal{H}_s$. Hence a pgFCS in general does not belong to the class of states considered in Theorem III.2 (the converse is also true), and never belongs to the class considered in Theorem III.1.

The quantum channels considered in Proposition III.3 of Ref. \onlinecite{fcs-mpdo} have a forgetful component, i.e., can be decomposed as
\begin{equation}\label{brand-forget-comp}
    \Phi(X)=(1-\eta)\mathrm{tr}(X)\chi+\eta\mathcal{N}(X),
\end{equation}
where $\chi\in\mathcal{D}(\mathcal{H}_s\otimes\mathcal{H}_M)$, $0<\eta<1$ and $\mathcal{N}:\mathcal{B}(\mathcal{H}_M)\rightarrow\mathcal{B}(\mathcal{H}_s)\otimes\mathcal{B}(\mathcal{H}_M)$ is some quantum channel. This class of states is also distinct from pgFCS, since no channel of the form $\Phi(X)=VXV^\dagger$ can be decomposed as (\ref{brand-forget-comp}). It is intuitively obvious that a channel constructed as a conjugation by isometries does not have a forgetful component. As proof, consider two orthogonal vectors $|\psi\rangle,|\phi\rangle\in \mathcal{H}_M$, $\langle\psi|\phi\rangle=0$. Then, on the one hand, 
\begin{align}
    &\mathrm{tr}\left(\Phi(|\psi\rangle\langle\psi|)\Phi(|\phi\rangle\langle\phi|)\right) \\ \nonumber
    &\qquad\qquad=\mathrm{tr}\left(V|\psi\rangle\langle\psi|V^\dagger V|\phi\rangle\langle\phi|V^\dagger\right)\\ \nonumber
    &\qquad\qquad=\left|\langle\psi|\phi\rangle\right|^2\\ \nonumber
    &\qquad\qquad=0,
\end{align}
and on the other hand, for quantum channels admitting the decomposition (\ref{brand-forget-comp}),
\begin{align}
    &\mathrm{tr}\left(\Phi(|\psi\rangle\langle\psi|)\Phi(|\phi\rangle\langle\phi|)\right) \\ \nonumber
    &=\mathrm{tr}\left(\left((1-\eta)\chi+\eta\mathcal{N}(|\psi\rangle\langle\psi|)\right)\right.\\ \nonumber
    &\qquad\qquad\times\left.\left((1-\eta)\chi+\eta\mathcal{N}(|\phi\rangle\langle\phi|)\right)\right)\\ \nonumber
    &=(1-\eta)^2\mathrm{tr}(\chi^2)+\eta(1-\eta)\mathrm{tr}\left(\chi\mathcal{N}(|\psi\rangle\langle\psi|+|\phi\rangle\langle\phi|)\right)\\ \nonumber
    &\qquad\qquad+\eta^2\mathrm{tr}\left({N}(|\psi\rangle\langle\psi|){N}(|\phi\rangle\langle\phi|)\right)\\ \nonumber
    &>0.
\end{align} 
since all the summands are non-negative, and $(1-\eta)^2\mathrm{tr}(\chi^2)>0$. Hence, the intersection between the class of pgFCS and the states generated by a partially forgetful channel (\ref{brand-forget-comp}) is empty.

\bibliography{bibliography}

\end{document}